\documentclass[sigconf]{acmart}
\renewcommand\footnotetextcopyrightpermission[1]{} 

\usepackage{amsmath,amsthm, bm}
\usepackage{graphicx,xcolor}
\usepackage[mathscr]{euscript}
\usepackage{tikz}
\usetikzlibrary{trees,arrows,automata}
\usetikzlibrary{shapes.geometric} 
\usetikzlibrary{shadows,fadings}
\usepackage{wrapfig}   
\usepackage{stmaryrd} 
\usepackage{url}
\usepackage{pifont}
\usepackage{lipsum}
\usetikzlibrary{arrows,shapes,positioning,shadows,trees}
\usepackage{multirow}
\usepackage{mathtools}
\usepackage{tcolorbox}
\tcbuselibrary{skins}
\usepackage{environ}
\usepackage{colortbl}
\usepackage{hhline}
\usepackage{rotating}
\usepackage{array}
\usepackage{libertine}
\usepackage{enumitem}
\usepackage{mdframed}
\usepackage{cancel}
\usepackage{soul}
\usepackage{fontawesome}
\usepackage{physics}
\usepackage{xspace}
\usepackage{hhline}

\usepackage{etoolbox}

\newbool{draft_formatting}
\setbool{draft_formatting}{false} 

\newbool{arxiv_version}
\setbool{arxiv_version}{false} 

\ifbool{draft_formatting}{
    \usepackage{xwatermark}
    \newwatermark[%
      allpages,
      angle=55,
      scale=2.2,
      textcolor=gray!15
    ]{WORK IN PROGRESS\\PLEASE DO NOT SHARE}
}{}

\graphicspath{{SPAA-2021-Camera-Ready/images/}}

\newcommand{\Periodic}{\textsc{StencilFFT-P}}
\newcommand{\Aperiodic}{\textsc{StencilFFT-A}}
\newcommand{\Boundary}{\textsc{RecursiveBoundary}}

\newcommand{\StencilViaFFTPeriodic}{\textsc{PStencil-1D-FFT}\xspace}

\newcommand{\StencilviaNestedLoopPeriodic}{\textsc{PStencil-Loop}}
\newcommand{\StencilviaNestedLoopAperiodic}{\textsc{AStencil-Loop}}
\newcommand{\MFFT}{\textsc{Multi-FFT}}
\newcommand{\IMFFT}{\textsc{Inverse-Multi-FFT}}

\newcommand{\fft}{\mathcal{F}}
\newcommand{\ifft}{\mathcal{F}^{-1}}
\newcommand{\pluto}{PLuTo}
\newcommand{\plutostandard}{\texttt{standard}}
\newcommand{\plutodiamond}{\texttt{diamond}}

\newcommand{\hide}[1]{}
\newcommand{\Arxiv}[1]{
    \ifbool{arxiv_version}{#1}{}
}

\newcommand{\highlight}[1]{\textit{\textbf{#1}}}

\newcommand{\cnote}[2]{%
  \ifnotes
    {\marginpar{\color{#1}\scriptsize \sf \raggedright{#2}}}%
  \fi}

\newcommand{\citereferences}{
    \ifbool{draft_formatting}{}{\textsc{\textcolor{red}{CITE}}}
}
\newcommand{\autogen}{\ifbool{draft_formatting}{}{\textsc{Autogen}}}
\newcommand{\okay}{
    \ifbool{draft_formatting}{}{\pramod{Okay till here}}
}
\newcommand{\addcitations}{
    \ifbool{draft_formatting}{}{\textcolor{red}{[CITE]}}
}
\newcommand{\pramod}[1]{
    \ifbool{draft_formatting}{}{
        \vspace{0.3cm}
        \hrule{} \textcolor{red}{[\textsc{#1}]} \hrule{} \vspace{0.3cm}
    }
}

\newcommand{\rezaul}[1]{
    \ifbool{draft_formatting}{}{\textcolor{orange}{[\textsc{#1}]}}
}

\newcommand{\aaron}[1]{
    \ifbool{draft_formatting}{}{\textcolor{blue}{[\textsc{#1}]}}
}

\definecolor{orange}{rgb}{1,0.3,0}
\newcommand{\rezaulnote}[1]{
    \ifbool{draft_formatting}{}{\cnote{red}{Rezaul: #1}}
}

\newcommand{\vgap}{\vspace{-0.2cm}}

\newcommand{\para}[1]{\vspace{0.1cm}\noindent{\bf{#1.}}}

\newcommand{\codesize}{\small}

\newcommand{\T}{\hspace{1em}}

\definecolor{gray}{rgb}{0.3,0.3,0.3}

\newcommand{\Oh}[1]{{\mathcal O}\left({#1}\right)}
\newcommand{\LOh}[1]{{\mathcal O}\left({#1}\right.}
\newcommand{\ROh}[1]{\left.{#1}\right)}
\newcommand{\oh}[1]{{o}\left({#1}\right)}
\newcommand{\Om}[1]{{\Omega}\left({#1}\right)}

\newcommand{\Th}[1]{{\Theta}\left({#1}\right)}
\newcommand{\LTh}[1]{{\Theta}\left({#1}\right.}
\newcommand{\RTh}[1]{\left.{#1}\right)}

\newcommand{\xif}{{\bf{{if~}}}}
\newcommand{\xthen}{{\bf{{then~}}}}

\newcommand{\xfor}{{\bf{{for~}}}}

\newcommand{\xto}{{\bf{{to~}}}}

\newcommand{\xdo}{{\bf{{do~}}}}

\newcommand{\xreturn}{{\bf{{return~}}}}

\newcommand{\xparallelfor}{{\bf{{parallel for~}}}}

\newcommand{\xcomment}{\hfill $\rhd$ }

\newcommand{\vsitem}{\vspace{0cm}\item}

\makeatletter
\newcommand{\ALOOP}[1]{\ALC@it\algorithmicloop\ #1%
  \begin{ALC@loop}}
\newcommand{\ENDALOOP}{\end{ALC@loop}\ALC@it\algorithmicendloop}

\makeatother

\definecolor{desiredblue}{rgb}{0.803,0.940,0.995}
\definecolor{gray}{rgb}{0.3,0.3,0.3}
\colorlet{lightblue}{desiredblue}
\colorlet{lightred}{red!15}
\colorlet{lightgreen}{green!15}
\colorlet{mediumblue}{blue!25}
\colorlet{mediumred}{red!25}
\colorlet{mediumgreen}{green!25}
\colorlet{lightyellow}{yellow!15}
\colorlet{framecolor}{yellow!20}
\colorlet{frameloopcolor}{yellow!20}

\surroundwithmdframed[
  backgroundcolor   = algocolor,
  hidealllines      = true,
  innerleftmargin   = 0.5em,
  innerrightmargin  = 0.5em,
  innertopmargin    = -1.1em,
  innerbottommargin = 0em,
  leftline=true,
  rightline=true,
  bottomline=true,
  topline=true,
  skipabove         = .5\baselineskip,
  skipbelow         = .5\baselineskip
]{algorithm}

\makeatletter
   \newcommand\figcaption{\def\@captype{figure}\caption}
   \newcommand\tabcaption{\def\@captype{table}\caption}
\makeatother

\colorlet{algotitlebarcolor}{green!20}
\colorlet{algotitlebarcolortop}{lightblue}
\colorlet{algotitlebarcolorbottom}{lightblue}
\colorlet{framecolor}{yellow}
\colorlet{framecolortop}{white}
\colorlet{framecolorbottom}{white}
\colorlet{bluetop}{cyan!40}
\colorlet{bluebottom}{cyan!20}
\colorlet{greentop}{white!40}
\colorlet{greenbottom}{white!20}
\colorlet{shadowcolor}{gray}
\colorlet{algotitlecolor}{black}
\colorlet{algoframecolor}{gray}
\colorlet{tabletitlecolor}{desiredblue}
\colorlet{algocolor}{black}

\newtcolorbox{mycolorbox}[1]
{skin=enhanced,colbacktitle=algotitlebarcolor,coltitle=algotitlecolor,colback=algocolor,colframe=algoframecolor,boxrule=1pt,left=5mm,right=1mm,title style={top color=algotitlebarcolortop, bottom color=algotitlebarcolorbottom},interior style={top color=framecolortop, bottom color=framecolorbottom},title={\codesize #1}}

\newcommand{\algotopspace}{\vspace{-0.25cm}}

\newcommand{\algobottomspace}{\vspace{-0.15cm}}

\arrayrulecolor{gray}
\newsavebox{\tablebox}
\NewEnviron{colortabular}[1]{%
  \addtolength{\extrarowheight}{1ex}%
  \savebox{\tablebox}{%
    \begin{tabular}{#1}%
	  \rowcolor{tabletitlecolor}
      \BODY%
    \end{tabular}}
  \begin{tikzpicture}
    \begin{scope}
      \clip[rounded corners=1ex] (0,-\dp\tablebox) -- (\wd\tablebox,-\dp\tablebox) -- (\wd\tablebox,\ht\tablebox) -- (0,\ht\tablebox) -- cycle;
      \node at (0,-\dp\tablebox) [anchor=south west,inner sep=0pt]{\usebox{\tablebox}};
    \end{scope}
    \draw[rounded corners=1ex] (0,-\dp\tablebox) -- (\wd\tablebox,-\dp\tablebox) -- (\wd\tablebox,\ht\tablebox) -- (0,\ht\tablebox) -- cycle;
  \end{tikzpicture}
}

\setlist[enumerate,1]{leftmargin=\dimexpr 26pt-0.3cm}

\usepackage{subcaption}
\usepackage{booktabs}

\setcopyright{none}

\begin{document}

	\fancyhead{}

\title{Fast Stencil Computations\\ using Fast Fourier Transforms}         


\author{Zafar Ahmad}
\affiliation{%
	\institution{Stony Brook University}
}
\email{zafahmad@cs.stonybrook.edu}

\author{Rezaul Chowdhury}
\affiliation{%
	\institution{Stony Brook University}
}
\email{rezaul@cs.stonybrook.edu}

\author{Rathish Das}
\affiliation{%
	\institution{University of Waterloo}
}
\email{rathish.das@uwaterloo.ca}

\author{Pramod Ganapathi}
\affiliation{%
	\institution{Stony Brook University}
}

\email{pramod.ganapathi@cs.stonybrook.edu}

\author{Aaron Gregory}
\affiliation{%
	\institution{Stony Brook University}
}

\email{afgregory@cs.stonybrook.edu}

\author{Yimin Zhu}
\affiliation{%
	\institution{Stony Brook University}
}
\email{yimzhu@cs.stonybrook.edu }

\begin{abstract}
Stencil computations are widely used to simulate the change of state of physical systems across a multidimensional grid over multiple timesteps. The state-of-the-art techniques in this area fall into three groups: cache-aware tiled looping algorithms, cache-oblivious divide-and-conquer trapezoidal algorithms, and Krylov subspace methods.

In this paper, we present two efficient parallel algorithms for performing linear stencil computations. Current direct solvers in this domain are computationally inefficient, and Krylov methods require manual labor and mathematical training. We solve these problems for linear stencils by using DFT preconditioning on a Krylov method to achieve a direct solver which is both fast and general. Indeed, while all currently available algorithms for solving general linear stencils perform $\Th{NT}$ work, where $N$ is the size of the spatial grid and $T$ is the number of timesteps, our algorithms perform $\oh{NT}$ work.

To the best of our knowledge, we give the first algorithms that use fast Fourier transforms to compute final grid data by evolving the initial data for many timesteps at once. Our algorithms handle both periodic and aperiodic boundary conditions, and achieve polynomially better performance bounds (i.e., computational complexity and parallel runtime) than all other existing solutions.

Initial experimental results show that implementations of our algorithms that evolve grids of roughly $10^7$ cells for around $10^5$ timesteps run orders of magnitude faster than state-of-the-art implementations for periodic stencil problems, and 1.3$\times$ to 8.5$\times$ faster for aperiodic stencil problems.\\

\noindent
{\bf{Code Repository:}} \texttt{https://github.com/TEAlab/FFTStencils}
\end{abstract}


\maketitle

\pagenumbering{arabic}

\section{Introduction}
\label{sec:introduction}

A \highlight{stencil} is a pattern used to compute the value of a cell in a spatial grid at some time step from the values of nearby cells at previous time steps. A \highlight{stencil computation} \cite{frigo2005cache, tang2011pochoir} applies a given stencil to the cells in a spatial grid for some set number of timesteps. Stencil computations have applications in a variety of fields including fluid dynamics \cite{ferziger2002computational, hirsch2007numerical, blazek2015computational}, electromagnetics \cite{komissarov2002time, taflove2005computational, van2012numerical, atangana2015numerical}, mechanical engineering \cite{renson2016numerical, szilard2004theories, rappaz2010numerical}, meteorology \cite{robert1981stable, robert1982semi, avissar1989parameterization, kalnay1990global}, cellular automata \cite{somers1993direct, nkwunonwo2019urban, sitko2016time, mendicino2015stability}, and image processing \cite{weickert2000applications, peyre2011numerical, vese2002numerical, weickert1996theoretical, russ1994image}. In particular, they are widely used for simulating the change of state of physical systems over time \cite{pang1999introduction, barth2013high, vesely1994computational, thijssen2007computational}. 

Due to the importance of stencil computations in scientific computing \cite{press2007numerical, dahlquist2008numerical, golub1992scientific}, various methods have been devised to improve their runtime performance on different machine architectures \cite{nguyen20103, kamil2010auto, sano2011scalable, datta2009optimization, christen2011patus}. All currently available stencil compilers \cite{christen2011patus, hagedorn2018high, henretty2013stencil, luporini2020architecture, ragan2013halide, bondhugula2008pluto, tang2011pochoir} that can accept arbitrary \highlight{linear}\footnote{
A linear stencil is one that uses exclusively linear combinations of grid values from prior timesteps.
} stencils perform $\Th{NT}$ work\footnote{
Let $T_p$ denote a program's runtime on a $p$-processor machine. Then, $T_1$ and $T_{\infty}$ are called \textit{work} and \textit{span}, respectively.\label{fn:workspan}
}, where $N$ is the number of cells in the spatial grid and $T$ is the number of timesteps.

In this paper, we present the first $\oh{NT}$-work stencil computation algorithms that support general linear stencils and arbitrary boundary conditions. Our algorithms have polynomially lower work than all other known options of equivalent or greater generality.

\para{Problem Specification} Consider a stencil computation to be performed over $T$ timesteps on a $d$-dimensional spatial grid of $N$ cells with initial data $a_0[\cdots]$. Cell data at subsequent timesteps are defined via application of the linear stencil $S$ across the grid, formalized as $a_{t+1} = S a_t$, where $a_t$ is the spatial grid data at timestep $t$. The stencil must define the value of a grid cell in terms of a fixed size neighborhood containing cells from only the prior timestep\footnote{
We will later extend the definition of stencils to allow for dependence on multiple prior timesteps.
}. We will not be able to apply the stencil to some cells near the boundaries of the grid; the values of these cells are instead defined via \highlight{boundary conditions}. Our goal is to compute the final grid data $a_T$ by evolving the initial data $a_0$ for $T$ timesteps.

There are two types of boundary conditions we can use: \highlight{periodic} and \highlight{aperiodic}. If the boundary conditions are periodic, it means that every dimension of the spatial grid wraps around onto itself, so the entire grid forms a torus. In this case modular arithmetic is used for all calculations involving spatial indices, and the stencil alone can be used to update all cells. On the other hand, if the boundary conditions are aperiodic, then the cells at the boundary of the grid have to be computed via some method other than straightforward application of the stencil. In this paper, we consider both types of boundary conditions.

\Arxiv{
\para{Example: 1-D Heat Equation with Aperiodic Grid} The temperature $u(x, \tau)$ at location $x$ and time $\tau$ in a homogeneous material is governed by the heat equation $\partial_{\tau} u = \partial_x^2 u$. By discretizing across space and time, we can recover the discrete heat equation,
$$\frac{u(x, \tau+\Delta \tau) - u(x, \tau)}{\Delta \tau} \approx \frac{u(x - \Delta x) - 2 u(x, \tau) + u(x + \Delta x, \tau)}{\Delta x^2}.$$
Further manipulation and substitution of the grid values $u_t[i] = u(i \Delta x, t \Delta \tau)$ give an explicit stencil scheme:
$$u_{t+1}[i] = u_t[i] + (\Delta \tau / \Delta x^2) (u_t[i-1] - 2u_t[i] + u_t[i+1]).$$
This stencil is depicted in Figure \ref{fig:stencil-exampled} (top left). Computing $T = \tau/\Delta \tau$ steps of $u_t$ then gives us an approximation of the exact function $u(x, \tau)$, as shown in Figure \ref{fig:stencil-exampled} (bottom). This approximation becomes increasingly precise as $\Delta x$ and $\Delta \tau$ go to 0.

\begin{figure}[!ht] 
\centering
\begin{minipage}{0.48\textwidth}
\centering
\includegraphics[width=0.48\textwidth]{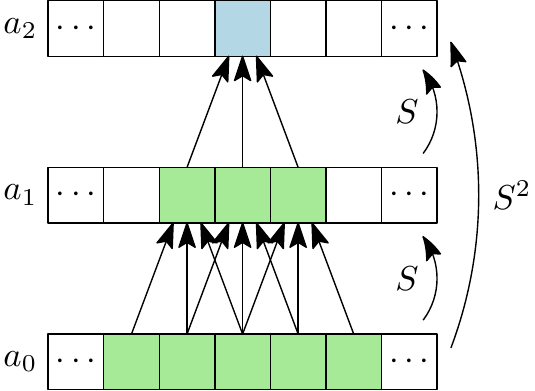}~\\[0.1cm]
\includegraphics[width=0.7\textwidth]{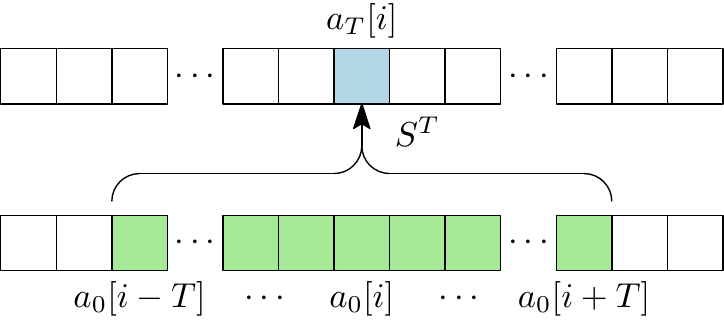}
\vgap{}
\caption{\aaron{This figure should come much later; we're not talking about exponentiation at all yet}\small (Top) Evolving via $S$ in two timesteps is the same as evolving via $S^2$ once. Here $S$ is a three point stencil, using only data from $a_t[i-1]$, $a_t[i]$, and $a_t[i+1]$ to computer $a_{t+1}[i]$. (Bottom) Final values $a_T$ are found by applying $S$ a total of $T$ times to $a_0$, i.e., by applying $S^T$ once to $a_0$.}
\label{fig:stencil-exampled}
\end{minipage}
\vspace{0.2cm}
\end{figure}
}

\para{Existing Algorithms} There are a handful of algorithms commonly used for carrying out general stencil computations, and many more designed for solving problems with specific boundary conditions and stencils. Here we will give an overview of algorithms which can be used for computations with arbitrary boundary conditions and linear stencils. It is worth noting that almost none of the following algorithms require stencil linearity\footnote{
Although some stencil compilers may not be able to apply their low-level optimizations to nonlinear stencils.
}.

\Arxiv{
\begin{figure}[!ht] 
\centering
\begin{minipage}{0.48\textwidth}
\begin{mycolorbox}{\StencilviaNestedLoopPeriodic$(a_0, S, \sigma, N, T)$}
\begin{minipage}{0.97\textwidth}
{\scriptsize
\algotopspace{}
\noindent
\begin{enumerate}
\setlength{\itemindent}{-2em}

\vsitem \xfor $t \gets 1$ \xto $T$ \xdo
\vsitem \T \xparallelfor $i \gets 0$ \xto $N - 1$ \xdo
\vsitem \T \T $a_{t}[i] \gets \sum_{j = -\sigma}^{\sigma} S[i, (i+j)\%N] a_{t-1}[(i + j)\%N]$ 
\vsitem \textbf{return} $a_{T}[0, \ldots, N-1]$ \xcomment final spatial grid data

\algobottomspace{}
\end{enumerate}
}
\end{minipage}
\end{mycolorbox}
\vgap{}
\begin{mycolorbox}{\StencilviaNestedLoopAperiodic$(a_0, S, \sigma, N, T)$}
\begin{minipage}{0.97\textwidth}
{\scriptsize
\algotopspace{}
\noindent
\begin{enumerate}
\setlength{\itemindent}{-2em}

\vsitem \xfor $t \gets 1$ \xto $T$ \xdo
\vsitem \T Apply boundary conditions to compute
\vsitem[] \T \T $a_{t}[0, \ldots, \sigma-1]$ and $a_{t}[(N - \sigma + 1), \ldots, (N - 1)]$
\vsitem \T \xparallelfor $i \gets \sigma$ \xto $N - 1 - \sigma$ \xdo
\vsitem \T \T $a_{t}[i] \gets \sum_{j = -\sigma}^{\sigma} S[i, i+j] a_{t-1}[i+j]$ 
\vsitem \textbf{return} $a_{T}[0, \ldots, N-1]$ \xcomment final spatial grid data

\algobottomspace{}
\end{enumerate}
}
\end{minipage}
\end{mycolorbox}
\vgap{}\vgap{}
\caption{\small Looping implementations of 1-D stencil computations with periodic and aperiodic boundary conditions. Here $\sigma$ is the radius of the neighborhood used by the stencil $S$, and $m \% n = m \bmod n$.}
\label{fig:looping-code}
\end{minipage}
\vgap{}\vgap{}\vgap{}
\end{figure}
}

\vspace{0.1cm}
\noindent
\textbf{1. [\textit{Looping Algorithms.}]} It is a simple matter to implement stencil computations using nested loops, as shown in Figure \ref{fig:looping-code}. However, such computations suffer from poor data locality and hence are inefficient both in theory and in practice. As can be seen from the pseudocode listing, looping codes require $\Th{NT}$ work to evolve a grid of $N$ cells forward $T$ timesteps, assuming that the stencil only uses values from a constant size neighborhood.

\begin{figure}[!t] 
\centering
\begin{minipage}{0.48\textwidth}
\begin{mycolorbox}{\textsc{Stencil-Loop}$(a_0, S, N, T)$}
\begin{minipage}{0.97\textwidth}
{\scriptsize
\algotopspace{}
\noindent
\begin{enumerate}
\setlength{\itemindent}{-2em}

\vsitem \xfor $t \gets 1$ \xto $T$ \xdo
\vsitem \T \xparallelfor $i \gets 0$ \xto $N - 1$ \xdo
\vsitem \T \T compute $a_{t}[i]$ by applying either the stencil $S$
\vsitem[] \T\T\T or the boundary conditions as appropriate
\vsitem \textbf{return} $a_{T}[0, \ldots, N-1]$ \xcomment final spatial grid data

\algobottomspace{}
\end{enumerate}
}
\end{minipage}
\end{mycolorbox}
\vgap{}\vgap{}
\caption{\small Looping code for 1-D stencil computations.}
\label{fig:looping-code}
\end{minipage}
\vgap{}\vgap{}\vgap{}
\end{figure}

\vspace{0.1cm}
\noindent
\textbf{2. [\textit{Tiled Looping Algorithms.}]} Adopting a tiling \cite{Bandishti2012, bondhugula2017} of the spatial grid is a common way to improve the data locality \cite{Wolfe1987,Wolf1991,Wolf1996,Wonnacott2002,bondhugula2016plutoplus} \hide{and parallelism \cite{Andonov1997,Hogstedt1999,bondhugula2017}} of looping algorithms. The tiles' size and shape have a strong influence on the runtime of the algorithm, and generally the best performance is attained when each tile fits into the cache. Most modern multicore machines have a hierarchy of caches; to make better use of the cache hierarchy, loop nests may need to be tiled at multiple levels \cite{kamil2005impact, datta2009optimization}.

The most common framework that can be used to derive tiled looping implementations is the \highlight{polyhedral model}, which uses a set of hyperplanes to partition the grid being solved for. The polyhedral model is extensively used in several code generators \cite{Pluto, Polly, PoCC, yuki2015polybench, verdoolaege2013}.

\vspace{0.1cm}
\noindent
\textbf{3. [\textit{Recursive Divide-and-Conquer Algorithms.}]} Instead of using a predetermined tiling of the spatial grid, these algorithms recursively break the region to be solved for into multiple smaller subregions. The \highlight{trapezoidal decomposition algorithm} \cite{FrigoSt2005, FrigoSt2009, Tang2011} is the most well known divide-and-conquer stencil algorithm. Its recursive approach to tiling allows it to be not only both cache-oblivious \cite{FrigoLePrRa1999} and cache-adaptive \cite{Bender2014, Bender2016}, but also to achieve asymptotic cache performance matching that of an optimally tiled stencil code across all levels of the memory hierarchy in a multicore machine.

\Arxiv{
\aaron{TODO: there should be a heading here} Stencils can be considered as special cases of dynamic programming (DP) recurrences, where the read-write dependencies are local. \highlight{\autogen{}} \addcitations{} and \highlight{Bellmaniac} \addcitations{} are divide-and-conquer algorithms that take as input simple specifications of DP recurrences (or its variants) and output recursive divide-and-conquer DP algorithms, automatically and semi-automatically, respectively. These algorithms apply to a wide class of DP problems, including most stencils. \aaron{This paragraph doesn't add much}

When applied to the 1-D heat equation stencil, all the divide-and-conquer algorithms above perform $\Th{NT}$ work\hide{, $\Th{T N^{\log 3 - 1}}$ span,} and incurs $\Oh{NT/(BM)}$ cache misses. \aaron{already stated multiple times}
}

\vspace{0.1cm}
\noindent
\textbf{4. [\textit{Krylov Subspace Methods.}]} Krylov methods compose a diverse set of mathematical techniques which are extensively used in numerical analysis to find successively better approximations of the exact solution to a stencil problem. Such methods are often used to solve problems for which there is no known direct (tiled-loop or divide-and-conquer) solution technique. Discrete Fourier transforms (DFTs) are frequently used in the analysis \cite{boyd2001chebyshev} and implementations \cite{andreussi2012revised, kabel2014efficient, guan2019two} of these methods. Krylov methods that use DFTs in their implementations are very restricted in their applicability, usually applying only to stencils from specific PDEs that benefit from spectral analysis.

There are several limitations of Krylov methods as a whole: $(i)$ their initial design requires nontrivial manual convergence analysis \cite{kuijlaars2006convergence, notay2008recursive, brown1994convergence}, $(ii)$ they are mostly applicable only to very small classes of problems \cite{kabel2014efficient, garrappa2015solving, asgharzadeh2017newton}, $(iii)$ they generally do not produce exact solutions in finite time, but exhibit a trade-off between runtime and accuracy. Improving this trade-off by finding near-optimal preconditioners \cite{benzi2011relaxed, benzi2004preconditioner, turkel1987preconditioned} is a hard problem \cite{chan1988optimal, knyazev2001toward} in the general case.

These common limitations should not be confused for rules, however: because of their diversity, Krylov methods can take on a variety of useful properties when specially designed. For example, when an optimal preconditioner is selected they can find the exact solution in a finite number of iterations. The algorithms we present in this paper will be partially based on an optimally preconditioned Krylov method that is applicable to a rather large class of stencil problems.

\para{Our FFT-Based Algorithms}
The computation $S^T a_0 = a_T$ (where $S^T$ denotes that the stencil $S$ is applied $T$ times) evolves the initial grid data $a_0$ for $T$ timesteps to produce the final data $a_T$. As will be seen in Section \ref{sec:fft}, any method of computing $a_{t+1} = S a_t$ is mathematically equivalent to a product where $S$ is viewed as a matrix and $a_t$ as a vector. All existing algorithms which find $a_T$ exactly do so by direct computation of $a_T = S ( S (S(\cdots a_0)))$, where $S$ is applied for a total of $T$ timesteps, incurring $\Th{NT}$ work in the process. The looping, tiled, and recursive algorithms we have described differ only in how they break up this series of matrix-vector products. We will instead evaluate this product by diagonalization and repeated squaring of the stencil matrix $S$.

In this paper, we present two FFT-based stencil algorithms: {\Periodic} for problems with periodic boundary conditions, and {\Aperiodic} for those with aperiodic boundary conditions. Our algorithms are applicable to arbitrary uniform linear stencils across vector-valued fields.

\vspace{0.1cm}
\noindent
\textbf{1. [\textit{Periodic Stencil Algorithm.}]}
Let the discrete Fourier transform (DFT) matrix be written $\fft[i,j] = \omega_N^{-ij} / \sqrt{N}$, where $\omega_N = e^{2\pi \sqrt{-1} / N}$, and let $\ifft$ be the inverse DFT matrix. We use fast Fourier transforms to compute $a_T$ as follows:
\begin{equation*}
\fbox{$a_T = \ifft (\fft S \ifft)^T \fft a_0$}
\end{equation*}
where $\fft S \ifft$ is a diagonal matrix, and $T$ = \#timesteps (not a transposition).

To the best of our knowledge, this is the first time that FFT is being applied directly to the problem of computing integral powers of the circulant \cite{gray2006toeplitz} stencil matrix $S$ that appears in linear stencil computations, even though there is a strong history of using FFT to improve the efficiency of matrix computations \cite{gohberg1994fast, tyrtyshnikov1996unifying, dietrich1997fast}.

\vspace{0.1cm}
\noindent
\textbf{2. [\textit{Aperiodic Stencil Algorithm.}]} 
When given aperiodic boundary conditions, we use a recursive divide-and-conquer strategy to solve for the boundary of the spatial grid; {\Periodic} is used as a subroutine to compute cells whose values are independent of the boundary. This method allows us to compute every timestep of the boundary in serial, yet to skip over computing most timesteps of cells near the middle of the grid.

Before we analyze the complexities of our algorithms briefly described above, we give the performance metrics that will be used. 

\para{Performance Metrics} We use the \textit{work-span} model \cite{CormenLeRiSt2009} to analyze the performance of dynamic multithreaded parallel programs. \textit{Work} $T_1(n)$ of an algorithm, where $n$ is the input parameter, denotes the total number of serial computations. \textit{Span} $T_{\infty}(n)$ of an algorithm, also called \textit{critical-path length} or \textit{depth}, denotes the maximum number of operations performed on any single processor when the algorithm is run on a machine with an unbounded number of processors. Our analysis of span is performed according to the \textit{binary-forking} model \cite{blelloch2019optimal}, in which spawning $n$ threads required $\Th{\log n}$ span. This model is stricter than PRAM, so all bounds we give hold in the PRAM model as well. \textit{Parallel running time} $T_p(n)$ of an algorithm when run on $p$ processors under a greedy scheduler is given by $T_p(n) = \Oh{T_1(n) / p + T_{\infty}(n)}$. \textit{Parallelism} of an algorithm is the average amount of work performed in each step of its critical path and is computed as $T_1(n) / T_{\infty}(n)$.

\Arxiv{
We use the ideal-cache model \cite{FrigoLePrRa1999} to measure data locality of algorithms. \textit{Serial cache complexity} $Q_1(n)$ denotes the total number of data block transfers between adjacent levels of memory. Cache-oblivious algorithms do not use cache parameters such as cache size $M$ and cache line size $B$ in their pseudocodes and hence are more portable across machines with different cache parameters.

Our algorithms are cache-oblivious \cite{FrigoLePrRa1999} and cache-adaptive \cite{Bender2014, Bender2016}.
}

\begin{table}[!t]
\centering
\hspace{-0.2cm}
\scalebox{0.62}{
\begin{colortabular}{ | l | l | l | l |}
\hline                       
Algorithm & Work $(T_1)$ & Span $(T_{\infty})$ & Result \\\hline
\rowcolor{lightred} \multicolumn{4}{|l|}{Existing Algorithms}\\ \hline
Nested Loop  & $\Th{NT}$ & $\Th{ T \log N }$ & 
 \\ 
Tiled Loop  & $\Th{NT}$ & $\Th{T \log M +  \frac{T}{M^{1/d}} \log \frac{N}{M} }$ & \cite{Bandishti2012, bondhugula2017}
 \\ 
D\&C  & $\Th{NT}$ & $\Th{ T (N^{1/d})^{\log {(d + 2)} - 1} }$ & \cite{Tang2011}
 \\ \hline
\rowcolor{lightgreen} \multicolumn{4}{|l|}{Our Algorithms}\\ \hline
{\Periodic} & $\Th{N \log (NT)}$ & $\Th{ \log T + \log N \log \log N }$ & Th. \ref{th:periodic-stencil-fft} \\

{\Aperiodic} & $\Th{ \begin{array}{@{}c@{}} TN^{1-1/d} \log \left( TN^{1-1/d} \right) \log T \\ + N \log N \end{array} }$ & 
$\begin{cases} \Th{T} &\mbox{if } d=1 \\
\Th{T \log N} & \mbox{if } d\ge2 \end{cases}$ & Th. \ref{th:aperiodic-stencil}
\end{colortabular} 
}
\caption{\small Complexity bounds for stencil algorithms, where $N =$ spatial grid size, $T =$ \#timesteps, and $M =$ cache size. It is important to note that the bounds given for our algorithms are for computations with $\Th{1}$-size stencils on $d = \Th{1}$ dimensional hypercubic grids, and that we have simplified the span of {\Aperiodic} by assuming $T = \Om{\log N \log \log N}$. The nested loop, tiled loop, and D\&C algorithms work for both periodic and aperiodic boundary conditions. The span for the tiled loop algorithm is $\Om{T \log \log N}$.}
\label{tab:summary}
\vgap{}\vgap{}
\centering
\scalebox{0.59}{
\begin{colortabular}{ | l | l | l | | l l | r  r | r r | r r |}
\hline                       

\multicolumn{5}{|l|}{\begin{tabular}{@{}p{2.3in}@{}}Benchmark\\\hline\end{tabular}} & \multicolumn{4}{c|}{\begin{tabular}{@{~}c@{~}}Parallel runtime in seconds\\\hline\end{tabular}} & \multicolumn{2}{c|}{Speedup factor}\\

\rowcolor{tabletitlecolor} \multicolumn{2}{|l|}{} & &  &  & \multicolumn{2}{c|}{\begin{tabular}{@{~}c@{~}}\pluto{}\\\hline\end{tabular}} & \multicolumn{2}{c|}{\begin{tabular}{@{~}c@{~}}Our algorithm\\\hline\end{tabular}} & \multicolumn{2}{c|}{\begin{tabular}{@{~}c@{~}}over \pluto{}\\\hline\end{tabular}} \\

\rowcolor{tabletitlecolor}\multicolumn{2}{|l|}{} & Stencil & $N$ & $T$ & KNL & SKX & KNL & SKX & KNL & SKX \\ \hline

\multicolumn{2}{|c|}{\multirow{6}{*}{\rotatebox{90}{Periodic}}} & \texttt{heat1d} & $1,600,000$ & $10^6$ & 79 & 19 & 0.25 & 0.03 & 1754.7 & 759.6\\
\multicolumn{2}{|l|}{}& \texttt{heat2d} & $8,000 \times 8,000$ & $ 10^5$ & 1,437 & 222 & 0.48 & 0.61 & 3,025.0 & 367.0\\
\multicolumn{2}{|l|}{}& \texttt{seidel2d} & $8,000 \times 8,000$ & $ 10^5$ & 500 & 808 & 0.48 & 0.64 & 1,032.7 & 1268.6\\ 
\multicolumn{2}{|l|}{}& \texttt{jacobi2d} & $8,000 \times 8,000$ & $ 10^5$ & 2,905 & 1017 &  0.48 & 0.68 & 6,084.7 & 1502.1\\
\multicolumn{2}{|l|}{}& \texttt{heat3d} & $800 \times 800 \times 800$ & $10^4$ & 816 & 1466 & 4.98 & 5.48 & 163.9 & 267.3\\
\multicolumn{2}{|l|}{}& \texttt{19pt3d} & $800 \times 800 \times 800$ & $10^4 $ & 141 & 158 & 4.84 & 5.78 & 29.1 & 27.3\\
\hline

\multirow{12}{*}{\rotatebox{90}{Aperiodic}} & \multirow{6}{*}{\rotatebox{90}{Experiment 1}} & \texttt{heat1d} & $1,600,000$ & $10^6$ & 50 & 35 & 5.85 & 6.69 & 8.5 & 5.2\\
&& \texttt{heat2d}& $8,000 \times 8,000$ & $ 10^5$ & 333 & 530 & 143.25 & 151.37 & 2.3 & 3.5 \\
&& \texttt{seidel2d}& $8,000 \times 8,000$ & $ 10^5$ & 345 & 601 & 145.42 & 132.97 & 2.4 & 4.5\\
&& \texttt{jacobi2d}& $8,000 \times 8,000$ & $ 10^5$ & 567 & 456 & 249.04 & 273.46 & 2.3 & 1.7\\
&& \texttt{heat3d}& $800 \times 800 \times 800$ & $10^4 $ & 513 & 763 & 395.10 & 605.89 & 1.3 & 1.3\\
&& \texttt{19pt3d}& $800 \times 800 \times 800$ & $10^4 $ & 645 & 848 & 425.22 & 616.71 & 1.5 & 1.4\\\cline{2-11}

 &\multirow{6}{*}{\rotatebox{90}{Experiment 2}}& \texttt{heat1d} & $1,600,000$ & $N$ & 32 & 23 & 5.63 & 6.87  & 5.7 & 3.3\\
&& \texttt{heat2d} & $8,000 \times 8,000$ & $ \sqrt{N}$ & 210 &  312 & 92.78 & 121.70 & 2.3 & 2.6\\
&& \texttt{seidel2d} & $8,000 \times 8,000$ & $ \sqrt{N}$ & 228 & 375 & 91.59 & 121.46 & 2.5 & 3.1\\
&& \texttt{jacobi2d} & $8,000 \times 8,000$ & $ \sqrt{N}$ & 372 & 281 & 151.31 & 198.00 & 2.5 & 1.4\\
&& \texttt{heat3d} & $800 \times 800 \times 800$ & $\sqrt[3]{N}$ & 45 & 71 & 32.29 & 50.52 & 1.4 & 1.4\\
&& \texttt{19pt3d} & $800 \times 800 \times 800$ & $\sqrt[3]{N}$ & 61 & 71 & 33.82 & 52.27 & 1.8 & 1.4

\end{colortabular}
}
\caption{Performance summary of parallel stencil algorithms on a KNL/SKX node.}
\vspace{-0.3cm}
\label{tab:experimental-results}
\vgap{}\vgap{}\vgap{}
\end{table}

\para{Performance Analysis of Our Algorithms} 
The performance of our periodic and aperiodic stencil algorithms are summarized in Table \ref{tab:summary}. We see that: $(i)$ Both work and span of {\Periodic} have only logarithmic dependence on $T$ compared with the linear dependence on $T$ in the existing algorithms. $(ii)$ For a $d$-D problem, {\Aperiodic} has work quasilinearly dependent on $(T N^{1-1/d} +N)$, whereas all existing algorithms for general linear stencils perform $\Th{NT}$ work -- a polynomially greater amount. This asymptotic improvement makes possible stencil computations over much larger spacetime grids.

Although we do not show explicit analysis of cache complexity in this paper, it is worth noting that our algorithms are cache-oblivious \cite{FrigoLePrRa1999} and cache-adaptive \cite{Bender2014, Bender2016}.

\para{Our Contributions}
Our key contributions are as follows:\\
\textbf{1. [\textit{Theory.}]} We present the first algorithms for general linear stencil computations (for both periodic and aperiodic boundary conditions) with $\oh{NT}$ work and low span, achieving polynomial speedups over the best existing stencil algorithms.\\
\textbf{2. [\textit{Practice.}]} We experimentally analyze the numerical accuracy and runtime of our algorithms as compared to PLuTo \cite{Pluto} code. Implementations of our algorithms for on the order of $10^7$ grid cells and $10^5$ timesteps suffer no more loss in accuracy from floating point arithmetic than PluTo code, yet run orders of magnitude faster than the best existing implementations of state-of-the-art algorithms for periodic stencil problems, and 1.3$\times$ to 8.5$\times$ faster for aperiodic stencil problems. This is shown in Table \ref{tab:experimental-results}. Our code is publicly available at\\
\texttt{https://github.com/TEAlab/FFTStencils}.

\tikzset{%
>={},
base/.style = {rectangle, rounded corners, draw=black,
text centered},
bluerect/.style = {base, fill=lightblue},
redrect/.style = {base, fill=lightred},
greenrect/.style = {base, fill=lightgreen},
yellowrect/.style = {base, fill=lightyellow}
}

\section{Related Work \& Its Limitations}
There is substantial literature devoted to the applications and analysis of stencils and Discrete Fourier Transforms (DFTs). Here we give a background of the relationship between these two areas and examine some of the limitations which appear in the current approaches to stencil codes.

\para{Discrete Fourier Transforms} DFTs are widely used in numerical analysis, with examples including Von Neumann stability analysis \cite{wesseling1996neumann, pereda2001analyzing} to show validity of numerical schemes, DFT-based preconditioning to optimize Krylov iterations \cite{chan1993fft, gutknecht2007brief, kirby2018solver}, and time-domain analysis to achieve partial solutions of given PDEs \cite{hockney1965fast, choi1986finite, mugler1988fast, schumann1988fast, ostashev2005equations}.

A Fast Fourier Transform (FFT) is an algorithm that quickly computes the DFT of an array. The use of FFTs will be important for our algorithms, as they represent a particularly efficient type of matrix-vector multiplication. Several $\Oh{N \log N}$-work FFT algorithms exist \cite{Good1958, Bruun1978, Winograd1978}, the most famous among which is the Cooley-Tukey algorithm \cite{CooleyTu65}.

\begin{theorem}[Cooley-Tukey Algorithm, \cite{CooleyTu65, FrigoLePrRa1999}]
\label{thm:fft-cooley-tukey}
The generic Cooley-Tukey FFT algorithm computes the DFT of an array of size $N$ in $\Th{N \log N}$ work, $\Oh{\log N \log \log N}$ span, and $\Th{N}$ space.
\end{theorem}

\para{Stencil Problems} \hide{
A stencil is a locally defined operator that is applied across a grid of data to approximate some function, generally time evolution.
}Stencils are often used in numerical analysis as discretizations of PDEs, since many simple PDEs have prohibitively complex analytical solutions \cite{ramani1989painleve, fokas2000integrability} but allow good numerical approximations with a proper choice of stencil.

There are two major types of methods related to stencils: those for deriving numerical schemes, and those for evolving grid data via a given stencil. A \highlight{discretization method} is a way of converting a PDE, which deals with quantities defined over a continuum, into a stencil, which relates quantities defined over discrete sets of variables. A \highlight{stencil solver} is an algorithm that takes stencils, boundary conditions, and initial data as input and performs stencil computations to output final data. This paper presents a pair of stencil solvers for linear stencils which support, respectively, periodic and aperiodic boundary conditions.

\para{Stencil Solvers} Numerical results from stencils are obtained through two primary paths: \highlight{direct solvers} and \highlight{Krylov methods} \cite{saad1989krylov, ipsen1998idea}.

Direct solvers are those which find the solution to a stencil problem in a finite number of steps. They often involve feeding the stencil into a stencil compiler such as \pluto{} \cite{bondhugula2008pluto}, Pochoir \cite{tang2011pochoir}, or Devito \cite{louboutin2019devito}, which will output optimized code to compute the action of the stencil across some prespecified grid of initial data for multiple timesteps. Cutting-edge stencil code generators feature many improvements over simple looping algorithms, including better cache efficiency \cite{frigo2005cache, korch2020depth}, parallelism \cite{kong2013polyhedral}, and low-level compiler optimizations. These systems all perform the same set of updates on the stencil grid, although they vary in the order that these updates are performed. \textit{In general they make no use of FFTs.}

Krylov subspace methods produce successively better approximations of the exact solution to a given stencil problem. Krylov solvers are often used to solve problems for which there is no known direct solution technique. 
It is common for DFTs to be used either in the analysis \cite{boyd2001chebyshev} or implementations \cite{andreussi2012revised, kabel2014efficient, guan2019two} of Krylov subspace methods, as Fourier analysis is useful for proving scheme stability \cite{chan1984stability, yuste2005explicit} and convergence rates \cite{gmeiner2013optimization}, and DFT matrices are good preconditioners \cite{erlangga2006novel} for a large class of matrix equations \cite{borrell2011parallel}. In some instances choosing the DFT matrix as a Krylov preconditioner can even convert an approximate solver into a direct one \cite{hockney1965fast, fuka2015poisfft}. Krylov subspace methods generally do not produce exact solutions in finite time, but exhibit a \textit{trade-off between runtime and accuracy}.

\vspace{0.2cm}
\para{Limitations of Current Methods for Stencil Problems} 

\vspace{0.1cm}
\noindent
\textbf{(1) \textit{Manual Analysis.}}
Krylov subspace methods are often accompanied by a mathematically nontrivial convergence analysis \cite{kuijlaars2006convergence, notay2008recursive, brown1994convergence}; this requires human labor for every new method developed. Since this analysis is not automated \cite{guttel2014nleigs, musco2015randomized, frommer2017block}, the quantity of time it takes varies widely from case to case. Also, the requirement of mathematical rigour strongly discourages the development of unnecessarily general Krylov methods, thus those in the literature usually only apply to specific stencils \cite{ashrafizadeh2015jacobian, mang2015inexact, garrappa2015solving, asgharzadeh2017newton} in order to simplify analysis.

\vspace{0.1cm}
\noindent
\textbf{(2) \textit{Specialization.}}
The methods \cite{schumann1988fast, kabel2014efficient, garrappa2015solving, asgharzadeh2017newton} published in most of the existing literature on computation with numerical schemes are only applicable to very small classes of problems. DFT preconditioning for Krylov iterations has been used for specific stencils before \cite{janpugdee2006accelerated, fritz2009application, storti2013fft}. However, these techniques have not been generalized to work for higher dimensional grids with general linear stencils.

\vspace{0.1cm}
\noindent
\textbf{(3) \textit{Inexact Solution.}}
Krylov methods often cannot produce exact solutions, even in the absence of floating-point rounding errors. Using them optimally and reliably thus requires expertise in numerical analysis \cite{van2003iterative, ipsen1998idea}, as well as for substantial effort to be put into uncertainty quantification \cite{le2010spectral, koutsourelakis2009accurate}.

\vspace{0.1cm}
\noindent
\textbf{(4) \textit{Nonoptimal Computational Complexity.}}
All currently available code compilers \cite{christen2011patus, hagedorn2018high, henretty2013stencil, luporini2020architecture, ragan2013halide, bondhugula2008pluto, tang2011pochoir} generate code which has linear work complexity in the number of grid cells and number of time steps to compute, no matter what stencil they are given. Improving this bound has not been addressed in the literature, even when only considering linear stencils. 

We show in this paper that when dealing with linear stencils it is possible to produce code that has significantly better asymptotic performance.


\vspace{0.1cm}
\noindent
\textbf{(5) \textit{No Support for Implicit Stencils.}}
Direct solvers for stencil problems usually do not support implicit stencils; stencil compilers in particular are weak in this respect. Neither Pochoir, PLuTo, nor Devito can be used for directly evolving data via an implicit stencil. This is a significant limitation, as several important stencils are implicit \cite{acary2010implicit, kloeden1992higher, najm1998semi, robert1982semi, crank1947practical, van1984enhancements}.

\para{Significance of This Paper}
Current direct solvers for general linear stencil computations are inefficient, and Krylov methods require manual labor and mathematical training. We solve these problems for linear stencils by using DFT preconditioning on a Krylov method to achieve a direct solver which is both fast and general.
\section{Applicability of Our Algorithms}
In this section, we describe the classes of stencil problems on which our FFT-based stencil algorithms do or do not apply.

\para{Supported Stencil Types}
Our algorithms are most directly applicable to \highlight{homogeneous linear stencils across vector-valued fields}. A homogeneous stencil is one that does not vary across the entire spatial grid, and by vector-valued fields we mean we allow each cell value across the spatial grid to be treated as a vector.

All homogeneous linear PDEs can be discretized into supported stencils by using a finite difference approximation \cite{isaacson2012analysis}. Thus all numerical results for these linear PDEs that were previously reached via analytically-motivated numerical schemes, including those set in the Fourier domain \cite{mugler1988fast, moaddy2011non}, can easily be reached computationally by our algorithms.

It is noteworthy that we support both explicit and implicit linear stencils, and also that vector-valued fields can be used to enable stencils that are dependent on more than one previous timestep. In addition, vector-valued fields allow us to support certain types of inhomogeneity, as mentioned at the end of this section.

Linear stencils are quite common in computational numerics. In fact, the majority of stencils currently used to benchmark \cite{Rawat2018, Pluto, bondhugula2008pluto, bondhugula2016plutoplus} stencil compilers are linear.

\Arxiv{
An example of a numerical scheme that can be handled by our algorithms can be found in Yee's method \cite{yee1966numerical, monk1994convergence}, a finite-difference time-domain scheme for solving Maxwell's equations in an isotropic medium. In the 2-D case, Yee's method can be written out \cite{sun2003unconditionally} as updating three variables, $E_x$, $E_y$, and $H_z$, which are the strength of the electric field in the $x$ and $y$ directions and strength of magnetic field, respectively. These fields are updated according to the equations
\begin{align*}
    E_x^{n+1}(i+1/2,j) &= E_x^n(i+1/2, j) + \frac{a_1}{\Delta y} (H^{n+1}_z(i+1/2, j+1/2)\\
    &- H^{n+1}_z(i+1/2, j-1/2) + H^{n}_z(i+1/2, j+1/2)\\
    &- H^{n}_z(i+1/2, j-1/2)),\\
    E_y^{n+1}(i,j+1/2) &= E_y^n(i, j+1/2) - \frac{a_1}{\Delta x} (H^{n+1}_z(i+1/2, j+1/2) \\
    &- H^{n+1}_z(i-1/2, j+1/2) + H^{n}_z(i+1/2, j+1/2)\\
    &- H^{n}_z(i-1/2, j+1/2)),\\
    H_z^{n+1}(i+1/2,j+1/2) &= H_z^n(i+1/2, j+1/2)\\
    &+ \frac{a_2}{\Delta y} (E^{n+1}_x(i+1/2, j+1) - E^{n+1}_x(i+1/2, j)\\
    &+ E^{n}_x(i+1/2, j+1) - E^{n}_x(i+1/2, j))\\
    &- \frac{a_2}{\Delta x} (E^{n+1}_y(i+1, j+1/2) - E^{n+1}_y(i, j+1/2)\\
    &+ E^{n}_y(i+1, j+1/2) - E^{n}_y(i, j+1/2)),
\end{align*}
where the parameters $a_1 = \Delta t / (2\varepsilon)$ and $a_2 = \Delta t / (2\mu)$ are dependent on the permittivity $\varepsilon$ and permeability $\mu$ of the medium. These equations are linear and can be rewritten as a single stencil operator acting on a vector valued field with $(E^n_x(i+1/2, j), E^n_y(i, j+1/2), H^n_z(i+1/2, j+1/2))$ at every point $(i,j)$ and time $n$.
}

\para{Unsupported Stencil Types}
Our algorithms are not applicable to \highlight{nonlinear stencils}. This is because introducing nonlinearity of any sort invalidates our technique of using DFTs to simplify the action of the stencil. Common examples of nonlinearity in stencils include conditionals, i.e. max/min/if-else, and quadratic dependence on cell values. Most discretizations of nonlinear PDEs pass the nonlinearity on to the stencil, so in general our algorithms cannot be used for stencils from nonlinear PDEs.

Our algorithms cannot be applied to \highlight{inhomogeneous stencils} either. There are two ways that a stencil can break homogeneity. The first is spatially, by having the stencil itself be dependent on local field data, such as is used in slope limiter and flux limiter methods \cite{sweby1984high, gammie2003harm, kuzmin2010vertex} and in mixed media problems \cite{saleheen1997new, teixeira1998finite, chen2005time, xu2006simulations, teixeira2008time, virieux1986p}. The second way to break homogeneity is temporally, i.e. using a stencil that is dependent on time \cite{he2011robust, konuk2020modeling}, as would arise from the presence of a forcing term in the original PDE being discretized.

However, we note that there are some special types of inhomogeneity our algorithms can handle, such as those arising from forcing terms which are low-order polynomials in time. These are handled by discretizing homogeneous systems of PDEs to mimic the behaviour of a single inhomogenous PDE.

\Arxiv{
To show how both stencil homogeneity and linearity can be broken through the introduction of a flux limiter, consider a conservation law $u_t + f(u)_x = 0$ being discretized in conservation form \cite{sweby1984high}
$$u_k^{n+1} = u_k^n - \frac{\Delta t}{\Delta x} (h_{k+1/2}^n - h_{k-1/2}^n),$$
where $h_{k+1/2}^n$ represents the edge flux through the right side of cell $k$. When we have both high and low resolution schemes $F$ and $f$ for these edge fluxes, we mix them using a flux limiter $\phi(r)$ like so:
$$h_{k+1/2}^n = f_{k+1/2}^n - \phi\left(\frac{u_k^n - u_{k-1}^n}{u_{k+1}^n - u_k^n}\right) \left(f_{k+1/2}^n - F_{k+1/2}^n\right).$$
Most choices of flux limiter will clearly produce a nonlinear stencil. A prominent example would be the minmod limiter \cite{roe1986characteristic}
$$\phi(r) = \max[0, \min(1, r)].$$

In light of the nonlinearity introduced by flux limiters, the majority of conservative methods are not supported by the algorithms we present here.
}

\section{Periodic Stencil Algorithm}
\label{sec:fft}

In this section we present {\Periodic}, an efficient parallel algorithm for performing stencil computations with periodic boundary conditions using fast Fourier transforms (FFT). We begin by considering explicit linear stencils on one-dimensional spatial grids, after which we give simple extensions to high-dimensional grids, implicit stencils, and grids where cells are vector-valued.

\para{Mathematical Formulation}
Suppose we have a spatial grid of data that evolves in time according to some fixed stencil: cells in the grid at time $t$ are updated as a function of some local neighborhood of cell values at recent times before $t$.
For simplicity's sake, we will assume that the spatial grid is \highlight{one dimensional} until Section \ref{ssec:periodic_generalizations}.

In this section, we will exclusively consider linear stencils. A \highlight{linear stencil} $\mathbf{S}$ defines future array values $a_{t+1}[0,\ldots,N-1]$ as a linear function of current array values $a_t[0,\ldots,N-1]$. We will later allow array values to be higher dimensional, but for now these constraints are enough on their own to make a significant statement about stencils.

An arbitrary linear mapping from arrays of size $N$ to arrays of size $N$ is, by definition, an $N \times N$ \highlight{matrix}\footnote{
Here we are using the word matrix in the strictly \highlight{mathematical} sense, i.e. the object which is a stencil behaves in all respects identically to the way in which a matrix behaves. This should not be taken to mean that we will store stencils using the \highlight{data structure} called a matrix. In fact, we shall show that there are other more efficient ways of storing our stencils.
}, and therefore the update rule can be written $a_{t+1}[i] = \sum_j S[i,j] a_{t}[j]$. We will usually omit the indices in formulas like this, writing $a_{t+1} = S a_{t}$.
As shown in Figure \ref{fig:periodicity_need}, on the surface this update rule may look incomplete -- since we require $S$ to use the same neighborhood around each point for updates, how should we proceed when the neighborhood extends beyond the bounds of our spatial array? As filling in these cells is exactly the purpose of boundary conditions, it should come as no surprise that in this section the resolution to this difficulty will come in the form of periodicity.

\begin{figure}[!ht] 
\centering
\includegraphics[width=0.37\textwidth]{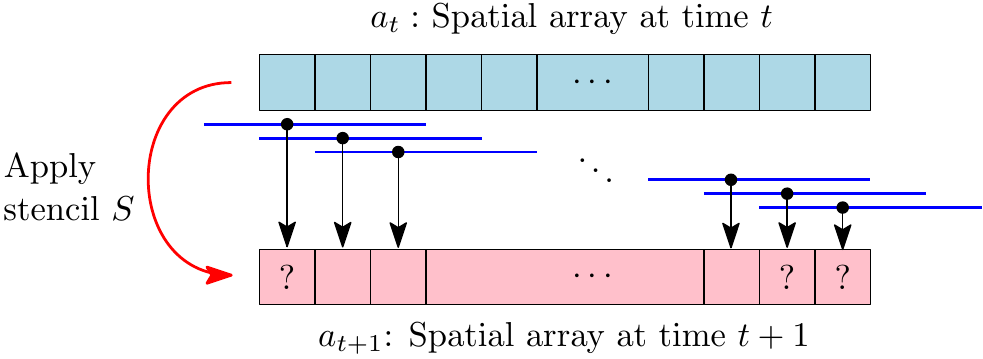}
\vgap{}
\caption{Updating data with a stencil that uses a neighborhood extending one cell to the left and two cells to the right. Cells marked with question marks do not have values defined by the stencil, as the neighborhood required to update them lies partially outside the bounds of the spatial array.}
\label{fig:periodicity_need}
\vgap{}\vgap{}
\end{figure}

Periodic boundary conditions consist of the rule that $a_t[i] = a_t[i+N]$, for all $i$. This allows the spatial grid to be extended arbitrarily far in either direction by wrapping around instead of moving outside of the array bounds; in the presence of periodic boundary conditions the spatial grid is effectively a torus, with no clearly defined boundary.

For a stencil to be uniform across space means that it defines updated cell values only from a set of cells which are selected based on their \highlight{relative} location to the cell being updated. For example, if we were to reindex all cells in the spatial array, incrementing them all so the bounds became $1$ and $N$ rather than $0$ and $N-1$, this change of index ought to be invisible to the stencil. Furthermore, in the presence of periodic boundary conditions, this reindexing is equivalent to cyclically shifting the field data $a_t$, since we have $a_t[0] = a_t[N]$.

We mathematically represent the concept of cyclically permuting grid data by introducing the \highlight{right shift matrix}\footnote{Under periodic boundary conditions, shifting the array is equivalent to rotating the array.} \cite{rota1973foundations} $X$. The array $Xa_t$ is defined to be the result of taking the rightmost element $a_t[0]$ off and appending it to the left side of the array, i.e. its action on arrays is $(Xa_t)[i] = a_t[i - 1]$. An equivalent definition is by $X$'s matrix elements $X[i,j] = \delta_{i, j+1}$, where the \highlight{Kronecker delta} $\delta_{i,j}$ is defined to be 1 if $i = j$ and 0 otherwise, and the arithmetic is understood to be \highlight{modular} with base $N$ in the presence of periodic boundary conditions. 

\begin{figure}[!ht] 
\centering
\includegraphics[width=0.47\textwidth]{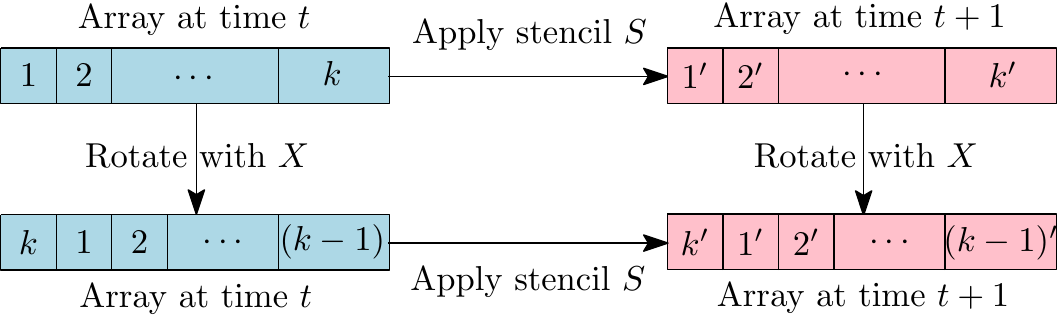}
\vgap{}
\caption{The stencil $S$ does not care where cells are with respect to the start of the array; it only cares about where they are with respect to one another. The shift matrix $X$ does not affect cells' relative locations, so it does not change the values of the updated cells.}
\label{fig:com_diag}
\vgap{}\vgap{}
\end{figure}

Given that we have periodic boundary conditions and the update rule is fixed across space, the action of our stencils must be invariant under spatial shifts of the grid values. As shown in Figure \ref{fig:com_diag}, cyclically permuting $a_t$ and then applying $S$ should give the exact same result as applying $S$ and then cyclically permuting $a_{t+1}$. In symbols, we have $SX = XS$, which implies that $S$ must be a \highlight{circulant matrix} \cite{gray2006toeplitz}, satisfying $S[i,j] = S[i-j,0]$. If we name these elements $S[i-j,0] = s_{i-j}$, we can write out the full update equation as follows:

\newenvironment{sbmatrix}[1]
 {\def\mysubscript{#1}\mathop\bgroup\begin{bmatrix}}
 {\end{bmatrix}\egroup_{\textstyle\mathstrut\mysubscript}}

\noindent
$$\begin{sbmatrix}{a_{t+1}} a_{t+1}[0] \\ a_{t+1}[1] \\ \vdots \\ a_{t+1}[N-2] \\ a_{t+1}[N-1]\end{sbmatrix} = \begin{sbmatrix}{S} s_0 & s_{N-1} & \cdots & s_2 & s_1 \\
s_1 & s_0 & s_{N-1} &   & s_2 \\
\vdots & s_1 & s_0 & \ddots & \vdots \\
s_{N-2} &  & \ddots & \ddots & s_{N-1} \\
s_{N-1} & s_{N-2} & \cdots & s_1 & s_0\end{sbmatrix} \begin{sbmatrix}{a_{t}} a_{t}[0] \\ a_{t}[1] \\ \vdots \\ a_{t}[N-2] \\ a_{t}[N-1]\end{sbmatrix}.$$

A example of this matrix with concrete numbers is constructed in Appendix \ref{ssec:stencil_matrix_walkthrough}.

A useful representation of circulant matrices is found through the right shift matrix. Notice that powers of $X$ have the property $(X^k a_t)[i] = a_t[i-k]$, which means that their matrix elements are given by $(X^k)[i,j] = \delta_{i,j+k}$. Thus powers of $X$ allow us to pick out the individual diagonals that appear in circulant matrices; any circulant matrix $S$ can be expanded in terms of shift operators as
$$S = \sum_i S[i,0]X^i,$$
the proof of which is given in Appendix \ref{ssec:shift_matrix_decomposition}.

The above equation shows that circulant matrices can be uniquely specified by a single one of their columns or rows. We will make use of this fact to avoid performing redundant computations: for all of the algorithms presented in this paper, we will store \highlight{only the first column} of $S$ in memory.

\para{Reformulating the Final Data} We now turn back to the definition of the final data $a_T = S^T a_0$ and the DFT matrix $\mathcal{F}$. Here the exponent $T$ will always denote a matrix power, \highlight{not a transposition}. As before, the DFT matrix has elements $\mathcal{F}[i,j] = \omega_N^{-ij}$, where $\omega_N = e^{2\pi\sqrt{-1}/N}$ is a primitive $N$th root of unity, and $\mathcal{F}$'s inverse has elements $\mathcal{F}^{-1}[i,j] = \omega_N^{ij}/N$. Since we know that $\mathcal{F}^{-1} \mathcal{F}$ is the identity, it can make no difference mathematically to drop it into our equation for the final data:
$$a_T = \mathcal{F}^{-1} \mathcal{F} S^T \mathcal{F}^{-1} \mathcal{F} a_0.$$

Continuing to insert identities and regrouping, we find that $\mathcal{F} S^T \mathcal{F}^{-1} = \mathcal{F} S \mathcal{F}^{-1}\mathcal{F} S \mathcal{F}^{-1} \cdots \mathcal{F} S \mathcal{F}^{-1} = (\mathcal{F} S \mathcal{F}^{-1})^T$, so we can rewrite our equation as
$$a_T = \mathcal{F}^{-1} (\mathcal{F} S \mathcal{F}^{-1})^T \mathcal{F} a_0.$$

This form of the final data equation points to a remarkably efficient way of computing $a_T$. We first apply the \highlight{convolution theorem} \cite{bracewell1986fourier}, which states that if $S$ is a circulant matrix, then $\mathcal{F}S\mathcal{F}^{-1} = \Lambda$ is diagonal. The equation for final data can now be regrouped with $\mathcal{F}S\mathcal{F}^{-1} = \Lambda$ and $\mathcal{F}a_0 = x$:
\begin{equation}\label{eq:final_data}
\boxed{a_T = \mathcal{F}^{-1} \Lambda^T x.}
\end{equation}

This equation may appear to be more complicated than what we started with, but really all we have done is made a change of basis into the frequency domain and performed all actions of the stencil there. This will now be shown to be computable with only a couple calls to highly efficient FFTs and some repeated squarings of scalars.

\subsection{1-D Explicit Stencil Algorithm}

Let $a_0[0,\ldots,N-1]$ be the initial 1-D spatial grid data to be acted on with the stencil $S = \sum_i S[i,0] X^i$. We will impose periodic boundary conditions; these allow us to benefit from pushing the stencil and initial data into the frequency domain with $\mathcal{F}S\mathcal{F}^{-1} = \Lambda$ and $\mathcal{F}a_0 = x$. Our goal is to compute the final data
$$a_T = S^T a_0 = \mathcal{F}^{-1} \Lambda^T x.$$

The prevalent approach to finding $a_T$ is currently through iterative applications of the stencil to field data, grouping $S^T a_0$ into $S(S(\cdots S(a_0)\cdots))$ and evaluating according to parenthesization. Here we will instead compute a power of the diagonalized stencil $(\mathcal{F} S\mathcal{F}^{-1})^T = \Lambda^T$ by repeated squaring, after which we will apply it to $\mathcal{F} a_0$, giving us $\mathcal{F} S^T a_0$, from which we can recover $S^T a_0$ with an inverse FFT.

It is shown in Appendix \ref{ssec:eig_fft_computing} 
that we can write the elements of the diagonal matrix $\Lambda$ as $\Lambda[i,i] = (\mathcal{F}s)[i]$, where $s$ is the column of $S$ that we are storing in memory, i.e. $s[i] = S[i,0]$. Since $\mathcal{F}$ is the DFT matrix, this means that $\Lambda$ can be computed with a single FFT.

Blindly using repeated squaring only allows us to compute $\Lambda^T$ when $T$ is an exact power of 2; arbitrary positive integer powers are computed as follows. Let $\sum_i b_i 2^i$ be the binary representation of $T$. As we compute successive squares of $\Lambda$, i.e. $\Lambda^{2^i}$, we will multiply them into a running total only if $b_i = 1$. Thus the final result will be
$$\Lambda^T = \prod_{i\colon b_i = 1} \Lambda^{2^i}.$$
Since $\Lambda$ is diagonal, elements of the large matrix power $\Lambda^T$ can be computed by taking powers of the original elements, $\Lambda^T[i,i] = \Lambda[i,i]^T$. Evaluating each of these elements in parallel will improve the span of our algorithm.

Wrapping up, we evaluate Equation \ref{eq:final_data} as follows: we find $x = \mathcal{F} a_0$ by applying FFT to the initial data, $\Lambda^T x$ by elementwise multiplication, and then $a_T$ with an inverse FFT.

We now present the {\StencilViaFFTPeriodic} algorithm, which efficiently performs stencil computations with periodic boundary conditions by transferring almost all relevant calculations to within the frequency domain. A diagrammatic outline is shown in Figure \ref{fig:periodic-block-diagram}, and the pseudocode is given in Figure \ref{fig:stencil-fft-algorithm-dD}.

\noindent
[\textbf{Step 1. FFT.}] We compute $(i)$ $\mathcal{F}S\mathcal{F}^{-1}$ from $S$, and $(ii)$ $\mathcal{F}a_0$ from $a_0$.
Since $S$ is circulant, we know that the FFT of $S$'s first column contains exactly the same information as $\mathcal{F}S\mathcal{F}^{-1}$. Thus for $(i)$ an FFT is applied to the first column of $S$ to get $S$'s eigenvalues; this FFT will be computed for $N$ points, irrespective of how many nonzero coefficients are present in the stencil. Note that only the first column of $S$ is needed here, which is why the rest of $S$ is never constructed or stored in memory. Likewise, for $(ii)$ an FFT is applied to $a_0$. \\
\noindent
[\textbf{Step 2. Repeated Squaring.}] We compute $\mathcal{F}S^T\mathcal{F}^{-1}$ from $\mathcal{F}S\mathcal{F}^{-1}$.
Since $\mathcal{F}S\mathcal{F}^{-1}$ is diagonal, the individual elements of $\mathcal{F}S^T\mathcal{F}^{-1} = (\mathcal{F}S\mathcal{F}^{-1})^T$ can be computed in parallel by performing $\lceil \log T \rceil$ sequential squarings for each element along the principal diagonal of $\mathcal{F}S\mathcal{F}^{-1}$ according to the decomposition of $\Lambda^T$ given earlier. \\
\noindent
[\textbf{Step 3. Elementwise Product.}] We compute $\mathcal{F} a_T$ by taking the product of $\mathcal{F}S^T\mathcal{F}^{-1}$ and $\mathcal{F} a_0$.
As in step 2, every element of $\mathcal{F} a_T = (\mathcal{F}S^T\mathcal{F}^{-1})(\mathcal{F} a_0)$ can be computed in parallel, since we are multiplying a vector by a diagonal matrix.\\
\noindent
[\textbf{Step 4. Inverse FFT.}] We now compute $a_T$ by applying an inverse FFT to $\mathcal{F} a_T$.

\begin{theorem}
\label{th:periodic-stencil-fft}
{\Periodic} computes the $T$th timestep of a stencil computation on a periodic grid of $N$ cells in $\Th{N \log (NT)}$ work and $\Th{\log T + \log N \log \log N}$ span.
\end{theorem}
\begin{proof}
Theorem follows from bounds given in Table \ref{tab:1d-periodic-complexity}.
\end{proof}

\begin{figure}[!t] 
\centering
\includegraphics[width=0.48\textwidth]{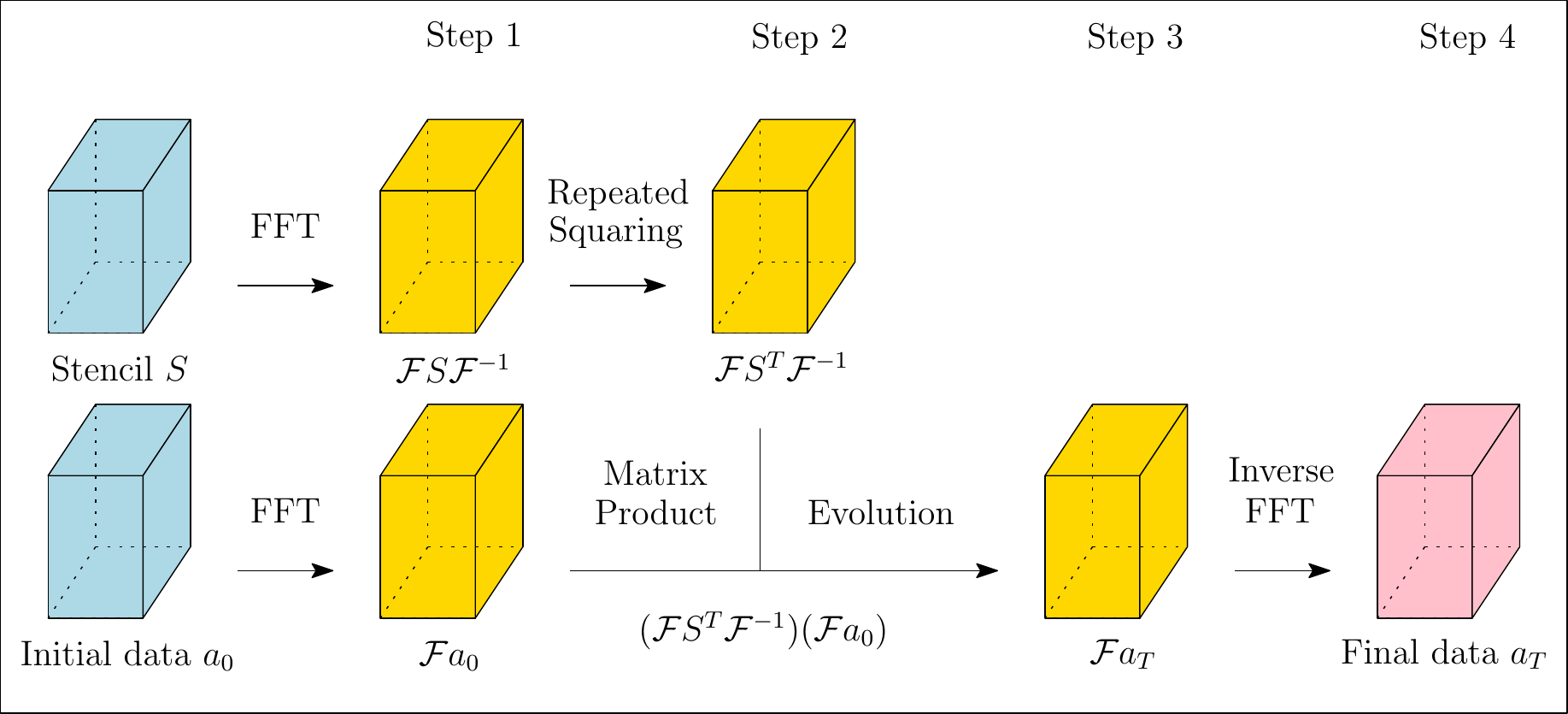}
\vgap{}\vgap{}\vgap{}
\caption{Block diagram of our FFT-based periodic stencil algorithm, which works for all grid sizes in all dimensions.}
\label{fig:periodic-block-diagram}
\begin{minipage}{0.5\textwidth}
\centering
\scalebox{0.95}{
\begin{mycolorbox}{$\Periodic(s, a_0, \ell_1, \ldots, \ell_d, T)$}
\begin{minipage}{0.99\textwidth}
{\scriptsize
\algotopspace{}
\noindent
\begin{enumerate}
\setlength{\itemindent}{-2em}

\vsitem[] [\textbf{Step 1. FFT.}] \hrulefill \; From $S$ to $\Lambda = \mathcal{F}S\mathcal{F}^{-1}$, from $a_0$ to $x = \mathcal{F}a_0$.
\vsitem $\Lambda \gets \MFFT(s)$
\vsitem $x \gets \MFFT(a_0)$

\vsitem[] [\textbf{Step 2. Repeated Squaring.}] \hrulefill \; From $\Lambda$ to $V = \Lambda^T$.
\vsitem $V \gets \text{Array of size $\ell_1 \times \cdots \times \ell_d$ initialized with all 1s}$
\vsitem $R \gets \text{Array of size $\ell_1 \times \cdots \times \ell_d$ initialized with all $T$s}$
\vsitem $\text{Squares} \gets \Lambda$ \xcomment{We'll store $\Lambda^{2^k}$ in Squares}
\vsitem \xparallelfor $j_1 \gets 0$ \xto $\ell_1 - 1$ \xdo \label{line:iter-sqr-start-dD}
\vsitem[] $\cdots$
\vsitem \xparallelfor $j_d \gets 0$ \xto $\ell_d - 1$ \xdo
\vsitem \T \xfor $k \gets 0$ \xto $\log T$ \xdo
\vsitem \T \T \xif $R[j_1 .. j_d]$ is odd \xthen \xcomment{Picking out binary representation of $T$}
\vsitem \T \T \T $V[j_1 .. j_d] \gets V[j_1 .. j_d] \times \text{Squares}[j_1 .. j_d]$
\vsitem \T \T \T $R[j_1 .. j_d] \gets R[j_1 .. j_d] - 1$
\vsitem \T \T $R[j_1 .. j_d] \gets R[j_1 .. j_d] / 2$
\vsitem \T \T $\text{Squares}[j_1 .. j_d] \gets \text{Squares}[j_1 .. j_d]^2$ \label{line:iter-sqr-end-dD}

\vsitem[] [\textbf{Step 3. Convolution.}] \hrulefill \; From $V = \Lambda^T$ and $x$ to $y = \Lambda^T x$.
\vsitem $y \gets \text{Array of size $\ell_1 \times \cdots \times \ell_d$}$
\vsitem \xparallelfor $j_1 \gets 0$ \xto $\ell_1 - 1$ \xdo
\vsitem[] $\cdots$
\vsitem \xparallelfor $j_d \gets 0$ \xto $\ell_d - 1$ \xdo
\vsitem \T $y[j_1 .. j_d] \gets V[j_1 .. j_d] \times x[j_1 .. j_d]$

\vsitem[] [\textbf{Step 4. Inverse FFT.}] \hrulefill \; From $y = \mathcal{F}a_T$ to $a_T$.
\vsitem $a_T \gets \IMFFT(y)$

\algobottomspace{}
\end{enumerate}
}

\end{minipage}
\end{mycolorbox}
}
\hide{
\vgap{}
\begin{mycolorbox}{$\MFFT(A, \ell_1, \ldots, \ell_d)$}
\begin{minipage}{0.97\textwidth}
{\scriptsize
\algotopspace{}
\noindent
\begin{enumerate}
\setlength{\itemindent}{-2em}

\vsitem \xfor $i \gets 1$ \xto $d$ \xdo
\vsitem \T \xparallelfor $j_1 \gets 0$ \xto $\ell_1 - 1$ \xdo
\vsitem[] \T $\cdots$
\vsitem \T \xparallelfor $j_{i-1} \gets 0$ \xto $\ell_{i-1} - 1$ \xdo
\vsitem \T \xparallelfor $j_{i+1} \gets 0$ \xto $\ell_{i+1} - 1$ \xdo
\vsitem[] \T $\cdots$
\vsitem \T \xparallelfor $j_d \gets 0$ \xto $\ell_d - 1$ \xdo
\vsitem \T \T $A[j_1 .. j_{i-1}, j_i[:], j_{i+1} .. j_d] \gets \textsc{FFT}(A[j_1 .. j_{i-1}, j_i[:], j_{i+1} .. j_d])$ \label{line:ffts-over-index}
\vsitem[] \xcomment{Taking the FFT over $j_i$}
\vsitem \textbf{return} $A$

\algobottomspace{}
\end{enumerate}
}
\end{minipage}
\end{mycolorbox}
}
\vgap{}\vgap{}
\caption{\small Arbitrary dimensional periodic stencil algorithm. The {\MFFT} algorithm takes FFTs over every index of the array passed to it. Lines \ref{line:iter-sqr-start-dD}-\ref{line:iter-sqr-end-dD} compute $\Lambda[j_1 .. j_d]^T$, where $T$ need not be an exact power of two.}
\label{fig:stencil-fft-algorithm-dD}
\begin{minipage}{\textwidth}
\centering
\scalebox{0.8}
{
\begin{colortabular}{ | c | l | l |}
\hline                       
Step(s) & $T_1$ & $T_{\infty}$ \\\hline
1,4 & $\Th{N \log N}$ & $\Th{\log N \log \log N}$ \\
2 & $\Th{N \log T}$ & $\Th{\log N + \log T}$ \\
3 & $\Th{N}$ & $\Th{\log N}$ \\
\hline
\end{colortabular}
}
\vgap{}\vgap{}
\tabcaption{Work and span complexity bounds for steps 1-4 of {\Periodic}. Note that we include thread spawning time from the Binary Forking model in these bounds.}
\label{tab:1d-periodic-complexity}
\vgap{}\vgap{}
\end{minipage}
\end{minipage}
\end{figure}

\subsection{Generalizations} \label{ssec:periodic_generalizations}
We now give overviews of some straightforward generalizations of the 1-D algorithm given above. These greatly enhance the scope of stencil problems our algorithm can handle, yet do not require significant conceptual work beyond what we have already presented. There are no changes in the asymptotic complexities of our algorithm when any of these generalizations are applied.

T

\para{Multi-Dimensional Stencils} The first extension of the 1-D version of our algorithm is to support arbitrarily high dimensional grids. The notions involved with solving across a grid of size $\ell_1 \times \cdots \times \ell_d = N$ differ from a 1-D grid only in that we now require indices for every grid dimension; details on the math can be found in Appendix \ref{ssec:math_for_P_Dd}. 
For our complexity analysis we assume that $d = \Th{1}$. 

\para{Implicit Stencils} The second way we expand the set of problems our algorithm can handle is by giving a method for handling implicit stencils \cite{kloeden1992higher}. This is accomplished by mapping them to mathematically equivalent explicit stencils, as explained in Appendix \ref{ssec:implicit_stencils}.

Supporting implicit stencils is a significant capability in a stencil solver, as implicit stencils can often be designed to be more stable than explicit ones.

\para{Vector-Valued Fields} The third way we extend our algorithms is to handle vector-valued fields, i.e. fields where the grid data for each cell is an array of fixed length instead of being just a scalar. We allow a set of scalar-valued fields $\{a^{(i)}_t\}$ to evolve according to linear stencils as
\begin{align*}
    a^{(i)}_{t+1} = \sum_j S^{(i,j)} a^{(j)}_{t},
\end{align*}
where $S^{(i,j)}$ is a circulant stencil matrix describing how $a_{t+1}^{(i)}$ is dependent on $a_t^{(j)}$. The details are given in Appendix \ref{ssec:vector_valued_fields}. 
It is noteworthy that if there are a constant number of scalar fields $a^{(i)}_t$ making up our grid data, then the complexity of finding the final data is only a constant factor greater than if the field were scalar.

Vector-valued grid data allows us to support \highlight{much larger classes of stencils}, such those that are affine or require data from multiple prior timesteps.

\Arxiv{
\begin{figure}[!ht] 
\centering
\includegraphics[width=0.45\textwidth]{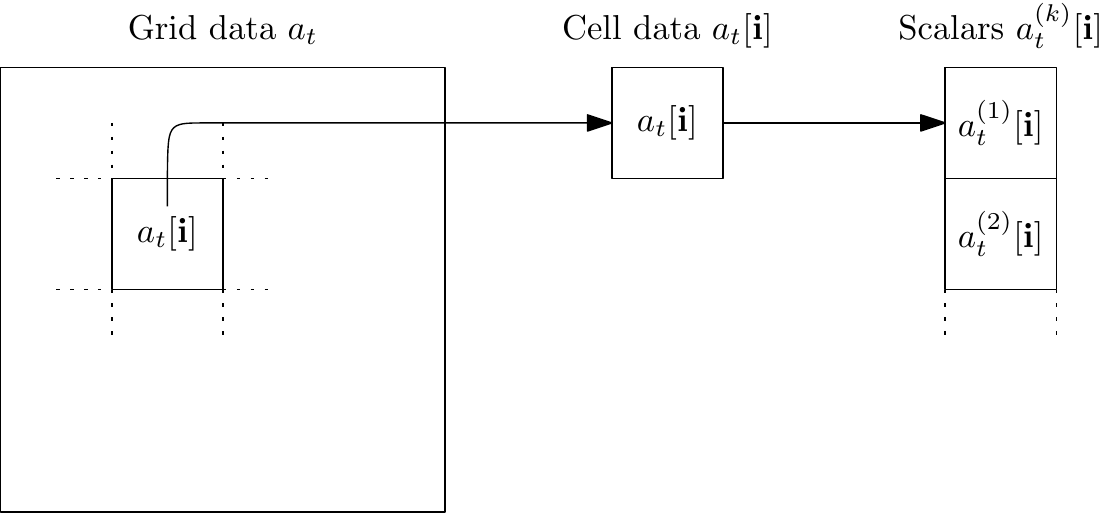}
\vgap{}
\caption{\small Generalized grid data. We allow the data to be evolved to be defined across an arbitrary dimensional grid, indexed by an array $\bm{i} = [i_1, i_2, ...]$. At every cell $a_t[\bm{i}]$ in this grid, there is cell data consisting of a set of scalars $[a_t^{(1)}[\bm{i}], a_t^{(2)}[\bm{i}], ...]$.}
\label{fig:vector-valued-mock}
\end{figure}
}

\section{Aperiodic Stencil Algorithm}
\label{sec:aperiodic}


In this section we will first consider how introducing aperiodicity into our boundary conditions leads to changes in the final data, then present an algorithm for computing the values of the specific cells that are affected by the aperiodic boundary conditions.

\subsection{The Effect of Aperiodicity}
Often numerical computations require boundary conditions that are not periodic \cite{bilbao2013modeling, beggs1992finite, mur1981absorbing, tam1994wall, johnston2002finite}, including well-known classes such as Dirichlet \cite{wang1993finite} and Von Neumann \cite{liao2013high}, as well as more exotic options \cite{strikwerda2004finite}. Indeed, the set of potential aperiodic boundary conditions is extremely diverse. Here we will not attempt to describe \highlight{how} given aperiodic boundary conditions change the final data, but rather \highlight{what sections} of the final data are changed.

\begin{figure}[!t] 
\centering
\includegraphics[width=0.35\textwidth]{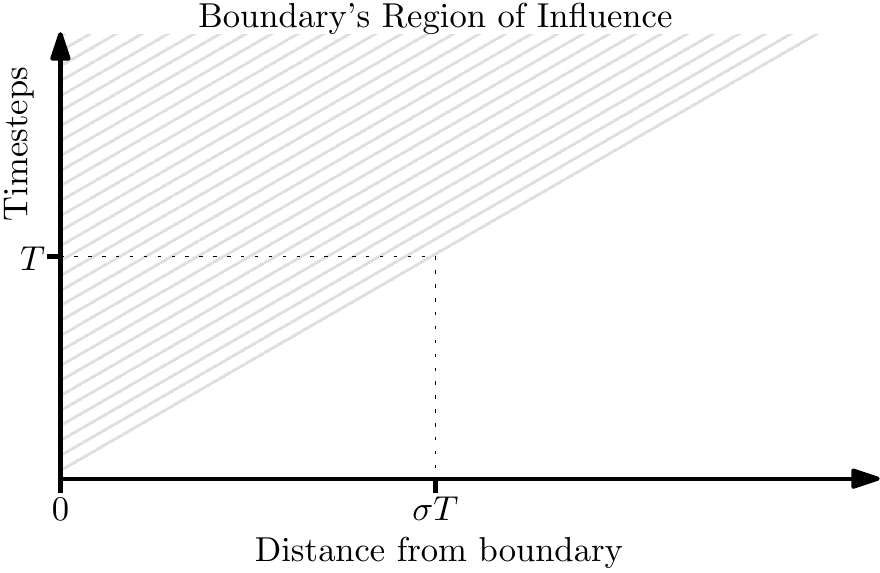}
\vgap{}\vgap{}
\caption{The growth of the boundary's region of influence over time. Here we assume a stencil that uses $\sigma$ cells in each direction. }
\label{fig:roi}
\vgap{}~\\[0.2cm]
\centering
\includegraphics[width=0.4\textwidth]{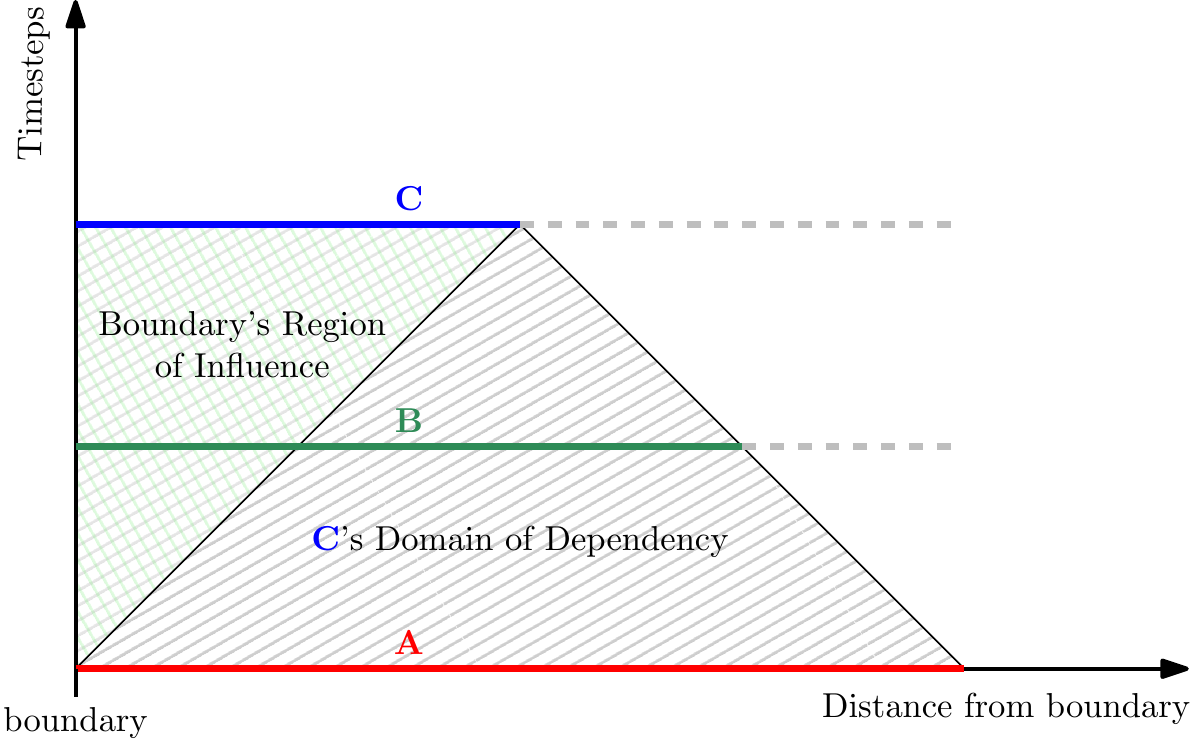}
\vgap{}\vgap{}
\caption{The set of cells relevant to our divide and conquer approach to solving for the boundary's region of influence. We start with set $A$, from which we will compute set $B$ as an intermediate step, after which we will solve for $C$. }
\label{fig:alg_coarse}
\vgap{}~\\[0.2cm]
\centering
\includegraphics[width=0.44\textwidth]{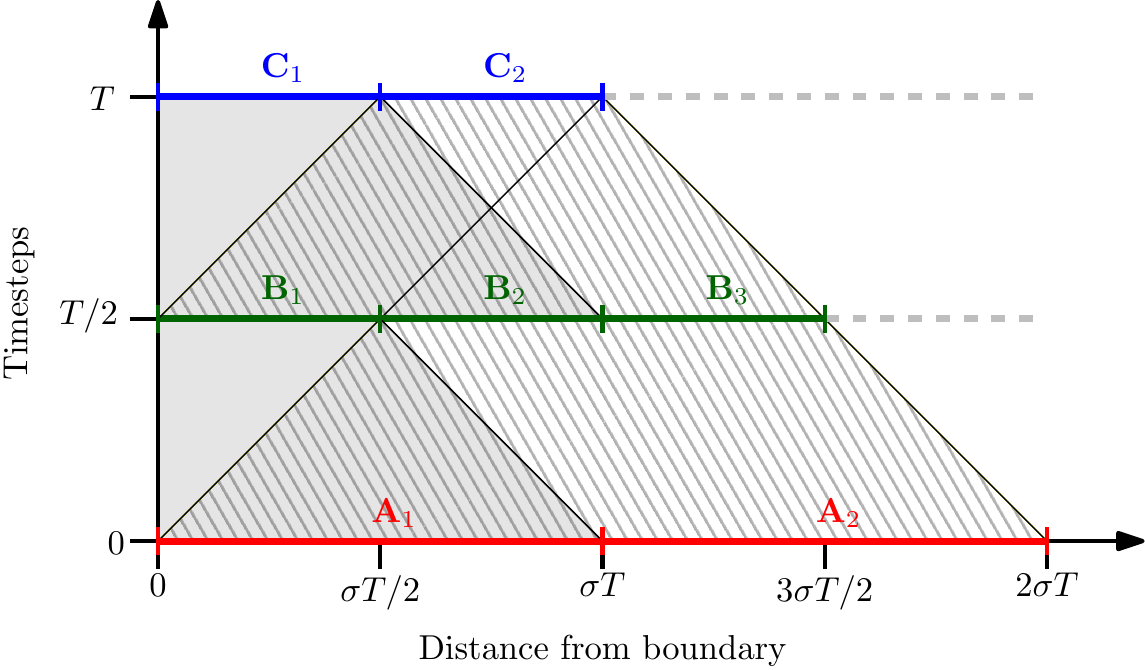}
\vgap{}\vgap{}
\caption{Additional detail on the regions already specified in Figure \ref{fig:alg_coarse}. The dark regions show areas that will be dealt with as {\Boundary} subproblems; the light areas can be handled by {\Periodic}.}
\label{fig:alg_fine}
\vgap{}\vgap{}
\end{figure}

The set of cells that are dependent on values from the grid boundary are called the boundary's \highlight{region of influence (ROI)} \cite{barnes1964technique, van1979towards, chizhonkov2000domain, moretti1979lambda}. It is common for the boundary's region of influence to be smaller than the entire spatial grid, as would be the case in a simulation of a large spatial region for a small time period. A cell's influence on its neighbors is mediated by the stencil, so the larger the stencil is the fewer timesteps it takes for one cell to influence the whole grid.

Consider a stencil with radius $\sigma$, i.e. one which only uses values from a neighborhood extending up to $\sigma$ cells away from its center. After one timestep of a computation based on this stencil, the set of cells which are influenced by aperiodic boundary conditions will be exactly those that are at distance $\sigma$ or less from the grid boundary. After $T$ timesteps, all cells within distance $\sigma T$ of the boundary will be influenced. This suggests that we should visualize how the boundary's region of influence grows by drawing a diagram in two variables (as in Figure \ref{fig:roi}), one being time and the other distance from the boundary. Figures of this type will be of significant use to us while describing our aperiodic algorithm.

We now move to our aperiodic algorithm {\Aperiodic}, at the core of which is {\Boundary}, a divide and conquer technique for correcting the values of the cells in the boundary's region of influence.

\subsection{Correcting Final Values in the Boundary's Region of Influence}

Our aperiodic algorithm breaks up the final data into two regions which will be computed with different methods. The \highlight{interior} region consists of cells whose values will not change no matter what boundary conditions we apply. In particular, periodic boundary conditions would not change the value of these cells. This region will thus be dealt with by using our efficient periodic algorithm. The \highlight{exterior} region is the boundary's region of influence, made up of cells that are dependent on boundary values.

Since we already know how to solve for the interior region, let us turn our attention to the boundary's region of influence, which we solve for with a recursive divide and conquer approach. The core idea here will be to perform a time-cut to reduce the number of cells affected by the aperiodic boundary; by splitting one large step into two smaller steps we will be able to use the periodic solver on more space, at the cost of having to compute cell values at an intermediate time.

Figure \ref{fig:alg_coarse} diagrammatically shows how {\Boundary} solves for the boundary region. First all values in region $A$ are used to solve for region $B$, then region $B$ is used to solve for region $C$. These regions can be much smaller than the entire spatial grid, since even though they wrap around the entirety of the boundary, they are only a fixed number of cells thick.

Suppose again that we are performing computations with a stencil of radius $\sigma$. Then set $C$ is the boundary's region of influence after $T$ timesteps, which includes all cells within distance $\sigma T$ of the boundary. Since we want to be able to fully compute $C$ from $B$ and $B$ from $A$, we must continue to extend our regions as we go $T/2$ steps back for each; $B$ includes everything within $3 \sigma T / 2$ cells of the boundary, and $A$ includes everything out to a distance of $2 \sigma T$. We now refine our naming of these regions, defining subregions $A_1, A_2, B_1, B_2, B_3, C_1$, and $C_2$, whose sizes and locations are shown in Figure \ref{fig:alg_fine}.

The language we have used up to this point has been dimension-free, in that every statement made has been independent of the number of dimensions in the spatial grid we are performing computations on. Now, to make clear the regions shown in Figures \ref{fig:alg_coarse} and \ref{fig:alg_fine}, we will consider what exactly they look like when we make a choice of spatial grid dimension.

In \textbf{1-D}, when the spatial grid $a[0, \dots, N-1]$ is just a linear array of $N$ cells, the boundary lies to the left of $a[0]$ and to the right of $a[N-1]$. The set of all cells within distance $d$ of the 1-D grid's boundary is given by the union of $a[0, \dots, d-1]$ and $a[N-d, \dots, N-1]$. For example, the region $C$ is the union of $a[0, \dots, \sigma T - 1]$ and $a[N-\sigma T, \dots, N-1]$.

In \textbf{2-D} things become more complicated, since the boundary of the $n_1 \times n_2$ spatial grid $a[0..(n_1-1),0..(n_2-1)]$ lies on the outside of the connected set of cells made up of the 1-D subarrays $a[0,0..(n_2-1)]$, $a[n_1 - 1,0..(n_2-1)]$, $a[0..(n_1-1), 0]$, and $a[0..(n_1-1), n_2 - 1]$. The set of all cells within distance $d$ of the 2-D grid's boundary is the union of $a[0..(d-1),0..(n_2-1)]$, $a[(n_1-d)..(n_1-1),0..(n_2-1)]$, $a[0..(n_1-1), 0..(d-1)]$, and $a[0..(n_1-1), (n_2 - d)..(n_2 - 1)]$. Note that we are listing some cells twice here, specifically those in the corners of the grid.

Now we are prepared to describe in detail\footnote{
For brevity here we use the region names as shown in Figure \ref{fig:alg_fine}. The reader who would prefer to see these regions specified directly in terms of their distance from the boundary is referred to Appendix \ref{ssec:ell_shells}.
} the {\Boundary} algorithm for computing the correct values of cells in the boundary's region of influence. Consider again Figure \ref{fig:alg_fine}. The {\Boundary} algorithm solves for two time-slices which are done sequentially. Each time-slice is divided into two distinct regions: one which is solved for with a periodic solver, and one which is solved for recursively. These two regions are handled in parallel.
\begin{enumerate}
    \item \highlight{Solving for $\bm{B}$ from $\bm{A}$.} We will begin by feeding data from region $A = A_1 \cup A_2$ into our periodic solver to find values for region $B_2 \cup B_3$. For a $d$-D grid this will take $2d$ calls to the periodic solver. At the same time (in parallel), a recursive call is made to the {\Boundary} algorithm with $A_1$ as input, writing the output to $B_1$. This completes all the values of $B$.
    \item \highlight{Solving for $\bm{C}$ from $\bm{B}$.} Now $B = B_1 \cup B_2 \cup B_3$ is fed into our periodic algorithm to find $C_2$, while in parallel we feed $B_1 \cup B_2$ into another recursive call to {\Boundary}, which will find the values for $C_1$. Thus all cell values in $C = C_1 \cup C_2$ are computed.
\end{enumerate}
The base cases in recursion occur when the only cells to be computed are within a constant distance of the boundary. In this case a standard looping algorithm is applied. Theoretically we can set the cutoff distance to any constant greater than or equal to $\sigma$, but in practice the constant is chosen such that the cost of computing the base case is balanced with the cost of further recursion. See Figure \ref{fig:stencil-fft-algorithm-1D-aperiodic} for pseudocode.

Combining our periodic solver for the interior region with this recursive solver for the boundary region gives our aperiodic algorithm {\Aperiodic}. A diagrammatic outline of {\Aperiodic} is given in Figure \ref{fig:stencil-fft-algorithm-1D-aperiodic}, and pseudocode is given in Figure \ref{fig:stencil-fft-algorithm-1D-aperiodic}.

\begin{theorem}
\label{th:aperiodic-stencil}
The {\Aperiodic} algorithm can compute the $T$th timestep of a stencil computation on a grid of size $N$ with aperiodic boundary conditions in $\LTh{bT \log (bT) \log T}$ $\RTh{+~N \log N}$ work and $\Th{T \log b + \log N \log \log N}$ span, where $b$ is the number of grid cells defined by boundary conditions.
\end{theorem}

The proof is presented in Appendix~\ref{ssec:proof-aperiodic-stencil}.

We believe that the bounds we achieve here for work and span are near-optimal (within a polylogarithmic factor) for fully general aperiodic stencil problems. This is because specific choices of nonlinear boundary conditions in combination with linear stencils can result in arbitrary cellular automata being embedded into the boundary of the spatial grid.    These automata can be computationally universal (any computation can be mapped to them in a way that preserves time and space complexity), hence the $\Th{T N^{1 - 1/d}}$ (size of the space-time boundary) component in our work complexity and the $\Th{T}$ component in our span complexity.

\begin{figure}[!t]
\centering
\begin{minipage}{0.47\textwidth}
\begin{mycolorbox}{$\Aperiodic(s, \sigma, a_0, \ell_0, \dots, \ell_d, T)$}
\begin{minipage}{0.97\textwidth}
{\scriptsize
\algotopspace{}
\noindent
\begin{enumerate}
\setlength{\itemindent}{-2em}

\vsitem $\Delta \gets \sigma T / 2$
\vsitem $\text{result} \gets \text{Array of size $\ell_1 \times \cdots \times \ell_d$}$
\vsitem \textbf{parallel (1):} \xcomment{Interior}
\vsitem \T $\text{center} \gets \Periodic(s, a_0, \ell_0, \dots, \ell_d, T)$
\vsitem \T $\text{result}[2 \Delta < \mathbf{dist}] \gets \text{center}[2 \Delta < \mathbf{dist}]$
\vsitem \textbf{parallel (2):} \xcomment{Boundary}
\vsitem \T $\text{boundary} \gets \Boundary(s, \sigma, a_0, \ell_0, \dots, \ell_d, T)$
\vsitem \T $\text{result}[0 < \mathbf{dist} \leq 2 \Delta] \gets \text{boundary}$

\algobottomspace{}
\end{enumerate}
}
\end{minipage}
\end{mycolorbox}
\vgap{}
\begin{mycolorbox}{$\Boundary(s, \sigma, a_0, \ell_0, \dots, \ell_d, T)$}
\begin{minipage}{0.97\textwidth}
{\scriptsize
\algotopspace{}
\noindent
\begin{enumerate}
\setlength{\itemindent}{-2em}

\vsitem $\Delta \gets \sigma T / 2$
\vsitem[] [\textbf{Base Case.}] \hrulefill
\vsitem \xif $T < $ cutoff \xthen
\vsitem \T Iteratively solve for cells within $2\Delta$ of the boundary at time $T$
\vsitem \T \xreturn
\vsitem[] [\textbf{Step 1.}] \hrulefill
\vsitem \textbf{parallel (1):} \xcomment{Interior}
\vsitem \T $A_{12} \gets a_0[0 < \mathbf{dist} \leq 4\Delta]$
\vsitem \T $B_{23} \gets \Periodic(s, A_{12}, \ell_1, \dots, \ell_d, T/2)$
\vsitem \T $\text{result}[\Delta < \mathbf{dist} \leq 3 \Delta] \gets B_{23}[\Delta < \mathbf{dist} \leq 3 \Delta]$
\vsitem \textbf{parallel (2):} \xcomment{Boundary}
\vsitem \T $A_{1} \gets a_0[0 < \mathbf{dist} < 2\Delta]$
\vsitem \T $B_{1} \gets \Boundary(s, A_{1}, \ell_1, \dots, \ell_d, T/2)$
\vsitem \T $\text{result}[0 < \mathbf{dist} \leq \Delta] \gets B_{1}$
\vsitem $B_{123} \gets \text{result}[0 < \mathbf{dist} \leq 3\Delta]$
\vsitem[] [\textbf{Step 2.}] \hrulefill
\vsitem \textbf{parallel (1):} \xcomment{Interior}
\vsitem \T $C_{2} \gets \Periodic(s, B_{123}, \ell_1, \dots, \ell_d, T/2)$
\vsitem \T $\text{result}[\Delta < \mathbf{dist} \leq 2 \Delta] \gets C_{2}[\Delta < \mathbf{dist} \leq 2 \Delta]$
\vsitem \textbf{parallel (2):} \xcomment{Boundary}
\vsitem \T $B_{12} \gets B[0 < \mathbf{dist} \leq 2\Delta]$
\vsitem \T $C_{1} \gets \Boundary(s, B_{12}, \ell_1, \dots, \ell_d, T/2)$
\vsitem \T $\text{result}[0 < \mathbf{dist} \leq \Delta] \gets C_{1}$
\vsitem $C_{12} \gets \text{result}[0 < \mathbf{dist} \leq 2\Delta]$

\algobottomspace{}
\end{enumerate}
}
\end{minipage}
\end{mycolorbox}
\vgap{}\vgap{}
\caption{The {\Aperiodic} algorithm and the {\Boundary} subroutine. The \textbf{parallel} keyword is used here to mark blocks of code which are run on separate processes. Throughout both listings we use a dimension-free indexing notation where $A[a < \mathbf{dist} \leq b]$ represents the set of cells in $A$ that have distance to the boundary in the range $(a,b]$.}
\label{fig:stencil-fft-algorithm-1D-aperiodic}
\end{minipage}
\end{figure}

\newcommand{\insertplotline}[3]{{\includegraphics[width=0.32\textwidth, height=3.7cm, clip = true]{#1}} & {\includegraphics[width=0.32\textwidth, height=3.7cm, clip = true]{#2}} & {\includegraphics[width=0.32\textwidth, height=3.7cm, clip = true]{#3}} \\ }

\section{Experiments}
In this section, we present the experimental evaluation of our algorithms as compared with the state-of-the-art stencil codes. Our experimental setup is shown in Table \ref{tab:experimental-setup}.

\begin{table}[!ht]
\centering
\scalebox{0.8}{
\begin{colortabular}{ |l | l | l|}
\hline                       
 \rowcolor{white}\cellcolor{tabletitlecolor} & \cellcolor{tabletitlecolor}Cores & 68 cores per socket, 1 socket (total: 68 threads) \\
\cellcolor{tabletitlecolor} & \cellcolor{tabletitlecolor}Cache sizes & L1 32 KB, L2 1 MB, L3 16 GB (shared)\\
\cellcolor{tabletitlecolor}\multirow{-3}{*}{\rotatebox{90}{KNL}} & \cellcolor{tabletitlecolor}Memory & 96 GB DDR RAM\\ \hline
\rowcolor{white}\cellcolor{tabletitlecolor}&\cellcolor{tabletitlecolor}Cores & 24 cores per socket, 2 sockets (total: 48 cores) \\
\cellcolor{tabletitlecolor}&\cellcolor{tabletitlecolor}Cache sizes & L1 32 KB, L2 1 MB, L3 33 MB\\
\cellcolor{tabletitlecolor}\multirow{-3}{*}{\rotatebox{90}{SKX}}&\cellcolor{tabletitlecolor}Memory & 144GB /tmp partition on a 200GB SSD\\ \hline
\multicolumn{2}{|l|}{\cellcolor{tabletitlecolor}Compiler} & Intel C++ Compiler (ICC) v18.0.2\\
\multicolumn{2}{|l|}{\cellcolor{tabletitlecolor}Compiler flags} & \texttt{-O3 -xhost -ansi-alias -ipo -AVX512}\\
\multicolumn{2}{|l|}{\cellcolor{tabletitlecolor}Parallelization} & OpenMP 5.0\\
\multicolumn{2}{|l|}{\cellcolor{tabletitlecolor}Thread affinity} & \texttt{GOMP\_CPU\_AFFINITY}\\
\end{colortabular}
}
\tabcaption{\small Experimental setup on the Stampede2 Supercomputer \cite{Stampede2} using Knights Landing (KNL) Intel Xeon Phi 7250 and Skylake (SKX) Intel Xeon Platinum 8160 nodes.}
\label{tab:experimental-setup}
\vgap{}\vgap{}\vgap{}\vgap{}
\end{table}

\hide{
\paragraph{Experimental Setup.}
All experiments were performed on Knights Landing (KNL) Intel Xeon Phi 7250 nodes of Stampede2 Supercomputer \cite{Stampede2}. Each node has 68 cores and 272 hardware threads (4 threads per core) on a single socket. Each KNL node has 32 KB L1 data cache per core, 1 MB L2 cache per two-core tile, 16 GB direct-mapped L3 cache, and a 96 GB DDR4 RAM.

All algorithms were implemented in C++. We used Intel C++ Compiler (ICC) v18.0.2 with \texttt{gnu++98} standard to compile our implementations with the optimization parameters \texttt{-O3 -xhost -ansi-alias -ipo -AVX512}. We used OpenMP 5.0 with ICC for our shared-memory parallel implementations. The threads were pinned to CPU cores using \texttt{GOMP\_CPU\_AFFINITY}.}

\vspace{0.2cm}
\noindent
\textbf{Benchmarks and Numerical Accuracy.} For benchmarks we use a variety of stencil problems including those with periodic and aperiodic boundary conditions and across 1, 2, and 3 dimensions. Our test stencils, primarily drawn from \cite{Rawat2018} and listed with details in Table \ref{tab:benchmarks}, are the following: \texttt{heat1d}, \texttt{heat2d}, \texttt{seidel2d}, \texttt{jacobi2d}, \texttt{heat3d}, and \texttt{19pt3d} (called \texttt{poisson3d} in \cite{Rawat2018}). We test two primary aspects of our algorithms: numerical accuracy and computational complexity.

To evaluate numerical accuracy we use max relative error against analytical solutions for the heat equation in 1, 2, and 3 dimensions. This is shown in Table \ref{tab:numerical_accuracy} against a na\"ive iterative looping implementations which is numerically equivalent to PLuTo. We see that our algorithms show no significant difference in loss from floating point accuracy when compared against standard looping codes.

\begin{table}[!ht]
\centering
\scalebox{0.8}{
\begin{colortabular}{ c | l | r | c | c|}
\hline                       

Dim & Benchmark & Stencil points $(\alpha)$ & Stencil radius $(\sigma)$ \\ \hline

1D & \texttt{heat1d} & 3pts & 1 \\\hline
2D & \texttt{heat2d} & 5pts & 1 \\
& \texttt{seidel2d} & 9pts & 1 \\
& \texttt{jacobi2d} & 25pts & 2 \\\hline
3D & \texttt{heat3d} & 7pts & 1 \\
& \texttt{19pt3d} & 19pts & 2 
\end{colortabular}
}

\tabcaption{Benchmark problems with the number of points in the corresponding stencils.
}
\label{tab:benchmarks}
\vgap{}\vgap{}\vgap{}\vgap{}
\end{table}

\begin{table}[!ht]
\centering
\scalebox{0.75}{
\begin{colortabular}{ l | l | l || l | l|}
\hline                       


Stencil & Grid size & Timesteps & Our algorithm & Looping code \\ \hline

\texttt{heat1d} & $1,000$ & $10^6$ & $5.71632\times 10^{-6}$& $5.71637\times 10^{-6}$\\
\texttt{heat2d} & $500\times500$ & $2.5 \times 10^5$& $2.73253\times 10^{-5}$ & $2.73253\times10^{-5}$\\
\texttt{heat3d} & $200\times200\times200$& $4 \times 10^4$& $1.72981\times10^{-4}$ & $1.72981\times10^{-4}$\\
\end{colortabular}
}

\tabcaption{\small Numerical accuracy comparison between our algorithms and looping code. Our analytical solutions were chosen so that the truth values fell within $[0.5,2]$ everywhere in the solution domain. }
\label{tab:numerical_accuracy}
\vgap{}\vgap{}\vgap{}\vgap{}
\end{table}



\begin{table}
\centering
\scalebox{0.59}{
\begin{colortabular}{ | l | l | l | | l l | r  r | r r | r r |}
\hline                       

\multicolumn{5}{|l|}{\begin{tabular}{@{}p{2.3in}@{}}Benchmark\\\hline\end{tabular}} & \multicolumn{4}{c|}{\begin{tabular}{@{~}c@{~}}Parallel runtime in seconds\\\hline\end{tabular}} & \multicolumn{2}{c|}{Speedup factor}\\

\rowcolor{tabletitlecolor} \multicolumn{2}{|l|}{} & &  &  & \multicolumn{2}{c|}{\begin{tabular}{@{~}c@{~}}\pluto{}\\\hline\end{tabular}} & \multicolumn{2}{c|}{\begin{tabular}{@{~}c@{~}}Our algorithm\\\hline\end{tabular}} & \multicolumn{2}{c|}{\begin{tabular}{@{~}c@{~}}over \pluto{}\\\hline\end{tabular}} \\

\rowcolor{tabletitlecolor}\multicolumn{2}{|l|}{} & Stencil & $N$ & $T$ & KNL & SKX & KNL & SKX & KNL & SKX \\ \hline

\multicolumn{2}{|c|}{\multirow{6}{*}{\rotatebox{90}{Periodic}}} & \texttt{heat1d} & $1,600,000$ & $10^6$ & 79 & 19 & 0.25 & 0.03 & 1754.7 & 759.6\\
\multicolumn{2}{|l|}{}& \texttt{heat2d} & $8,000 \times 8,000$ & $ 10^5$ & 1,437 & 222 & 0.48 & 0.61 & 3,025.0 & 367.0\\
\multicolumn{2}{|l|}{}& \texttt{seidel2d} & $8,000 \times 8,000$ & $ 10^5$ & 500 & 808 & 0.48 & 0.64 & 1,032.7 & 1268.6\\ 
\multicolumn{2}{|l|}{}& \texttt{jacobi2d} & $8,000 \times 8,000$ & $ 10^5$ & 2,905 & 1017 &  0.48 & 0.68 & 6,084.7 & 1502.1\\
\multicolumn{2}{|l|}{}& \texttt{heat3d} & $800 \times 800 \times 800$ & $10^4$ & 816 & 1466 & 4.98 & 5.48 & 163.9 & 267.3\\
\multicolumn{2}{|l|}{}& \texttt{19pt3d} & $800 \times 800 \times 800$ & $10^4 $ & 141 & 158 & 4.84 & 5.78 & 29.1 & 27.3\\
\hline

\multirow{12}{*}{\rotatebox{90}{Aperiodic}} & \multirow{6}{*}{\rotatebox{90}{Experiment 1}} & \texttt{heat1d} & $1,600,000$ & $10^6$ & 50 & 35 & 5.85 & 6.69 & 8.5 & 5.2\\
&& \texttt{heat2d}& $8,000 \times 8,000$ & $ 10^5$ & 333 & 530 & 143.25 & 151.37 & 2.3 & 3.5 \\
&& \texttt{seidel2d}& $8,000 \times 8,000$ & $ 10^5$ & 345 & 601 & 145.42 & 132.97 & 2.4 & 4.5\\
&& \texttt{jacobi2d}& $8,000 \times 8,000$ & $ 10^5$ & 567 & 456 & 249.04 & 273.46 & 2.3 & 1.7\\
&& \texttt{heat3d}& $800 \times 800 \times 800$ & $10^4 $ & 513 & 763 & 395.10 & 605.89 & 1.3 & 1.3\\
&& \texttt{19pt3d}& $800 \times 800 \times 800$ & $10^4 $ & 645 & 848 & 425.22 & 616.71 & 1.5 & 1.4\\\cline{2-11}

 &\multirow{6}{*}{\rotatebox{90}{Experiment 2}}& \texttt{heat1d} & $1,600,000$ & $N$ & 32 & 23 & 5.63 & 6.87  & 5.7 & 3.3\\
&& \texttt{heat2d} & $8,000 \times 8,000$ & $ \sqrt{N}$ & 210 &  312 & 92.78 & 121.70 & 2.3 & 2.6\\
&& \texttt{seidel2d} & $8,000 \times 8,000$ & $ \sqrt{N}$ & 228 & 375 & 91.59 & 121.46 & 2.5 & 3.1\\
&& \texttt{jacobi2d} & $8,000 \times 8,000$ & $ \sqrt{N}$ & 372 & 281 & 151.31 & 198.00 & 2.5 & 1.4\\
&& \texttt{heat3d} & $800 \times 800 \times 800$ & $\sqrt[3]{N}$ & 45 & 71 & 32.29 & 50.52 & 1.4 & 1.4\\
&& \texttt{19pt3d} & $800 \times 800 \times 800$ & $\sqrt[3]{N}$ & 61 & 71 & 33.82 & 52.27 & 1.8 & 1.4

\end{colortabular}
}
\caption{Performance summary of parallel stencil algorithms on a KNL/SKX node.}
\vspace{-0.3cm}
\label{tab:experimental-results}
\vgap{}\vgap{}\vgap{}
\end{table}

\vspace{0.2cm}
\noindent
\textbf{PLuTo-Generated Stencil Programs.} The tiled looping implementations were generated by \pluto{} \cite{Pluto} -- the state-of-the-art tiled looping code generator. The two main types of tiling methods used for performance comparison are: \plutostandard{} and \plutodiamond{}, whose tiles have the shapes of parallelograms and diamonds, respectively. In the plots, we use diamond and square symbols to denote \plutodiamond{} and \plutostandard{}, respectively.  \hide{The \plutostandard{} algorithm exploits better data locality and achieves better scaling than \plutostandard{} \cite{bondhugula2017}.} These parallel implementations were run on 68-core KNL and 48-core SKX nodes. The tile sizes were selected via an autotuning phase, exploring sizes from $\{ 8, 16, 32\}$ for the outer dimensions and from $\{64, 128, 512\}$ for the inner-most dimension to ensure that enough vectorization and multithreaded parallelism were exposed by \pluto{}, while ensuring the tile footprint neared the cache size.

\vspace{0.2cm}
\noindent
\textbf{Our FFT-Based Stencil Programs.}
Implementations of our FFT-based algorithms use FFT implementations available in the Intel Math Kernel Library (Intel MKL) \cite{MKL}. In the plots, we use the triangle symbol to represent our FFT-based implementations. \hide{The base case sizes were chosen from $\{ 8, 16,32,64, 128\}$.} The base case sizes used for \texttt{1d}, \texttt{2d}, \texttt{3d} are $128$, $64\times64$, $16\times16\times16$, respectively. 

\hide{
\paragraph{Number of Timesteps.} In the literature, stencil computations are usually performed for a few hundreds of timesteps. In our experiments, we vary the number of timesteps $T$ from a few thousands to very large values. The motivation for simulating stencil computations for large number of timesteps is as follows: $(i)$ the state of a big physical system (e.g. fluid dynamics) might change quickly which necessitates simulations over a large number of tiny timesteps, $(ii)$ quick simulations of stencil computations over large number of timesteps can help us to iteratively improve our initially chosen stencil to realistically model the real-world, and $(iii)$ rules (or stencils) in problems such as \textit{game of life} will not change over time and hence can be used for simulating a large number of timesteps. 
}

\subsection{Periodic Stencil Algorithms}
Figure \ref{fig:plots-periodic-stencil-algorithms} shows runtime, speedup (w.r.t. PLuTo), and scaling plots for our FFT-based periodic stencil algorithm on Intel KNL nodes with Figure \ref{fig:appendix-plots-periodic-stencil-algorithms} in the Appendix showing additional runtime plots. Figure \ref{fig:appendix-plots-periodic-stencil-algorithms} also includes the same set of plots for SKX nodes. We identify the implementations of our algorithms in the plots by prefixing the stencil name with ``FFT'' while PLuTo-generated implementations are identified by a ``PLuTo'' prefix.

For 1D, 2D, and 3D stencils we keep the value of $N$ fixed to 1.6M, $8K \times 8K$, and $800 \times 800 \times 800$, respectively, and vary $T$, where $1$M$ =10^6$ and $1$K$ =10^3$. 

In all of our 1D and 2D periodic stencil experiments \plutodiamond{} ran faster than \plutostandard{} while in case of 3D \plutostandard{} outperformed \plutodiamond{}. When $N$ is fixed, performance of our algorithm improved over PLuTo's as $T$ increased, significantly outperforming PLuTo for large $T$, e.g., for \texttt{seidel2d} our algorithm ran around $6000\times$ faster on KNL when $T = 10^5$. This increase in speedup with the increase of $T$ follows from theoretical predictions. Indeed, theoretical speedup of our algorithm over any existing stencil algorithms is $\Th{T/\log T}$ when $N$ is fixed (see Table \ref{tab:summary}). Hence, the speedup of our algorithm w.r.t. PLuTo-generated codes increased almost linearly with $T$ when $N$ was kept fixed.

Figures \ref{fig:plots-periodic-stencil-algorithms}$(v)$ and \ref{fig:appendix-plots-periodic-stencil-algorithms}$(xi)$ show the scalability of our algorithm on KNL and SKX nodes, respectively, when the number of threads is varied. Our implementations are highly parallel and should scale accordingly. However, we use FFT computations that are memory bound -- they perform only $\Th{N\log{N}}$ work on an input of size $\Th{N}$ and thus have very little data reuse\footnote{which is evident from their (optimal) $\Oh{(N/B)\log_{M}{N}}$ cache complexity with a low temporal locality}. 
We believe that as a result of this issue, our programs do not scale well beyond 32 threads on KNL and 16 threads on SKX. Indeed, we observe that the FFT and the inverse FFT computations are the scalability bottlenecks of our algorithms.

\subsection{Aperiodic Stencil Algorithms}
We performed two types of experiments for aperiodic stencils: $(1)$ grid size $N$ was kept fixed while time $T$ was varied, and $(2)$ grid size was set to $N^{1/d} \times  \cdots \times N^{1/d}$ and $T=N^{1/d}$ for $d$ dimensions, and $N$ was varied.

Figure \ref{fig:plots-aperiodic-stencil-algorithms} shows the runtime (w.r.t. PLuTo) plots of our aperiodic stencil algorithm for both experiments on KNL.
\hide{Figure \ref{fig:plots-aperiodic-knl} shows the plots for aperiodic stencil algorithms for KNL nodes} 
Figure \ref{fig:appendix-plots-aperiodic-knl} in the Appendix includes additional runtime plots for both experiments on KNL. Figures \ref{fig:appendix-plots-aperiodic-skx-1} and \ref{fig:appendix-plots-aperiodic-skx-2} in the Appendix show the corresponding plots on SKX for Experiments 1 and Experiment 2, respectively.

In Experiment 1, \plutodiamond{} outperformed \plutostandard{} for all stencils except for \texttt{heat1d} on both machines. Our algorithm always ran faster than PLuTo-generated code, reaching speedup factors of 8.5, 2.3--2.4, and 1.3--1.5 for 1D, 2D, and 3D stencils, respectively, on KNL (see Table \ref{tab:experimental-results} for details). The corresponding figures on SKX were 5.2, 1.7--4.5, and 1.3--1.4, respectively. Theoretical bounds in Table \ref{tab:summary} imply that our algorithm will run around $\Th{N^{1/d} / ( \log{TN^{1-1/d}})} \log{T}$ factor faster than PLuTo code for any given $N$ and $T$. So, for a fixed $N$, the speedup factor will not increase (may even slightly decrease) with the increase of $T$. The speedup plots match this prediction.

In Experiment 2, \plutodiamond{} ran faster than \plutostandard{} for all stencils except for \texttt{heat1d} and \texttt{heat3d} on KNL and \texttt{heat1d} on SKX. Our algorithm ran up to 5.7, 2.3--2.5, and 1.4--1.8 factor faster than PLuTo-generated code for 1D, 2D, and 3D stencils, respectively, on KNL (see Table \ref{tab:experimental-results}). The corresponding speedup factors on SKX were 3.3, 1.4--2.6, and 1.4, respectively. Our theoretical prediction for rough speedup factor from the previous paragraph implies that our speedup over PLuTo code will increase with the increase of $N$ which is confirmed by the speedup plots for this experiment. However, the speedup plot of \texttt{heat2d}, \texttt{seidel2d}, and \texttt{jacobi2d} on KNL, as shown in Figure \ref{fig:plots-aperiodic-stencil-algorithms}$(xi)$ has speedup drops at $N=9000\times9000$ and $N=15000\times15000$. Similar performance drops are also observed on SKX nodes (see Figure \ref{fig:appendix-plots-aperiodic-skx-2}$(viii)$ in the Appendix). We believe that this happens mainly because of a known phenomenon which is the drastic performance variations MKL suffers from when the sizes of the spatial grid dimensions change \cite{Khokhriakov2018}. This performance drop is also partly due to the changes in the base case kernel size of our implementations resulting from the changes in the grid size. 

Figure \ref{fig:plots-scalability-aperiodic-knl-exp2} and Appendix Figure \ref{fig:appendix-plots-scalability-aperiodic-skx-exp2} show the scalability plots of our FFT-based aperiodic stencil algorithm on KNL and SKX nodes, respectively. We used $N = 1M$, $N = 16K \times 16K$, and $N=800\times800\times800$ for 1-D, 2-D, and 3-D stencils, respectively, and set $T = N^{1/d}$ for our scalability analysis, where $d$ is the number of dimensions. Our implementations show highly scalable performance on KNL and almost similar scalability for 1D and 2D stencils on SKX.

\begin{figure}[ht!]
\includegraphics[width=0.32\textwidth, height=3.7cm, clip = true]{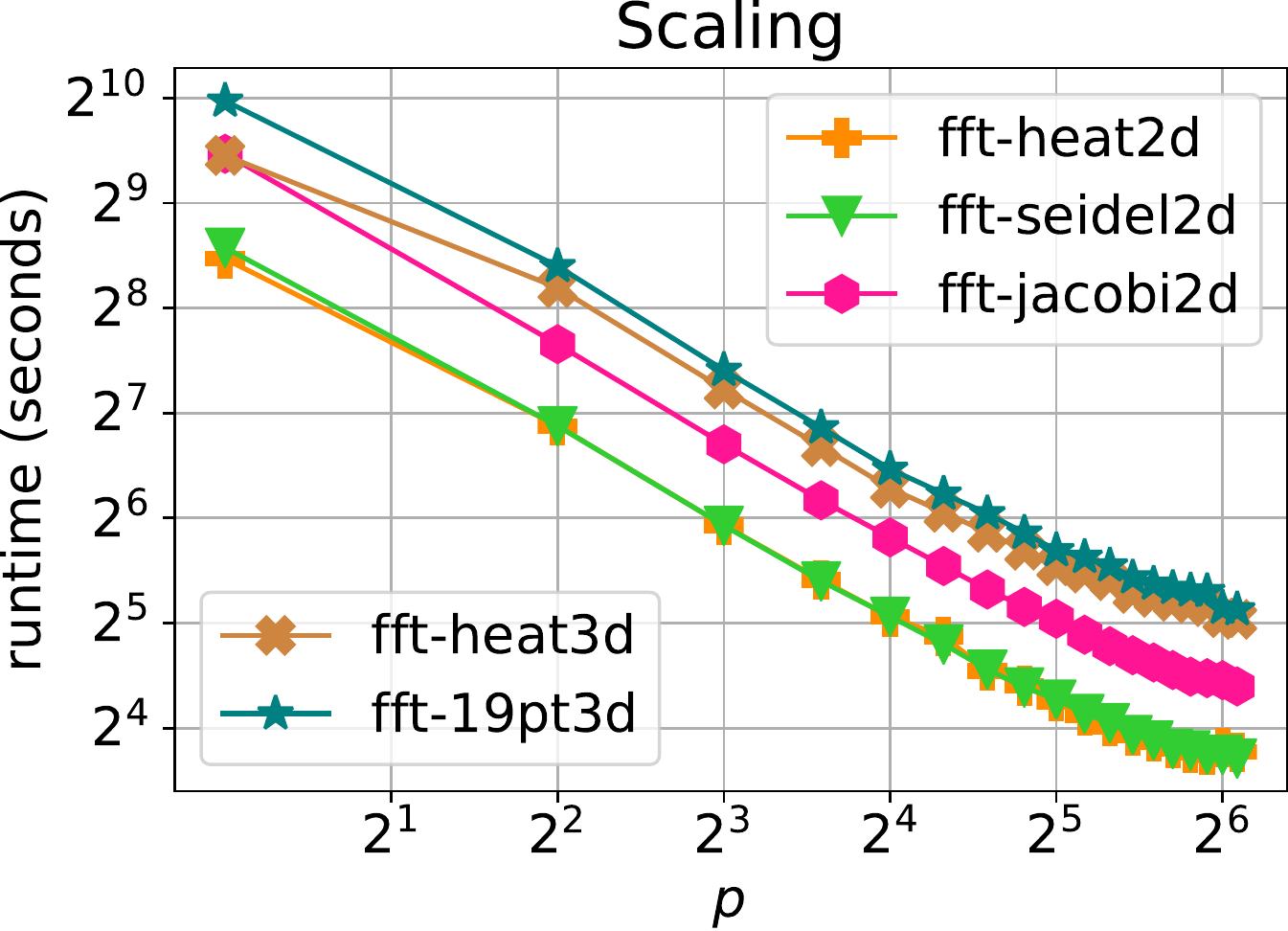}
\vgap{}
\caption{Scalability of our aperiodic algorithms for Experiment 2 on a KNL node.}
\label{fig:plots-scalability-aperiodic-knl-exp2}
\vgap{}\vgap{}\vgap{}
\end{figure}

\begin{table*}[!ht]
\begin{tabular}{c c c}
\hline\\[-1.5ex]

\multicolumn{3}{c}{\textbf{KNL Node (More plots are given in Figure \ref{fig:appendix-plots-periodic-stencil-algorithms})}} \\[0.5ex]

\insertplotline{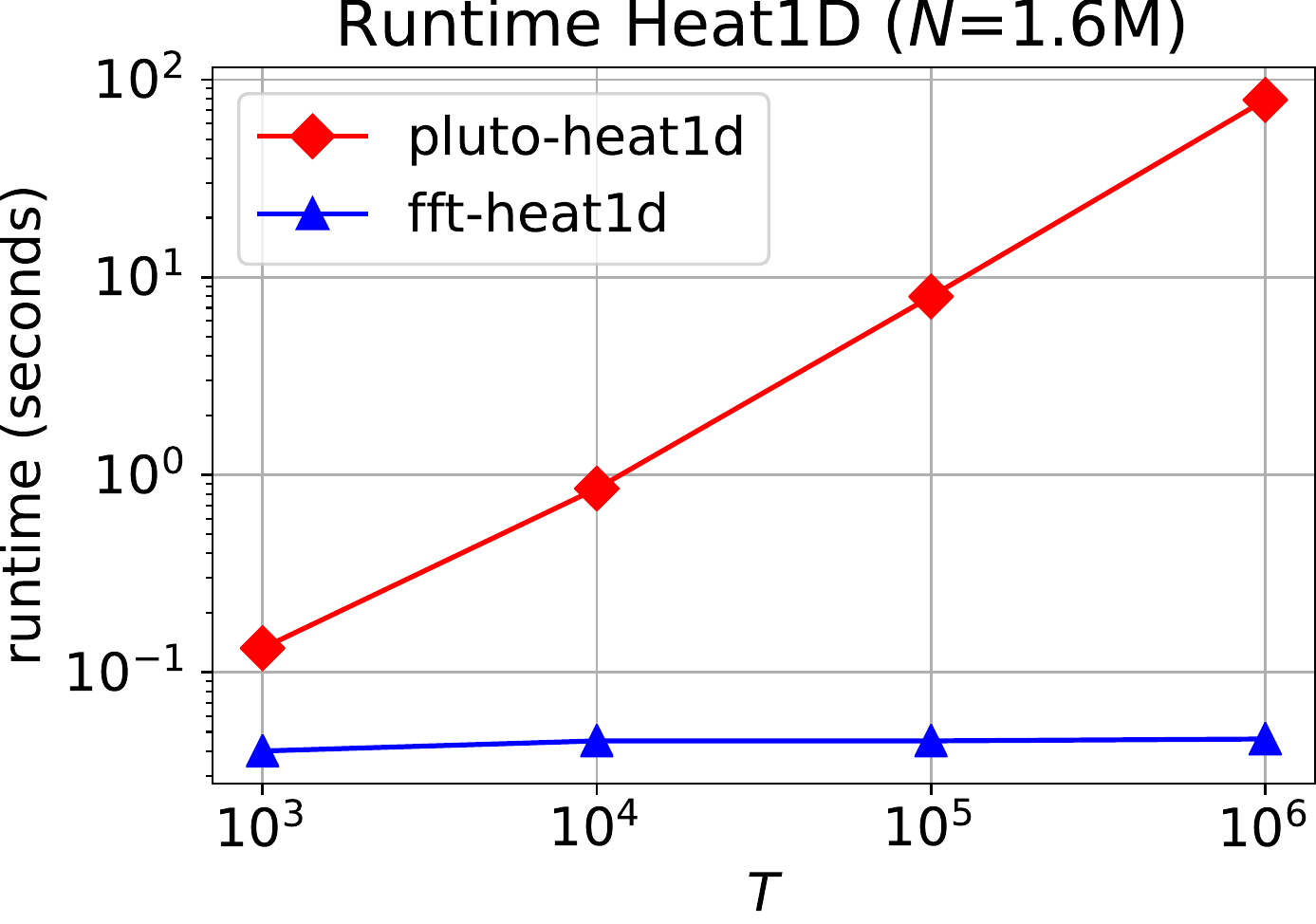}{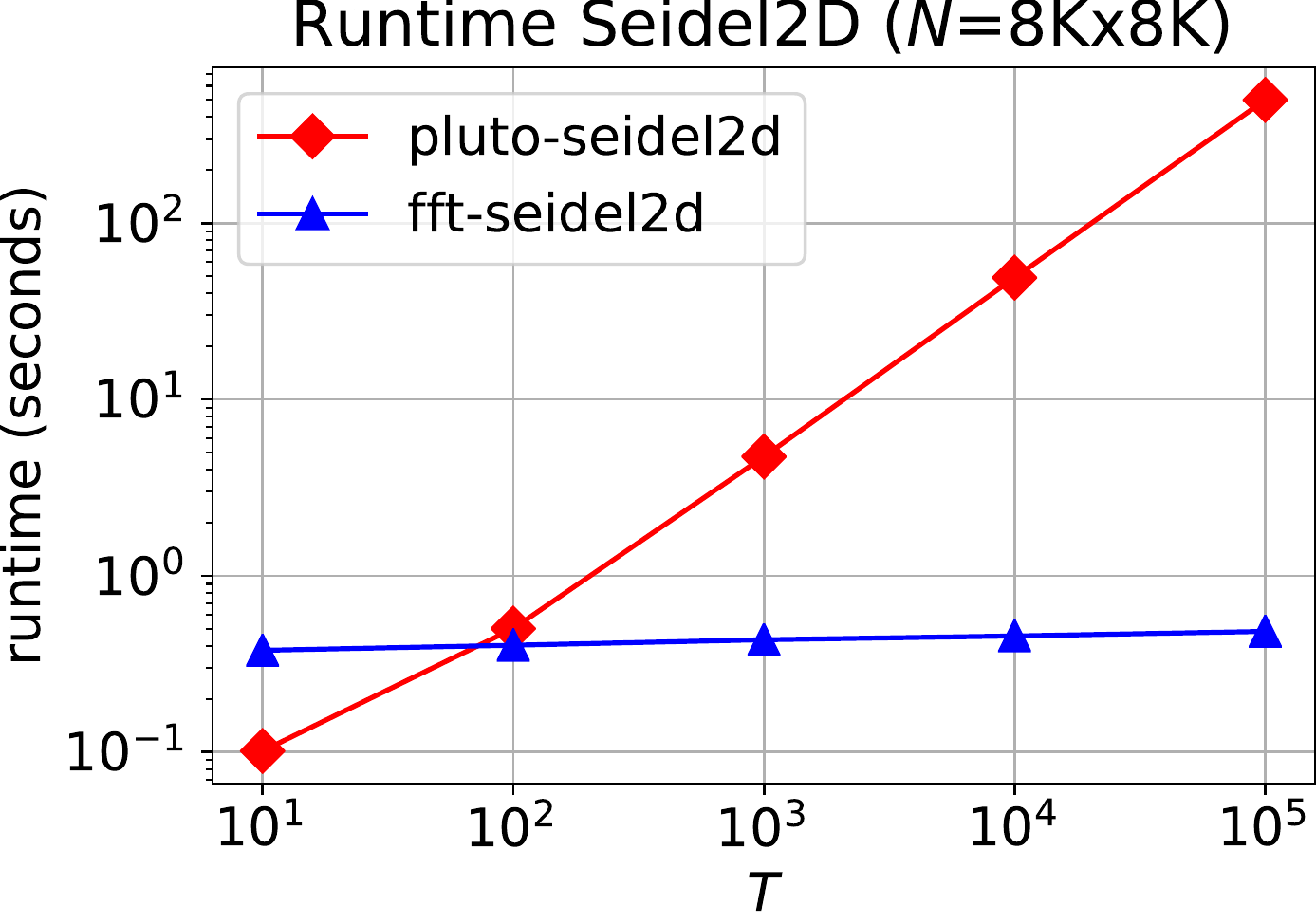}{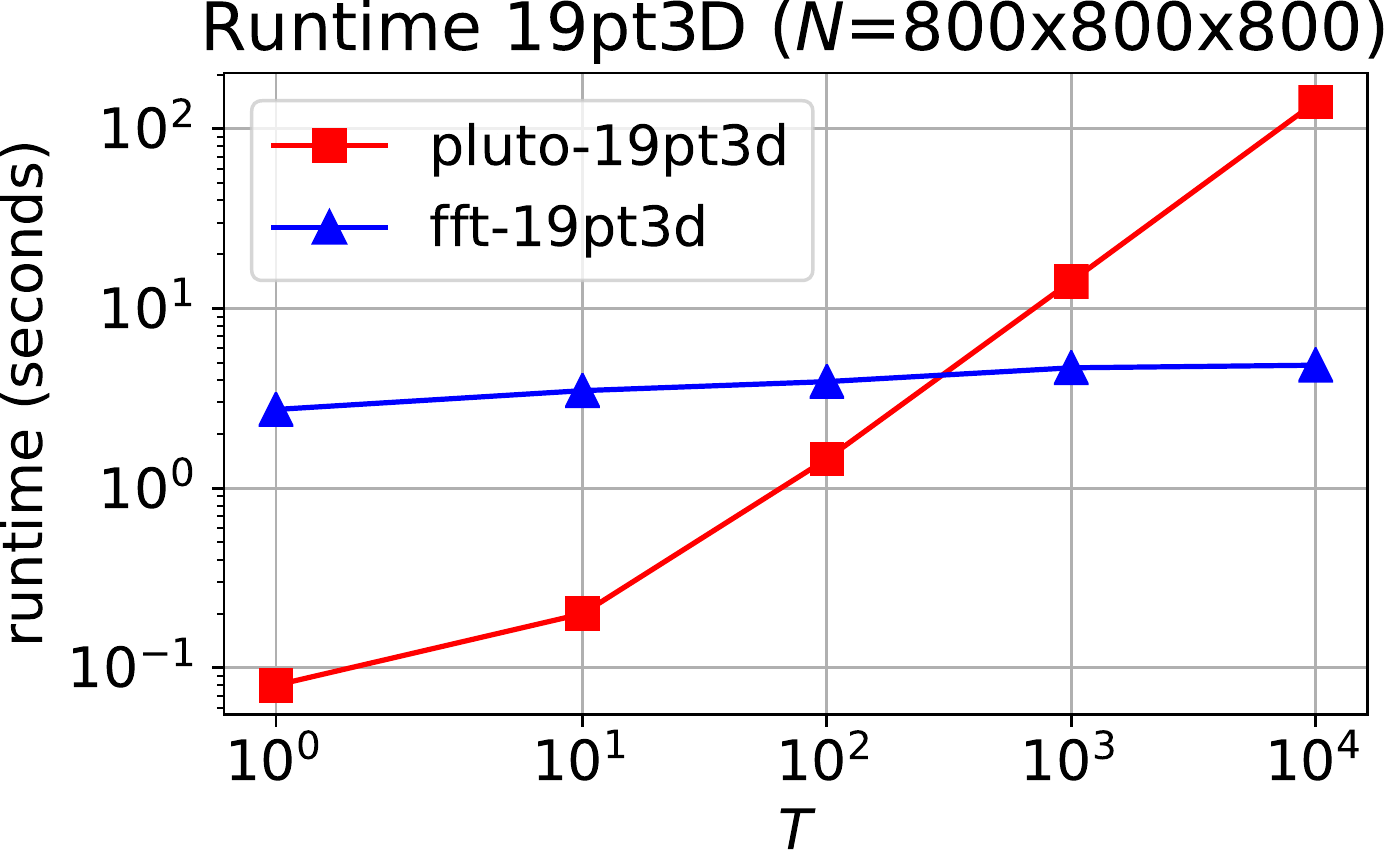}
$(i)$ & $(ii)$ & $(iii)$ \\[1ex]

{\includegraphics[width=0.32\textwidth, height=3.7cm, clip = true]{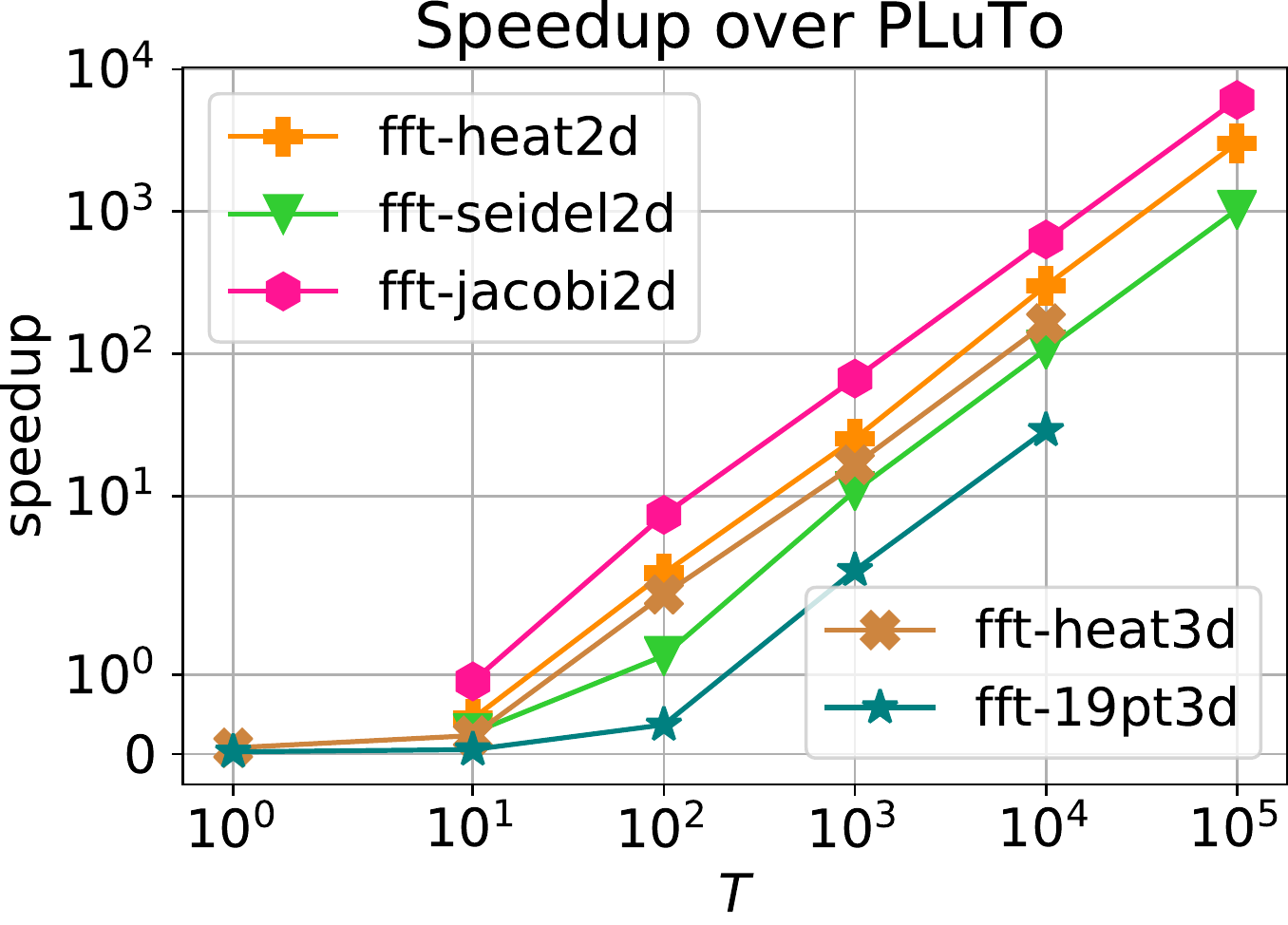}} & {\includegraphics[width=0.32\textwidth, height=3.7cm, clip = true]{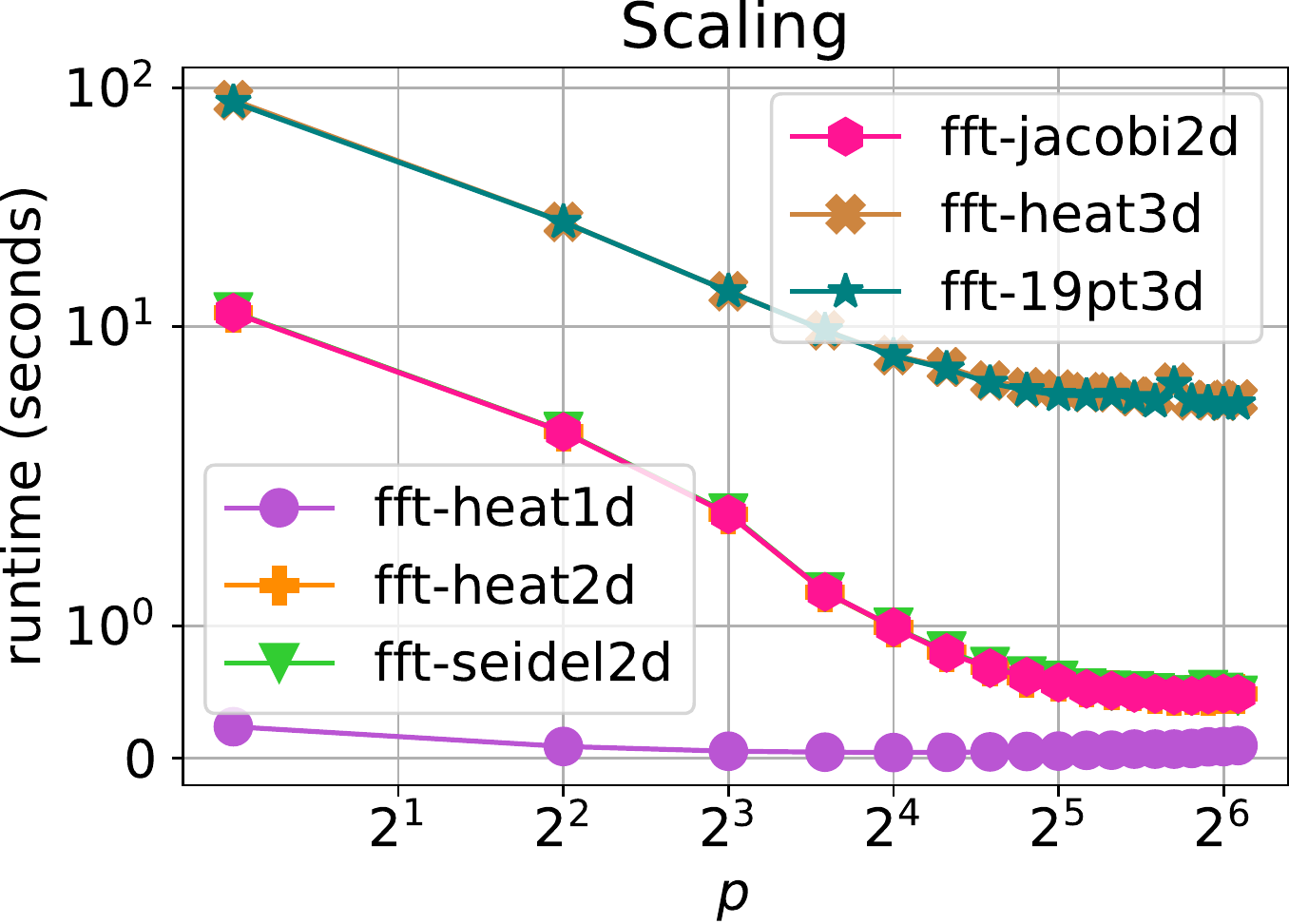}} &  \\
$(iv)$ & $(v)$ & \\\hline 



\end{tabular}
\figcaption{Performance comparison of our FFT-based \highlight{periodic} algorithms with the existing best stencil programs.}
\label{fig:plots-periodic-stencil-algorithms}
\vgap{}\vgap{}\vgap{}
\end{table*}

\begin{table*}[!ht]
\begin{tabular}{ccc}
\hline\\[-1.5ex]

\multicolumn{3}{c}{\textbf{KNL Node, Experiment 1 (More plots are given in Figures \ref{fig:appendix-plots-aperiodic-knl}, \ref{fig:appendix-plots-aperiodic-skx-1}, \ref{fig:appendix-plots-aperiodic-skx-2}, \ref{fig:appendix-plots-scalability-aperiodic-skx-exp2})}} \\[0.5ex]

\insertplotline{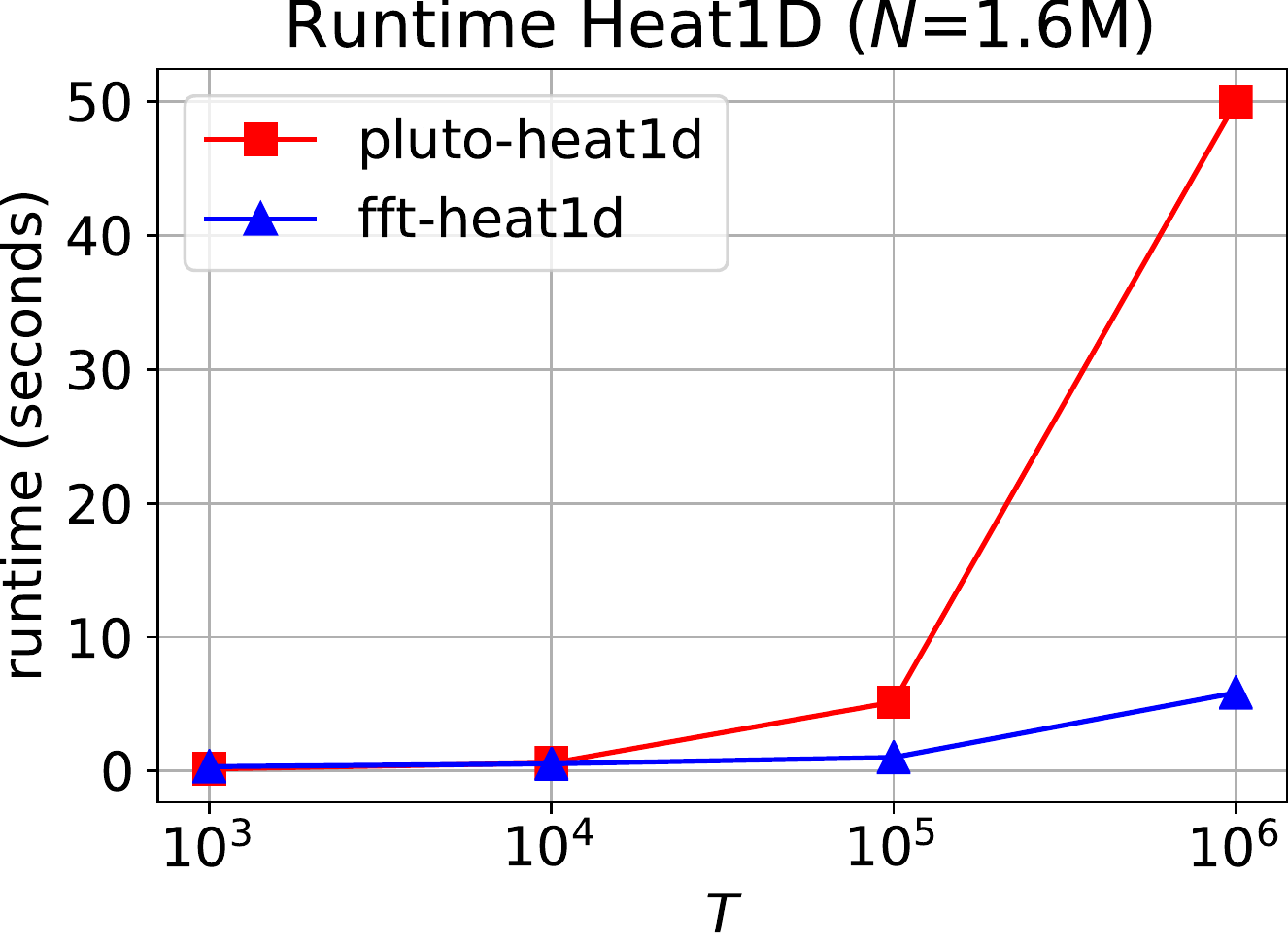}{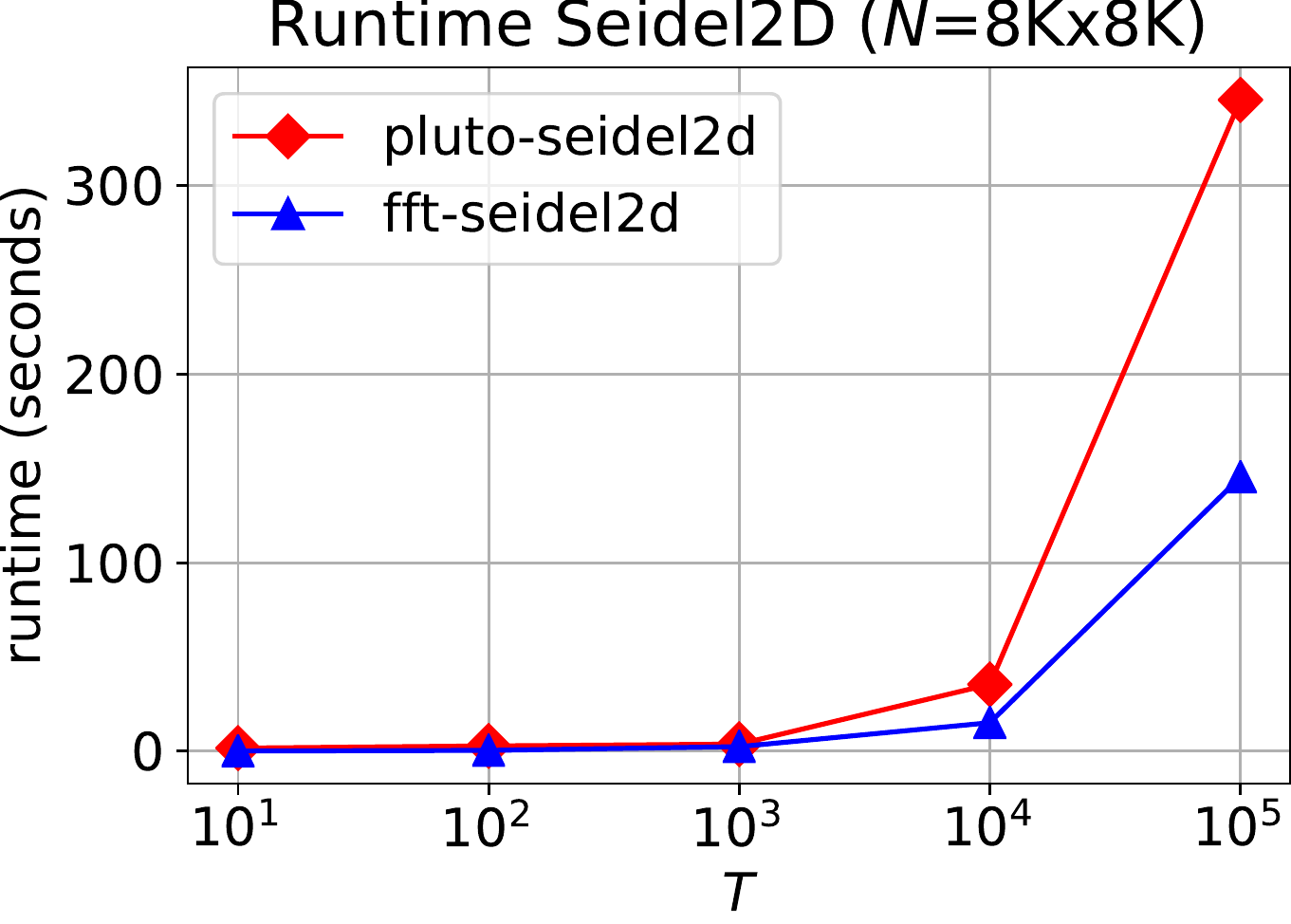}{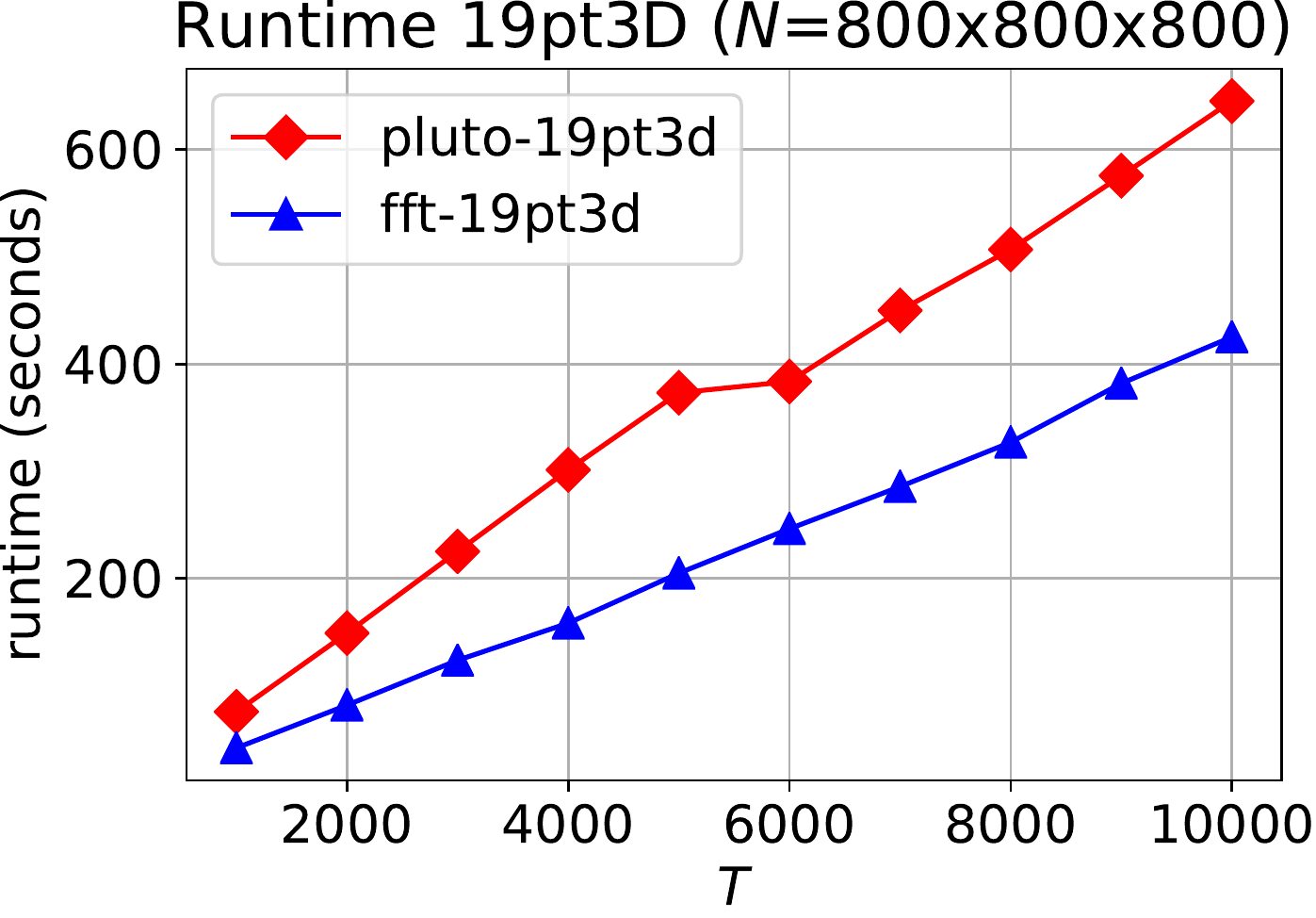}
$(i)$ & $(ii)$ & $(iii)$ \\[0.4ex]

\insertplotline{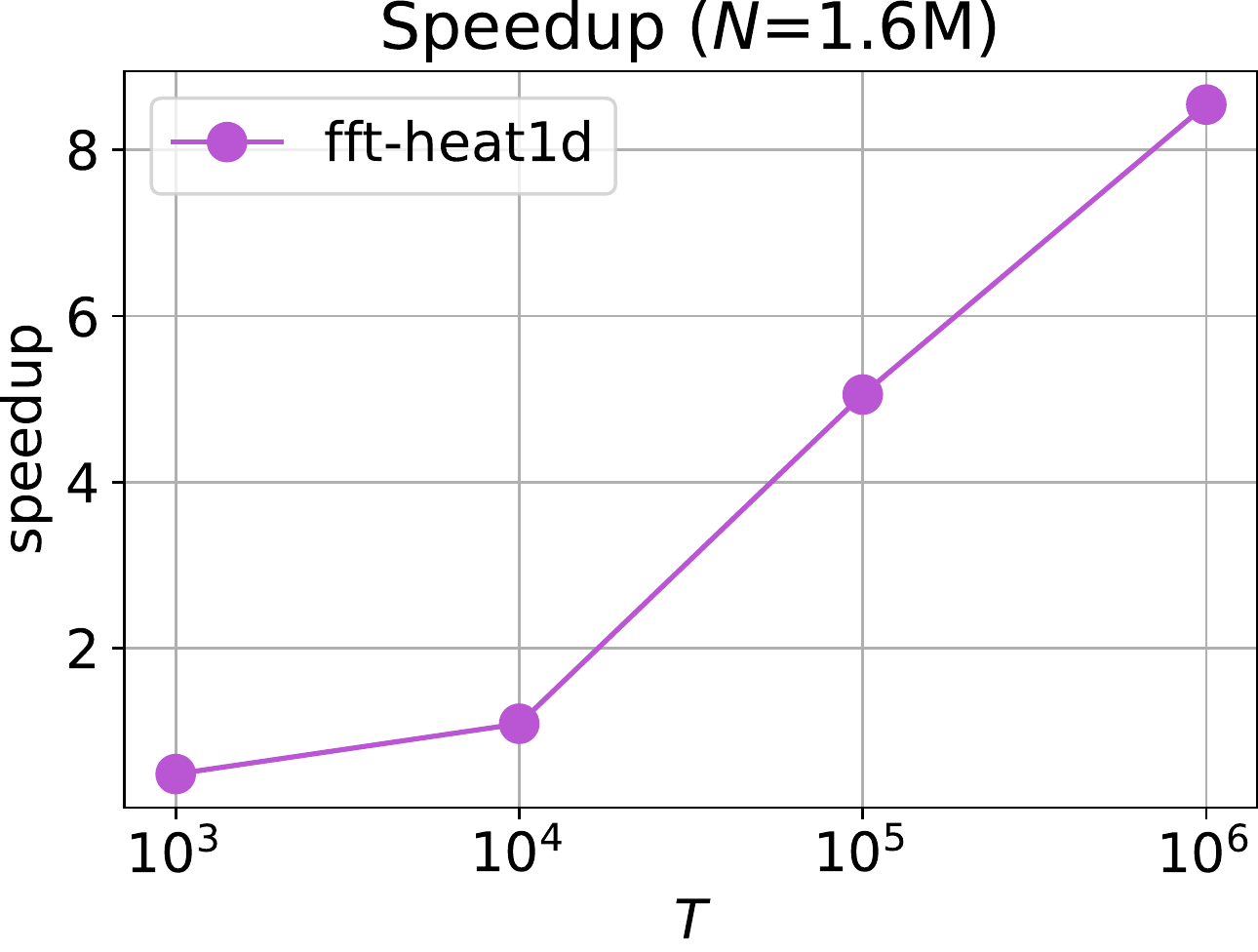}{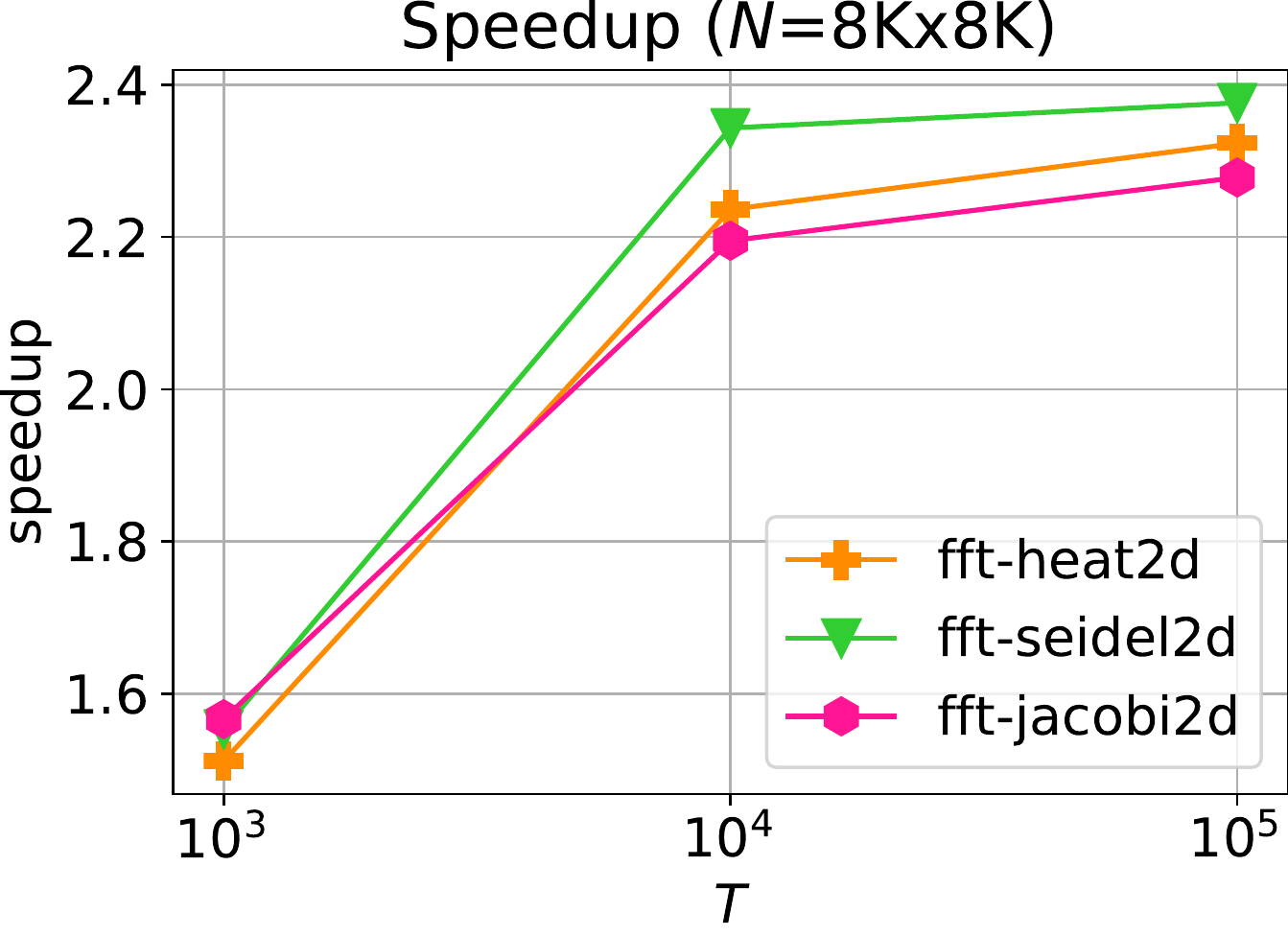}{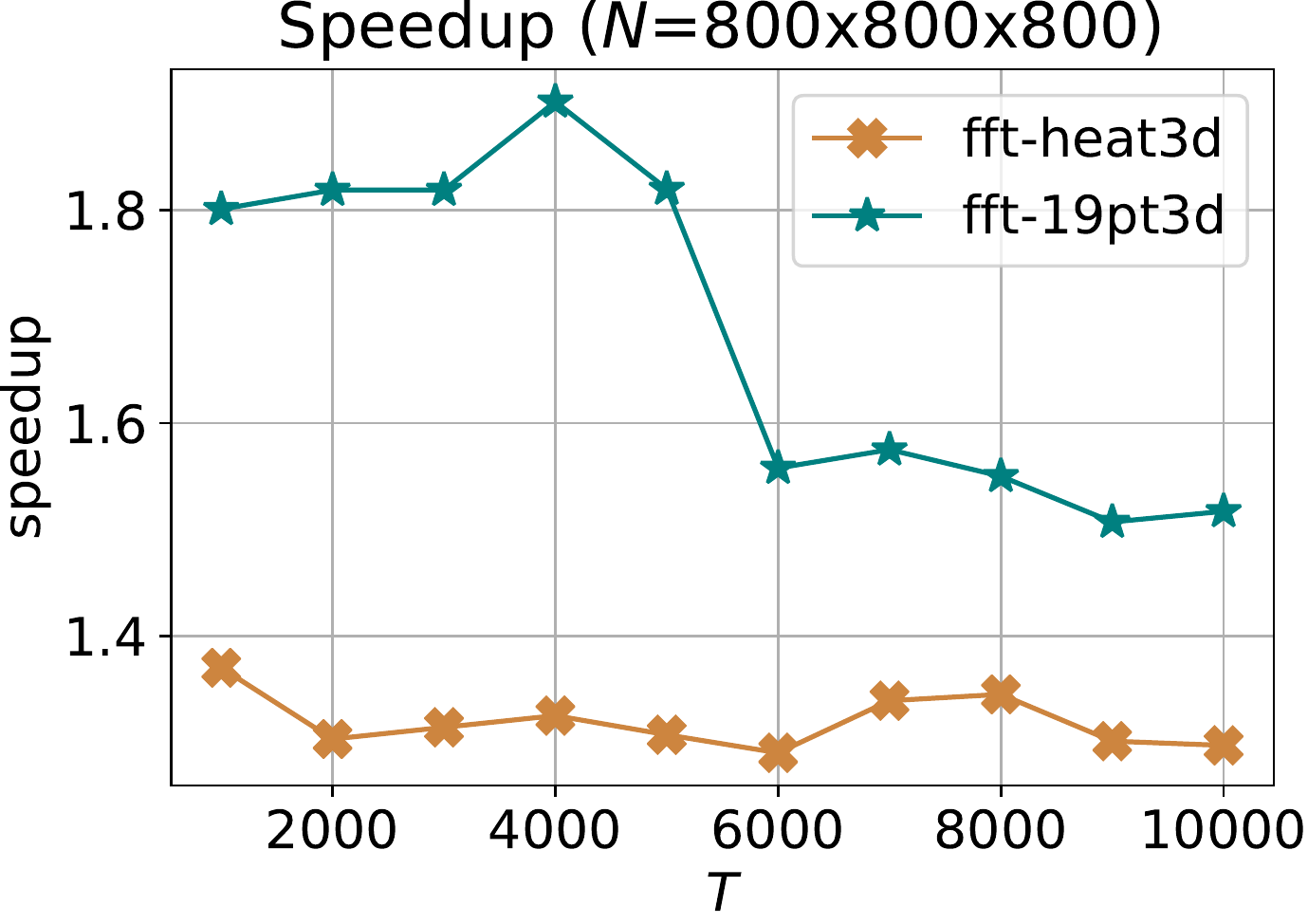}
$(iv)$ & $(v)$ & $(vi)$ \\\hline\rule{0pt}{1ex}\\[-1.5ex]

\multicolumn{3}{c}{\textbf{KNL Node, Experiment 2 (More plots are given in Figures \ref{fig:appendix-plots-aperiodic-knl}, \ref{fig:appendix-plots-aperiodic-skx-1}, \ref{fig:appendix-plots-aperiodic-skx-2}, \ref{fig:appendix-plots-scalability-aperiodic-skx-exp2})}} \\[0.5ex]

\insertplotline{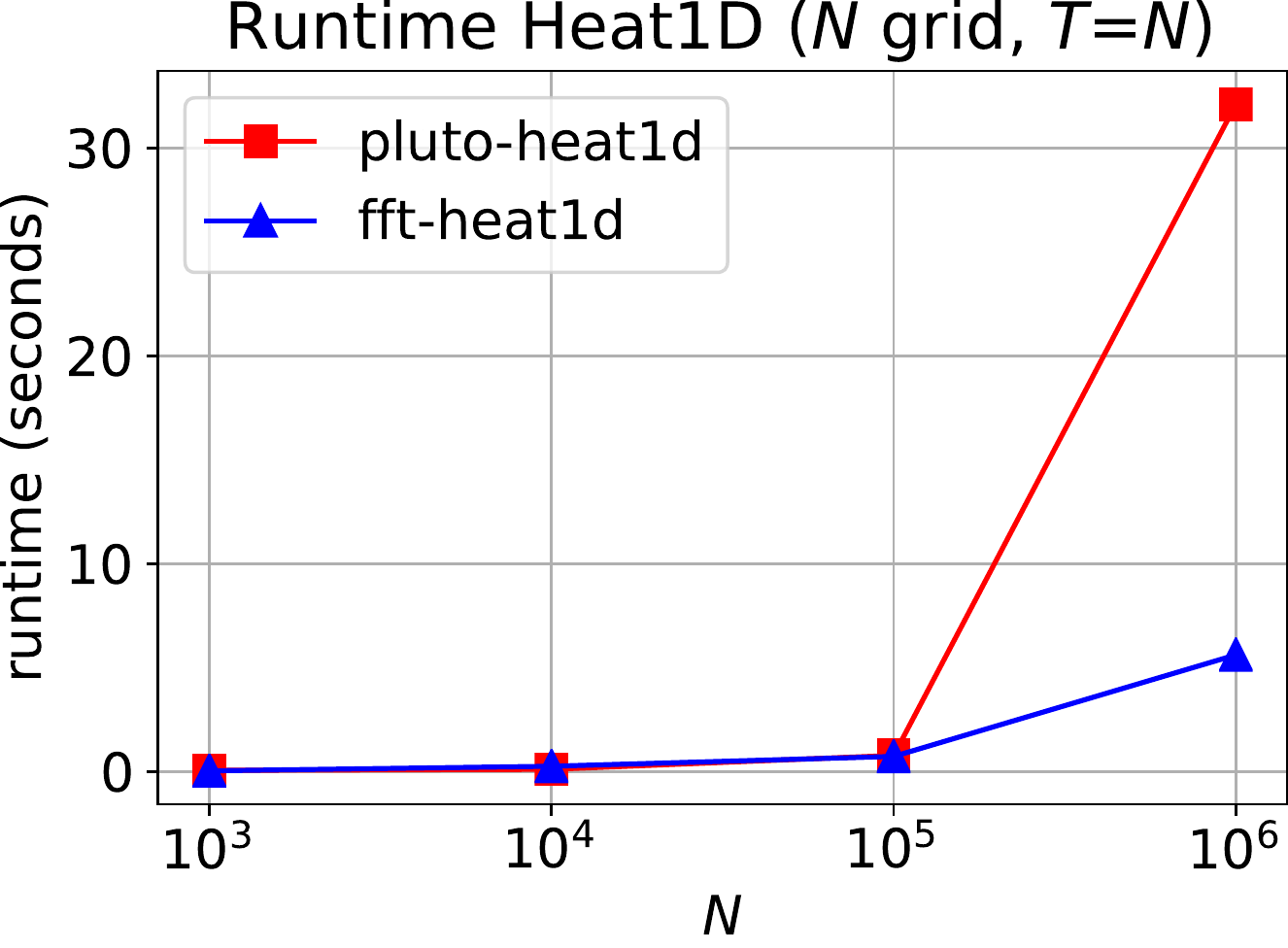}{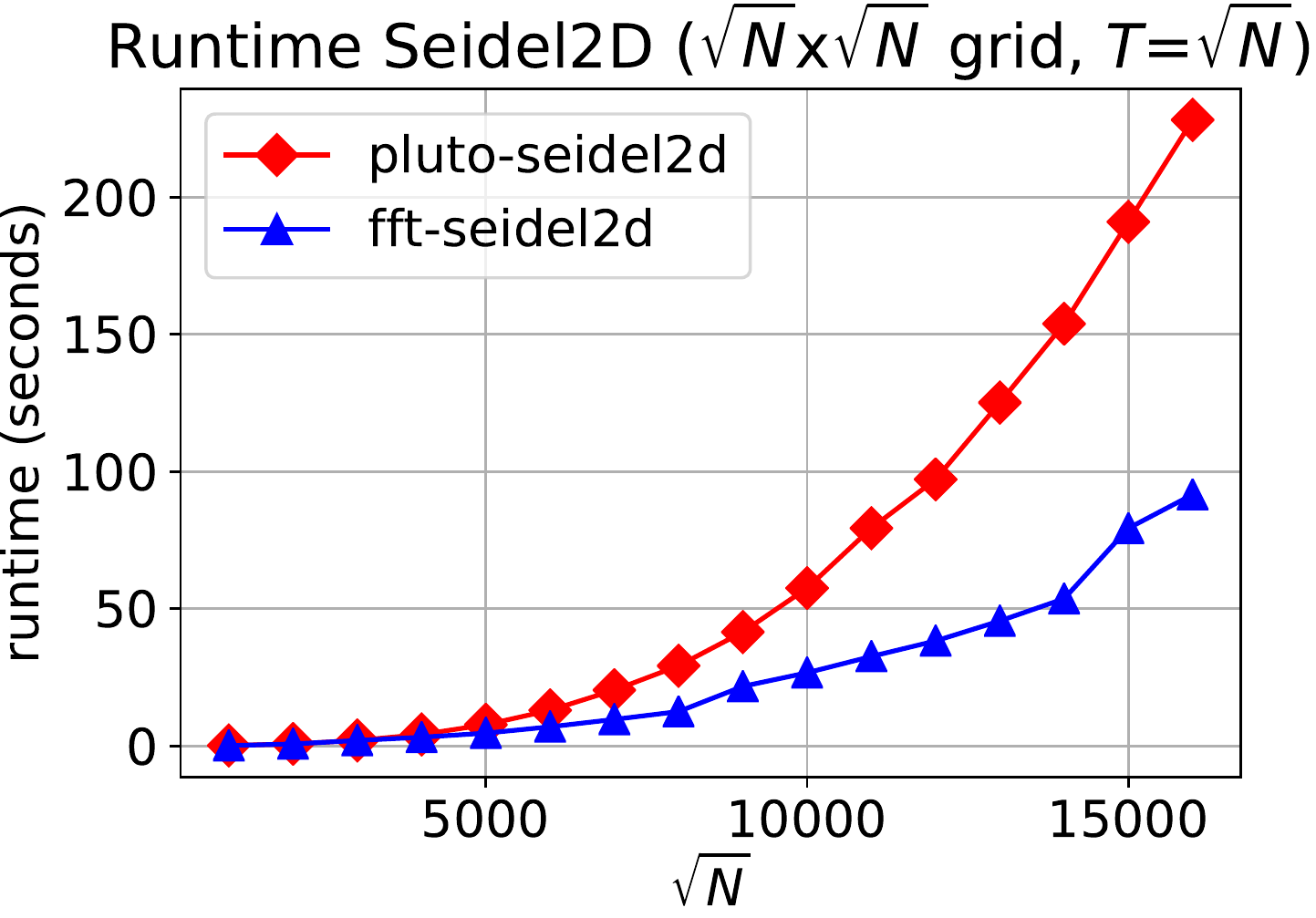}{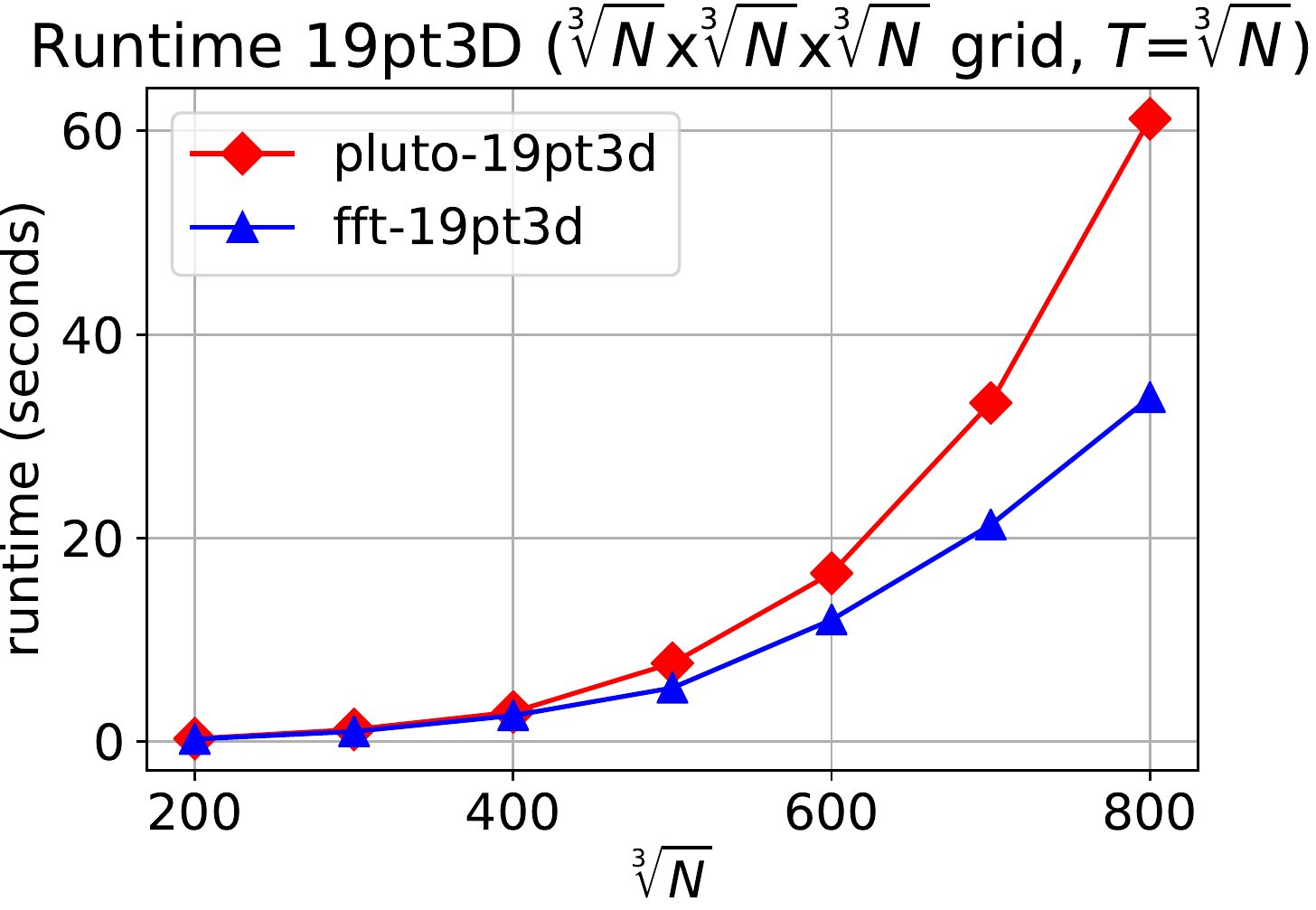}
$(vii)$ & $(viii)$ & $(ix)$ \\[0.4ex]

\insertplotline{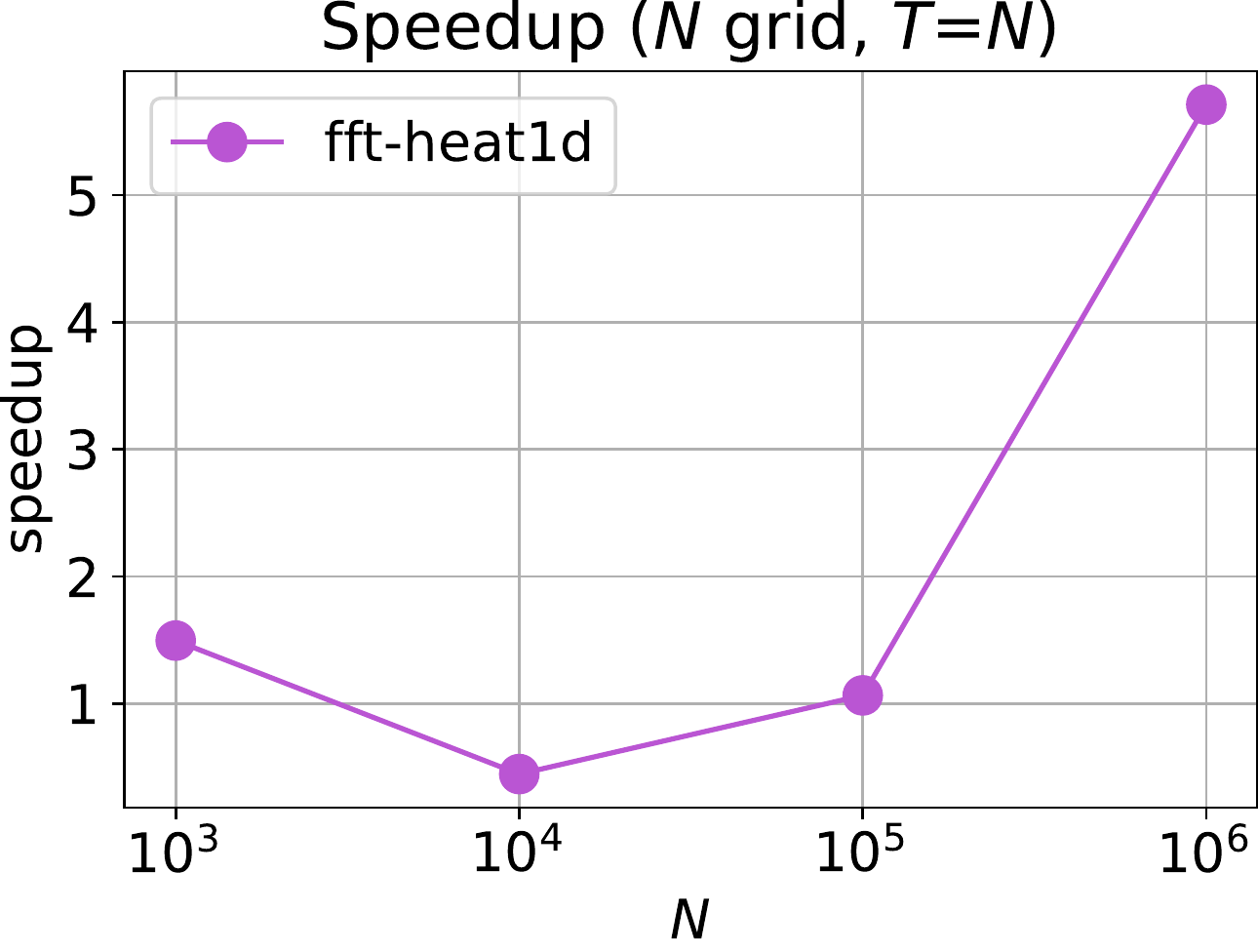}{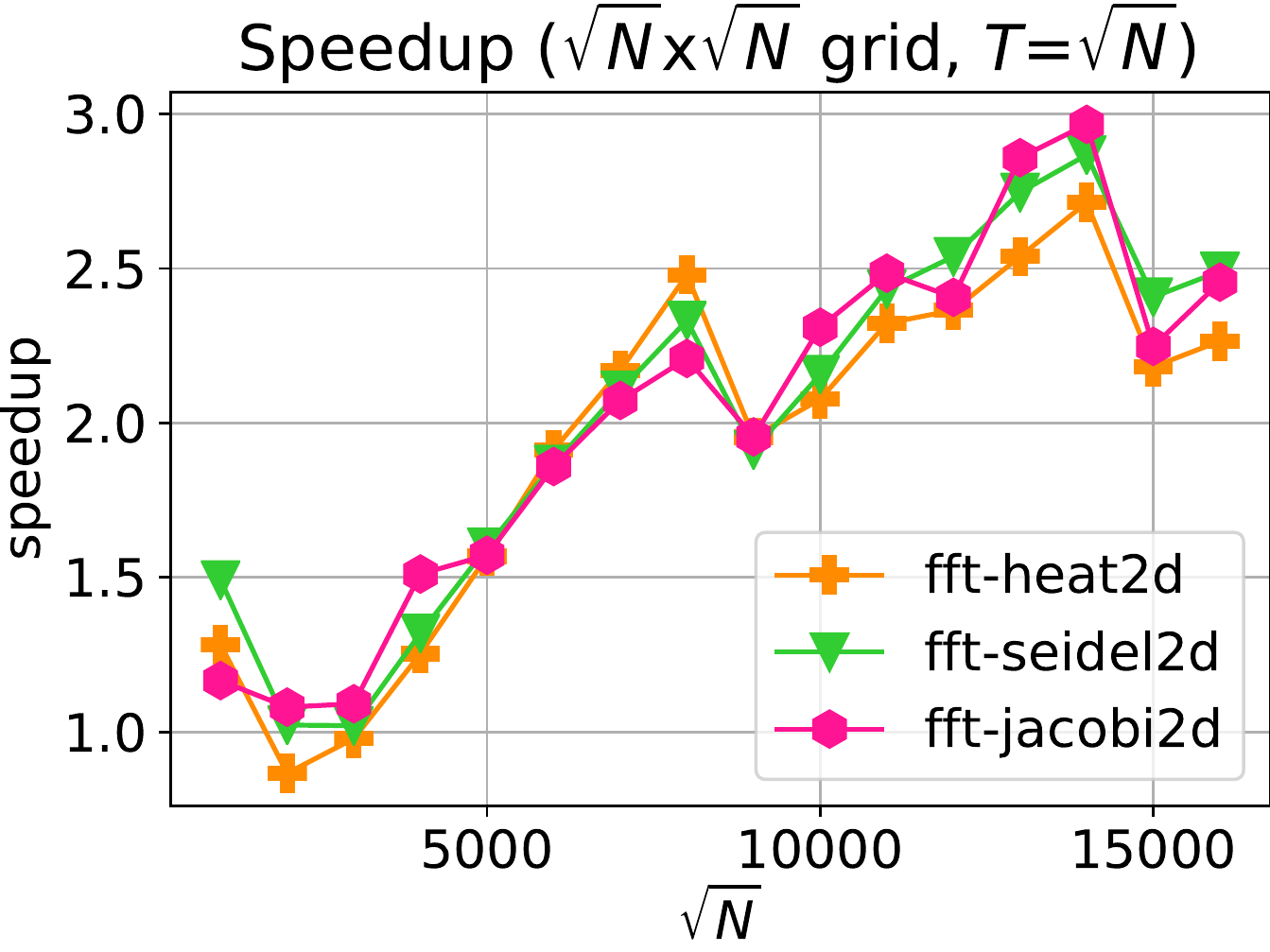}{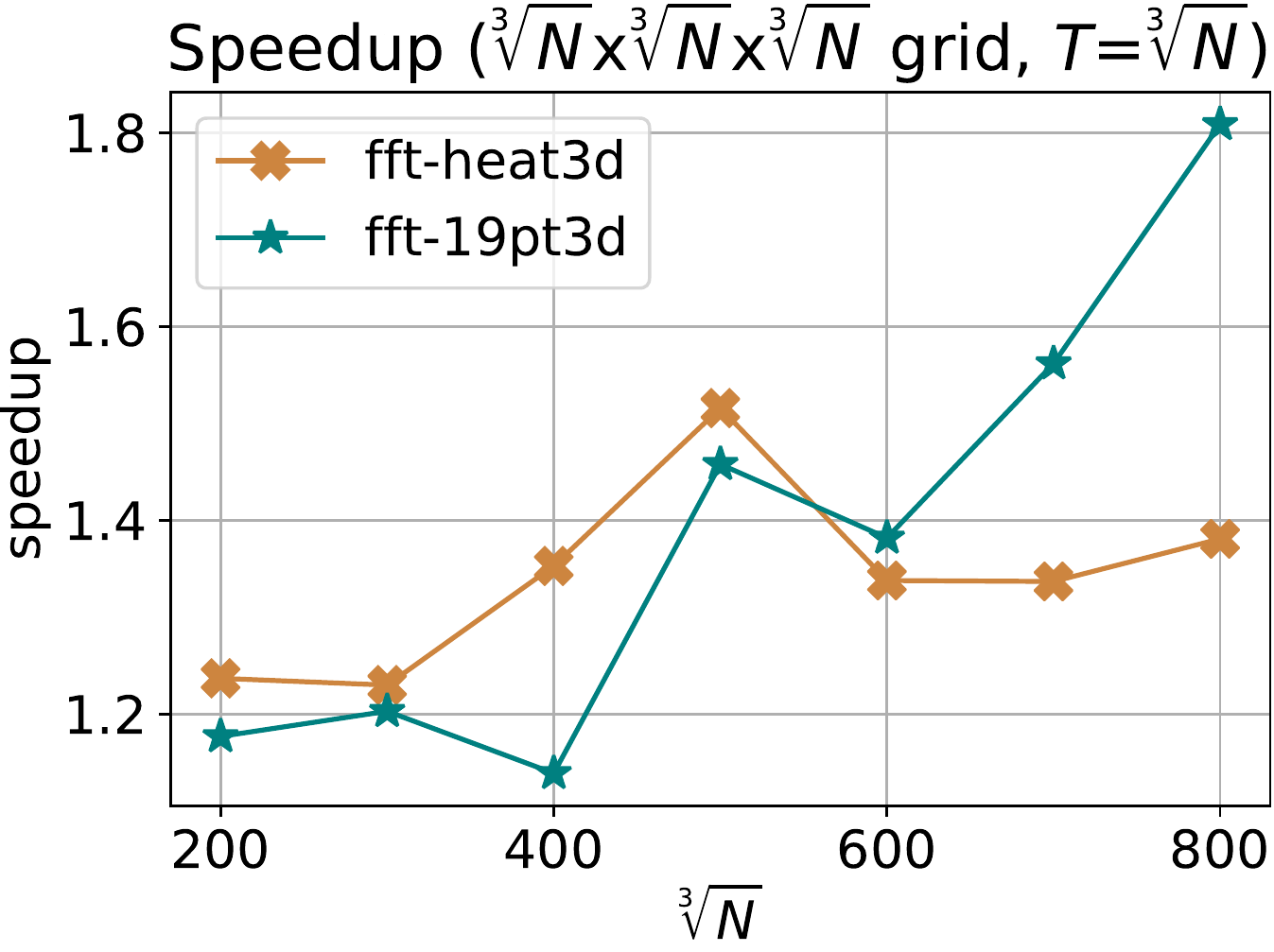} 
$(x)$ & $(xi)$ & $(xii)$ \\\hline


\end{tabular}
\figcaption{Performance comparison of our FFT-based \highlight{aperiodic} algorithms with the existing best stencil programs.}
\label{fig:plots-aperiodic-stencil-algorithms}
\vgap{}\vgap{}\vgap{}
\end{table*}

\section{Conclusion}
\label{sec:conclusion}

In this paper, we presented a pair of efficient algorithms based on \highlight{fast Fourier transforms} for performing \highlight{linear} stencil computations with periodic and aperiodic boundary conditions.  These are the first high-performing $\oh{NT}$-work\footnote{$N$ and $T$ are the spatial grid size and the number of timesteps, respectively.} stencil algorithms of significant generality for computing the spatial grid values at the final timestep from the input grid without explicitly generating values for most of the intermediate timesteps. Our stencil algorithms improve computational complexity and parallel running time bounds over the state-of-the-art stencil algorithms by a \highlight{polynomial factor}. Experimental results show that implementations of our algorithms run orders of magnitude faster than state-of-the-art implementations for periodic stencils and 1.3$\times$ to 8.5$\times$ faster for aperiodic stencils, while exhibiting no significant loss in numerical accuracy from floating point arithmetic.

A few interesting problems that one could aim to solve in the future include: $(1)$ designing efficient algorithms for certain classes of \textit{nonlinear} stencils and stencils with \textit{conditionals}, $(2)$ designing \textit{low-span} algorithms for aperiodic stencils, and $(3)$ designing algorithms to approximate \textit{inhomogenous} stencils.

\clearpage
\appendix
\begin{figure}[ht!]
\includegraphics[width=0.32\textwidth, height=3.7cm, clip = true]{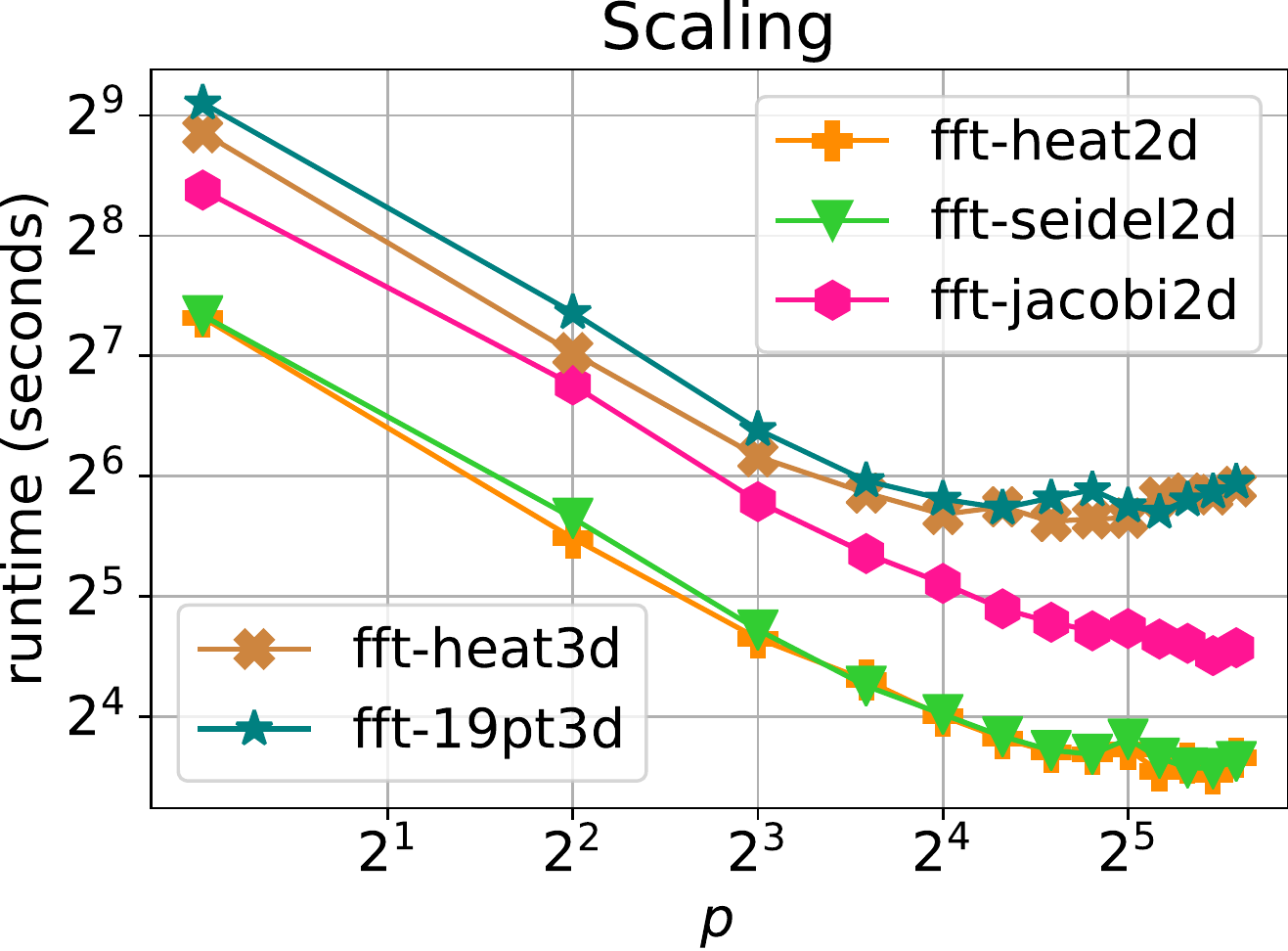}
\vgap{}
\caption{Scalability of our aperiodic algorithms for Experiment 2 on an SKX node.}
\label{fig:appendix-plots-scalability-aperiodic-skx-exp2}
\vgap{}\vgap{}\vgap{}
\end{figure}

\begin{table*}[!ht]
\begin{tabular}{c c c}
\hline\\[-1.5ex]

\multicolumn{3}{c}{\textbf{KNL Node}} \\[0.5ex]

\insertplotline{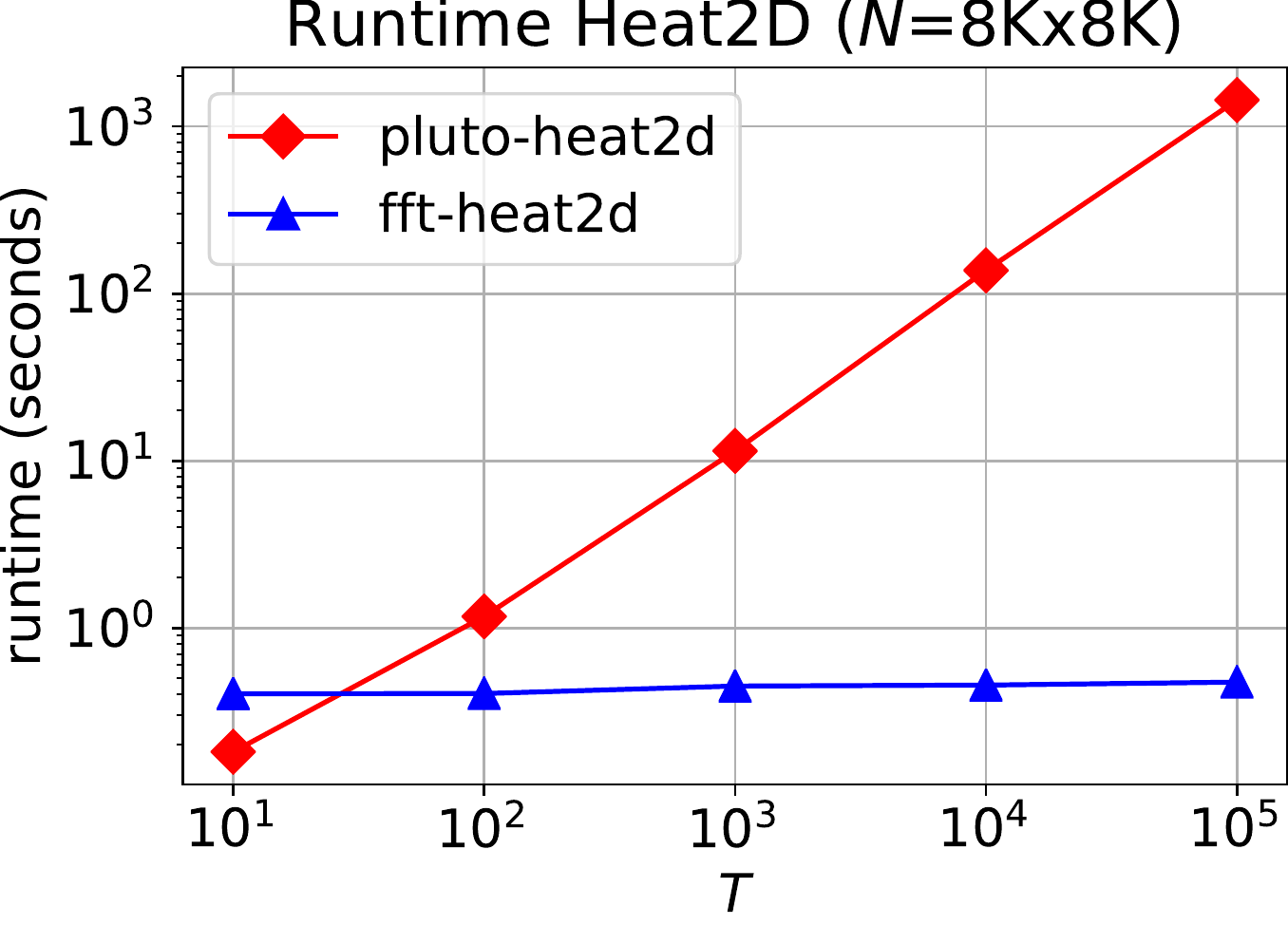}{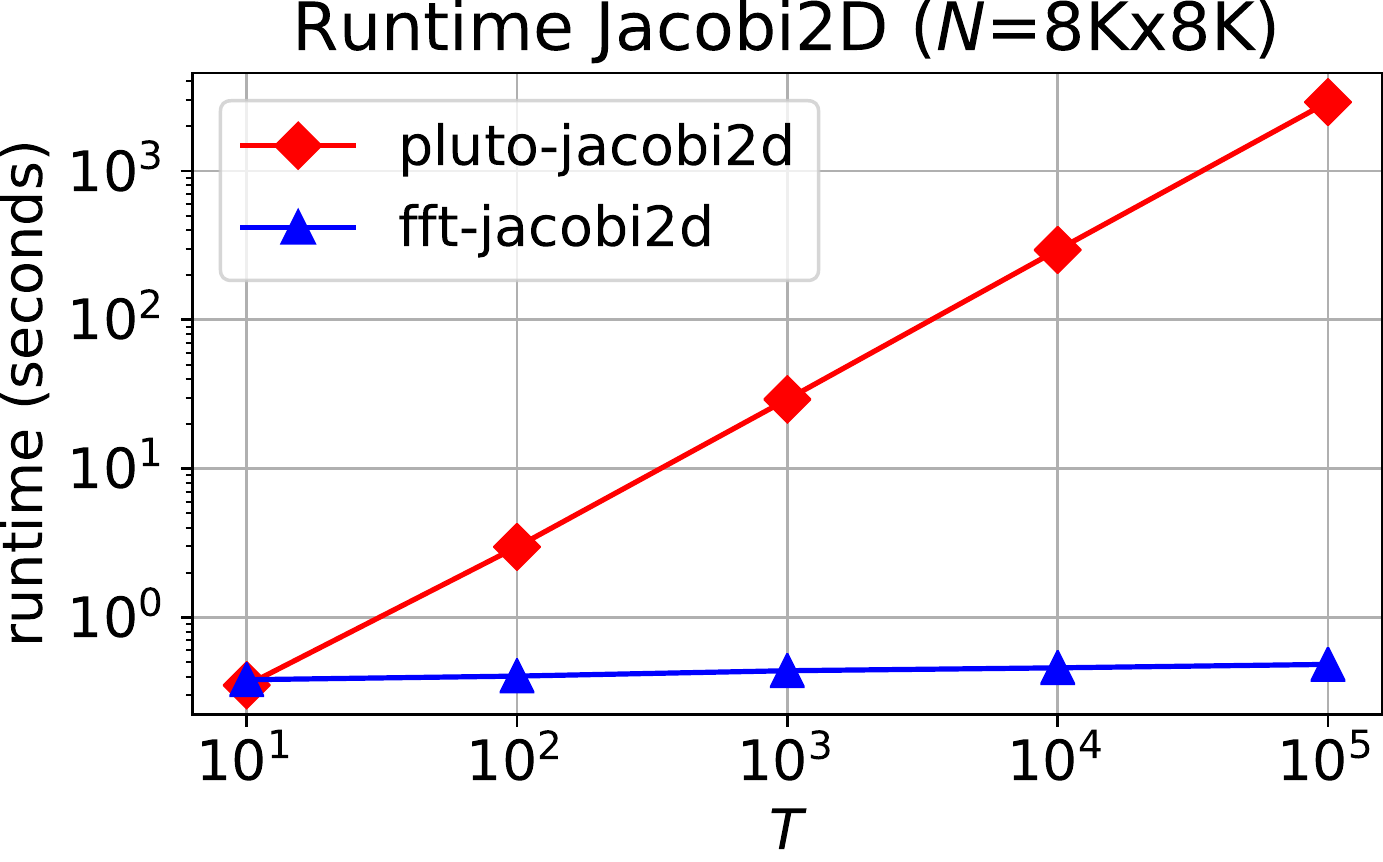}{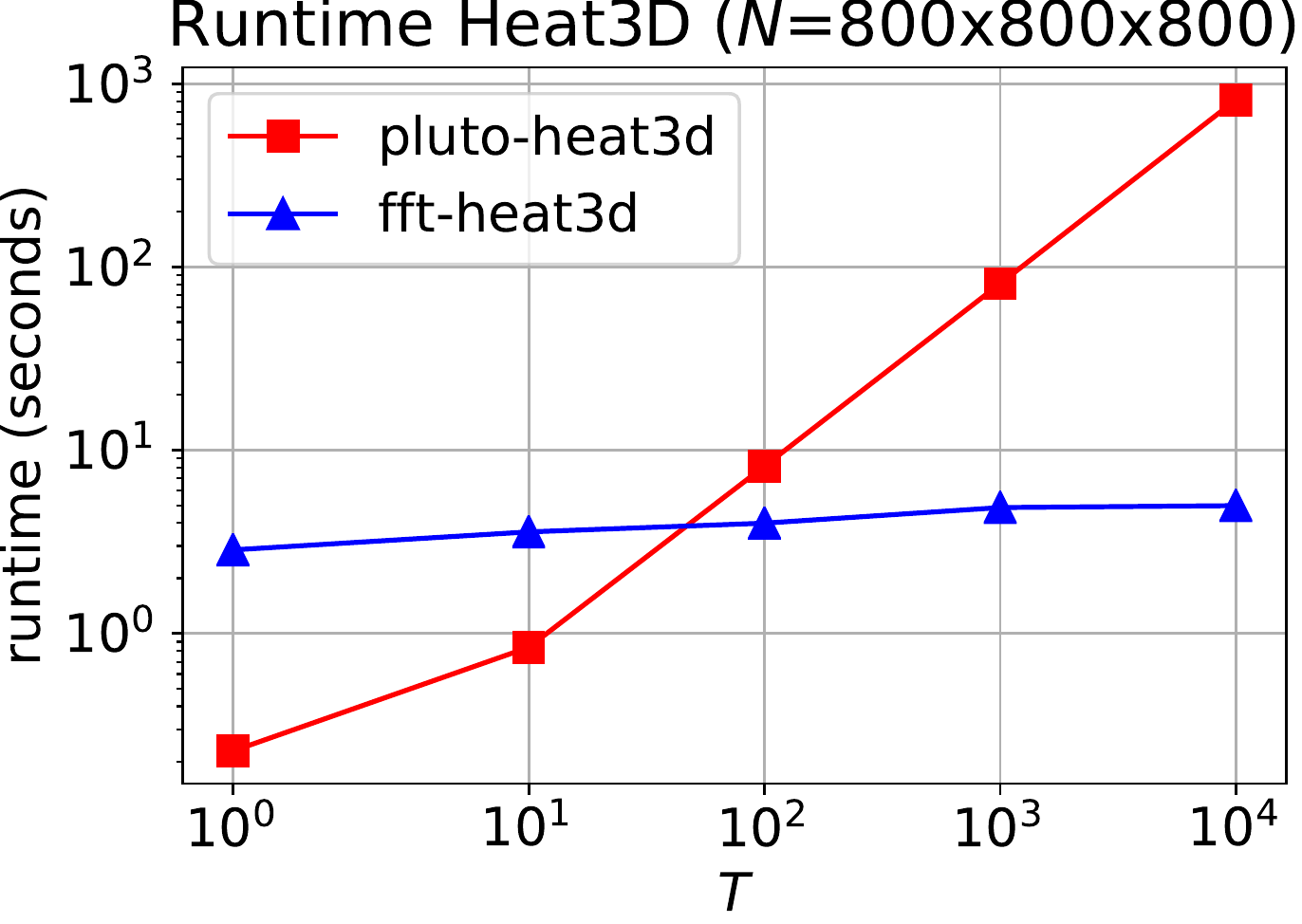}
$(i)$ & $(ii)$ & $(iii)$ \\\hline \rule{0pt}{1ex}\\[-1.5ex] 

\multicolumn{3}{c}{\textbf{SKX Node}} \\[0.5ex]

\insertplotline{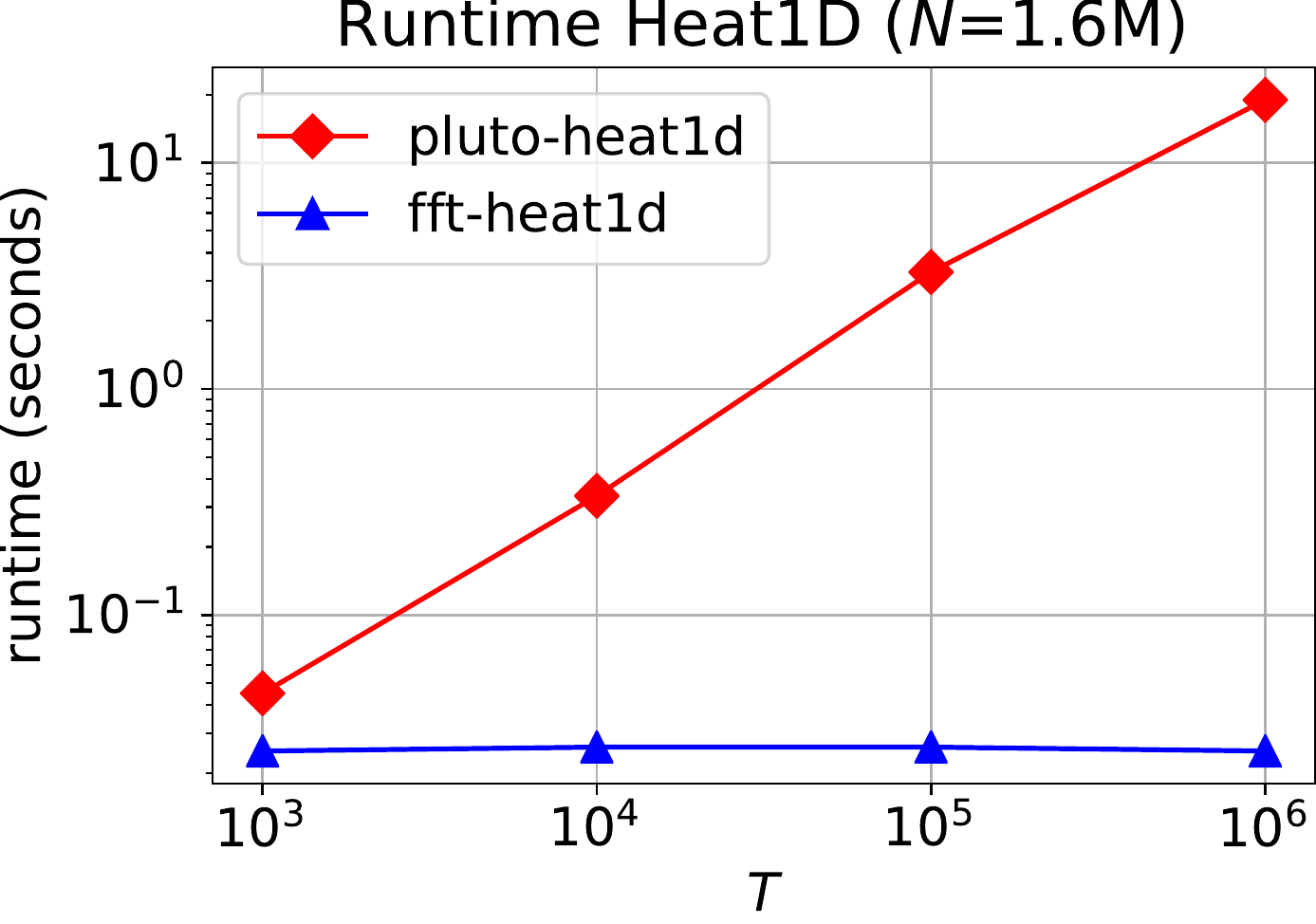}{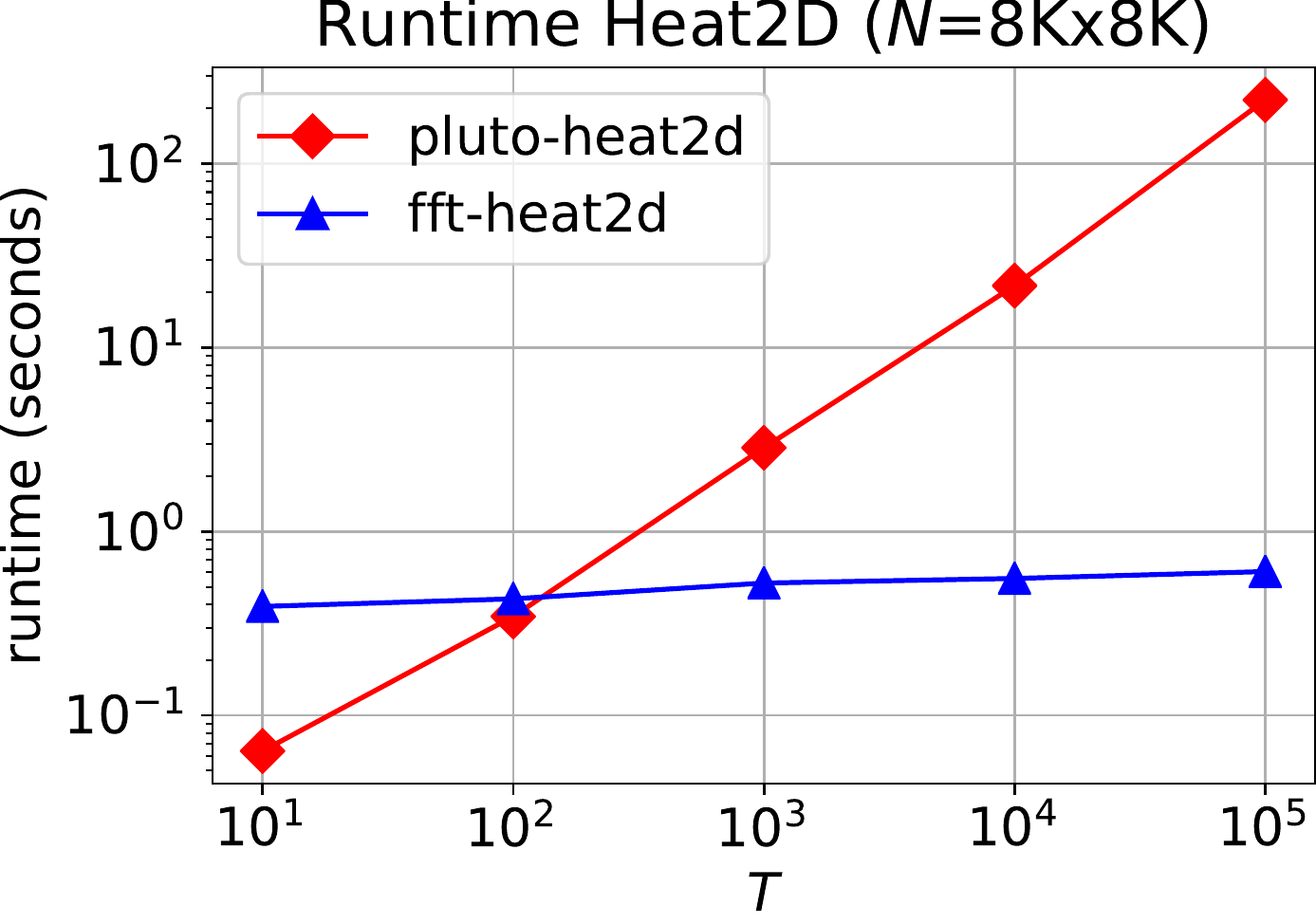}{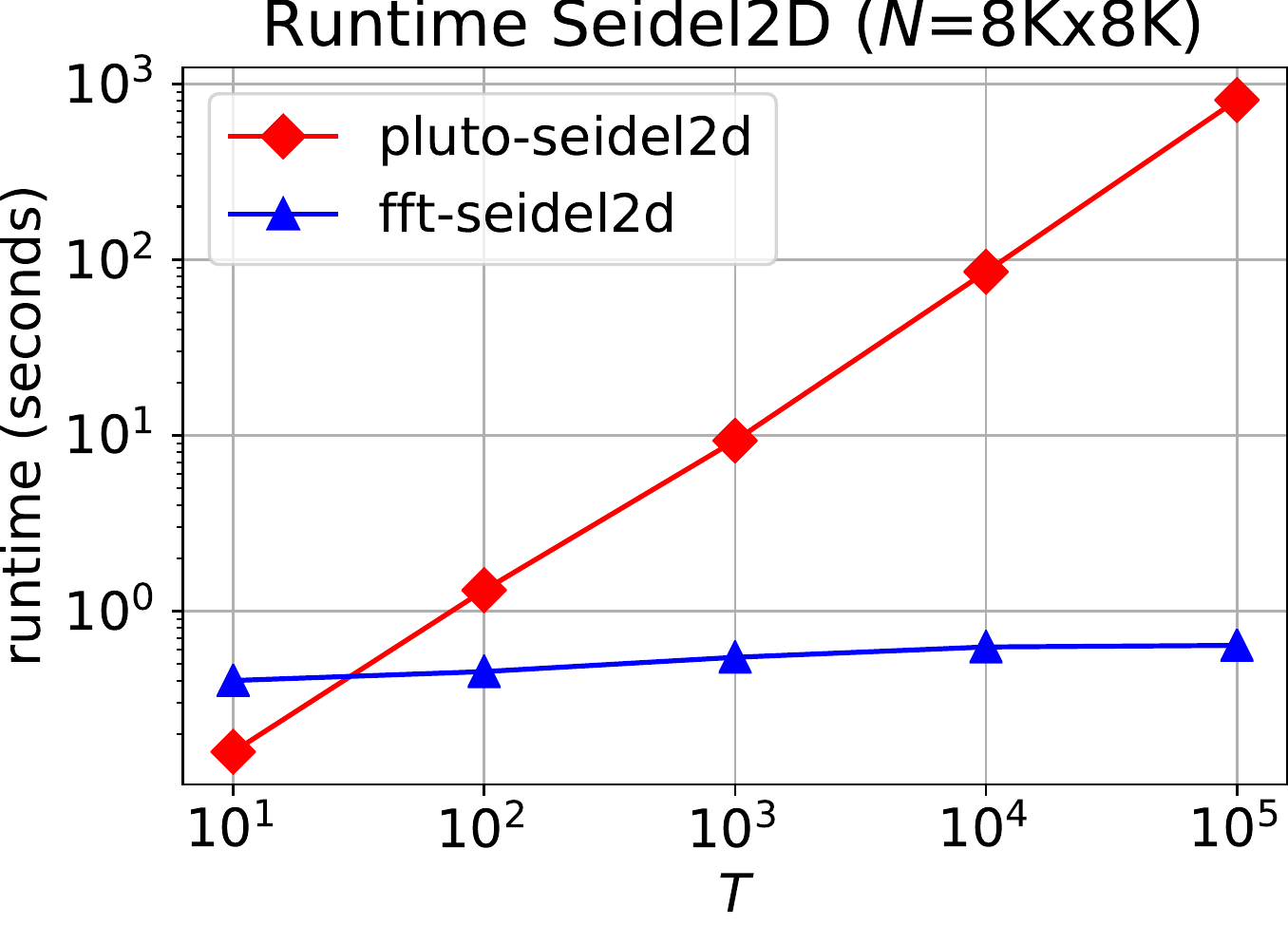}
$(iv)$ & $(v)$ & $(vi)$ \\[0.5ex]

\insertplotline{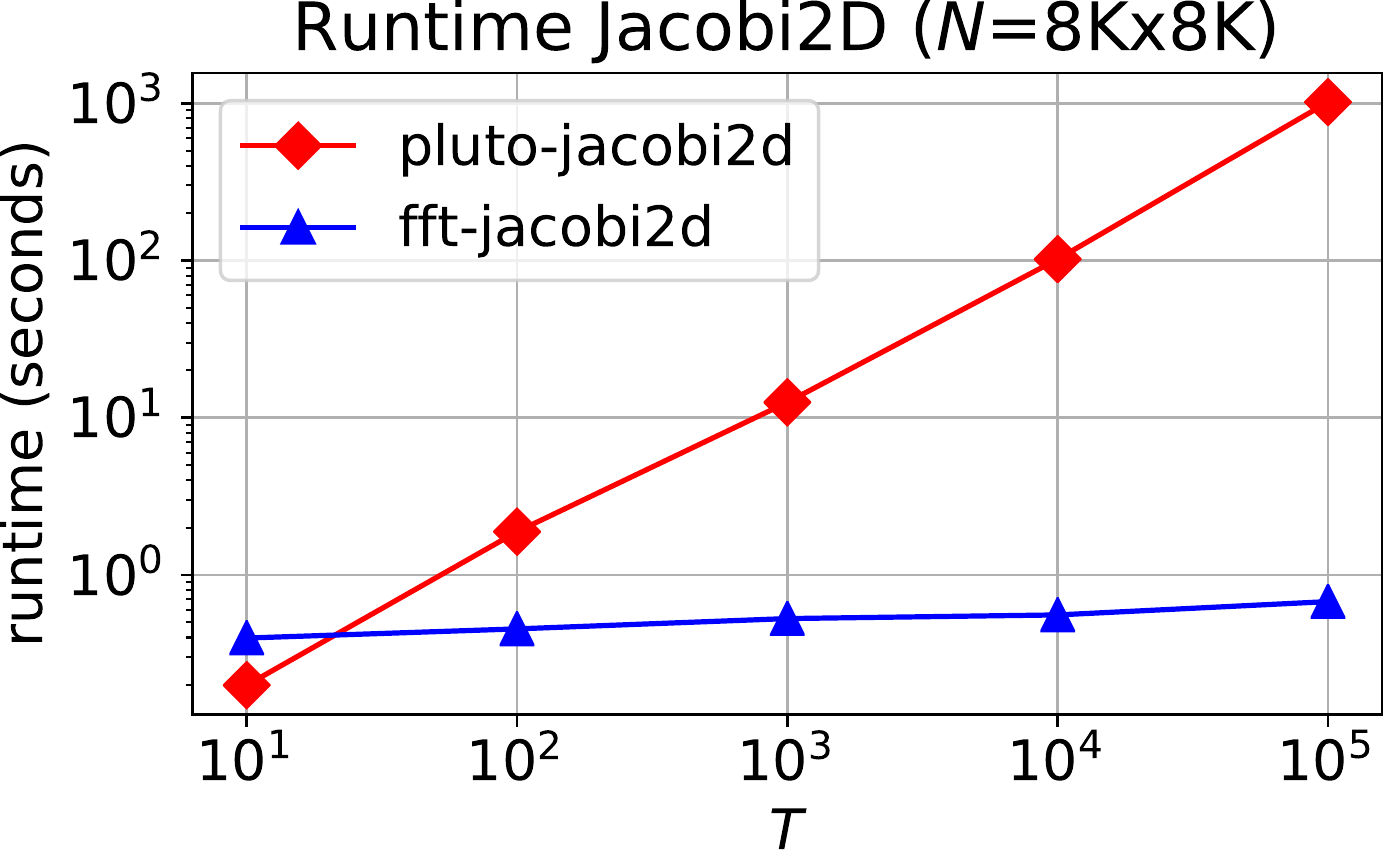}{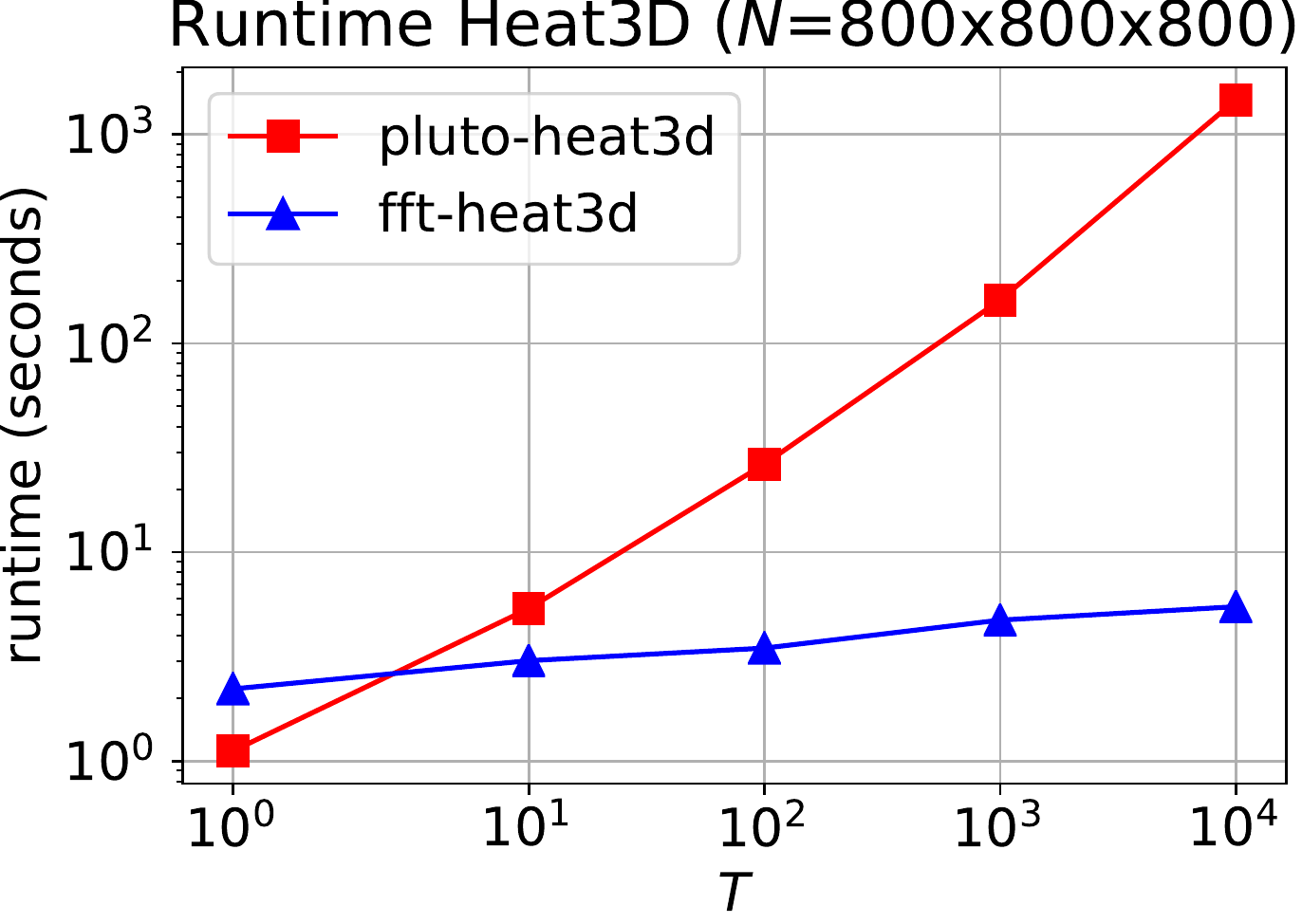}{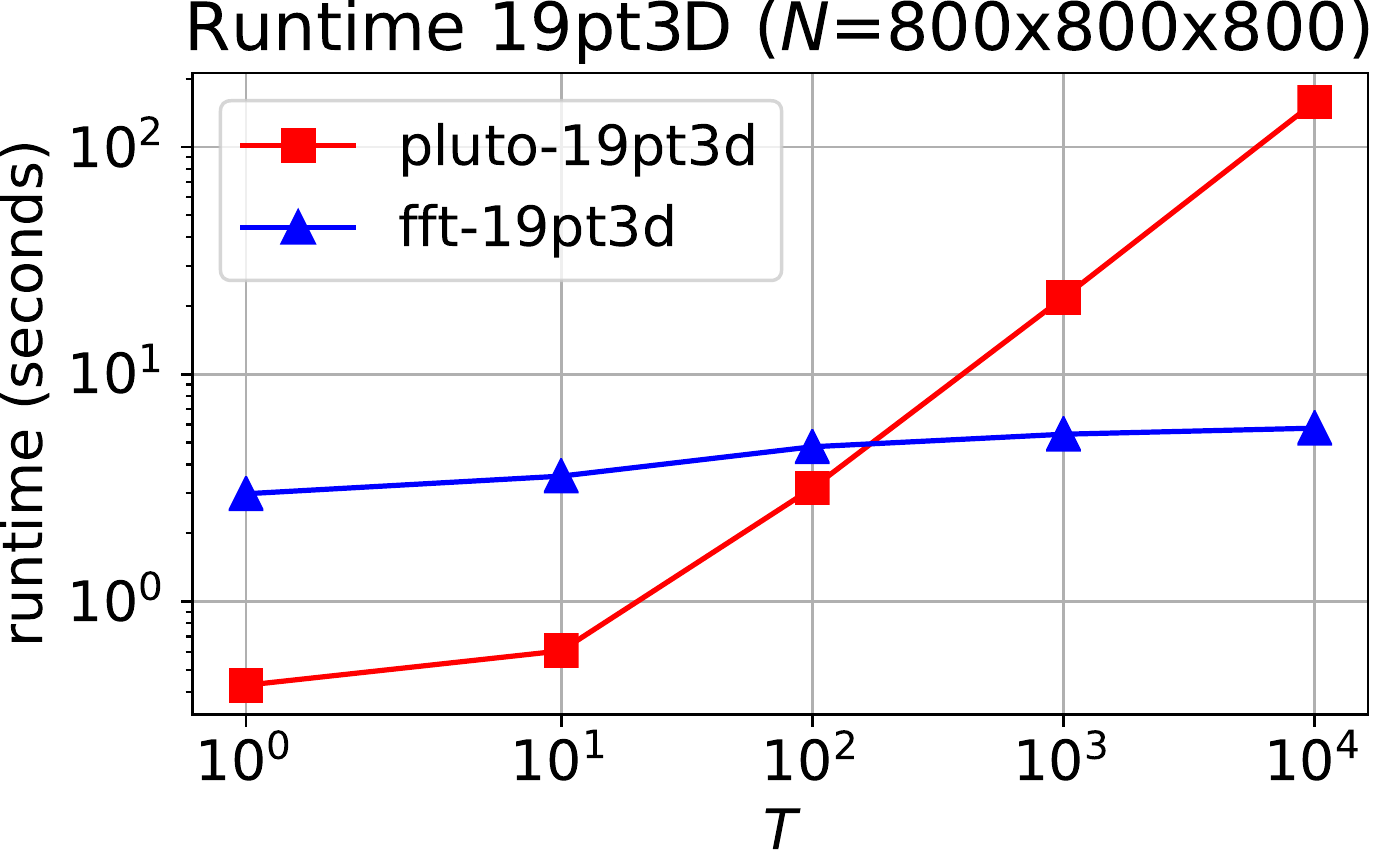}
$(vii)$ & $(viii)$ & $(ix)$ \\

{\includegraphics[width=0.32\textwidth, height=3.7cm, clip = true]{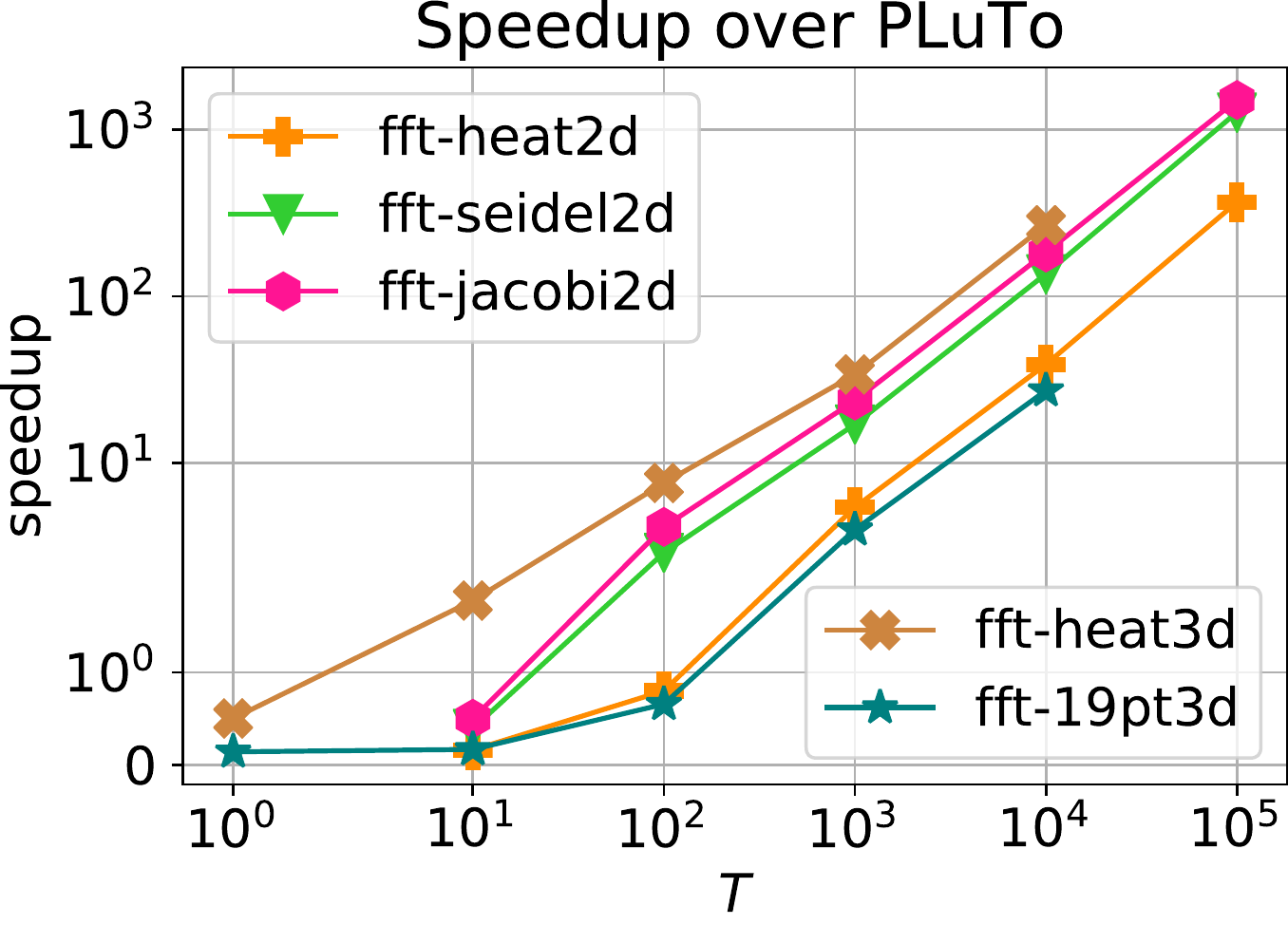}} & {\includegraphics[width=0.32\textwidth, height=3.7cm, clip = true]{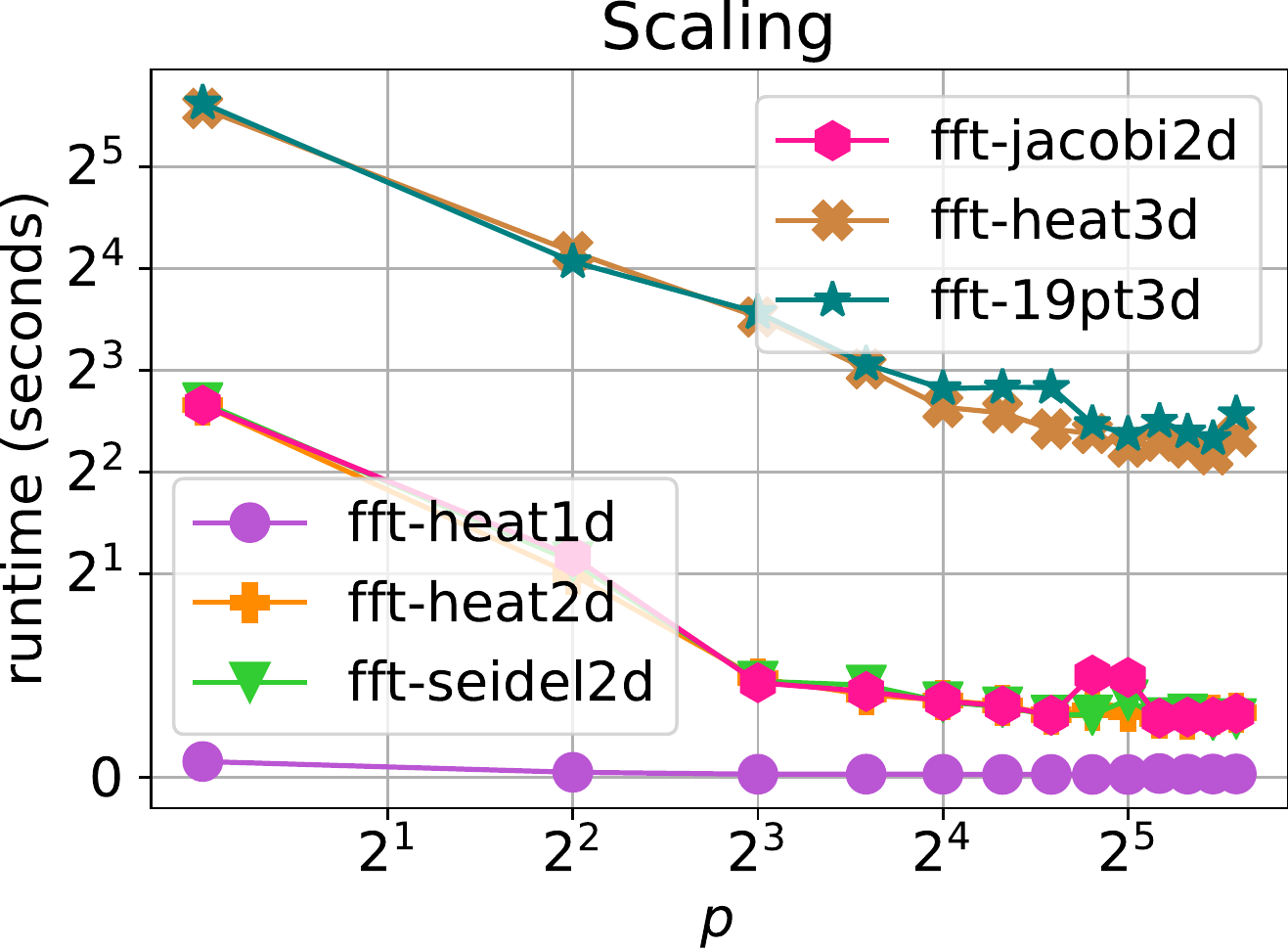}} &  \\
$(x)$ & $(xi)$ & \\\hline

\end{tabular}
\figcaption{Performance comparison of our FFT-based \highlight{periodic} algorithms with the existing best stencil programs.}
\label{fig:appendix-plots-periodic-stencil-algorithms}
\vgap{}\vgap{}\vgap{}
\end{table*}

\begin{table*}[!ht]
\begin{tabular}{ccc}
\hline\\[-1.5ex]

\multicolumn{3}{c}{\textbf{KNL Node, Experiment 1}} \\[0.5ex]

\insertplotline{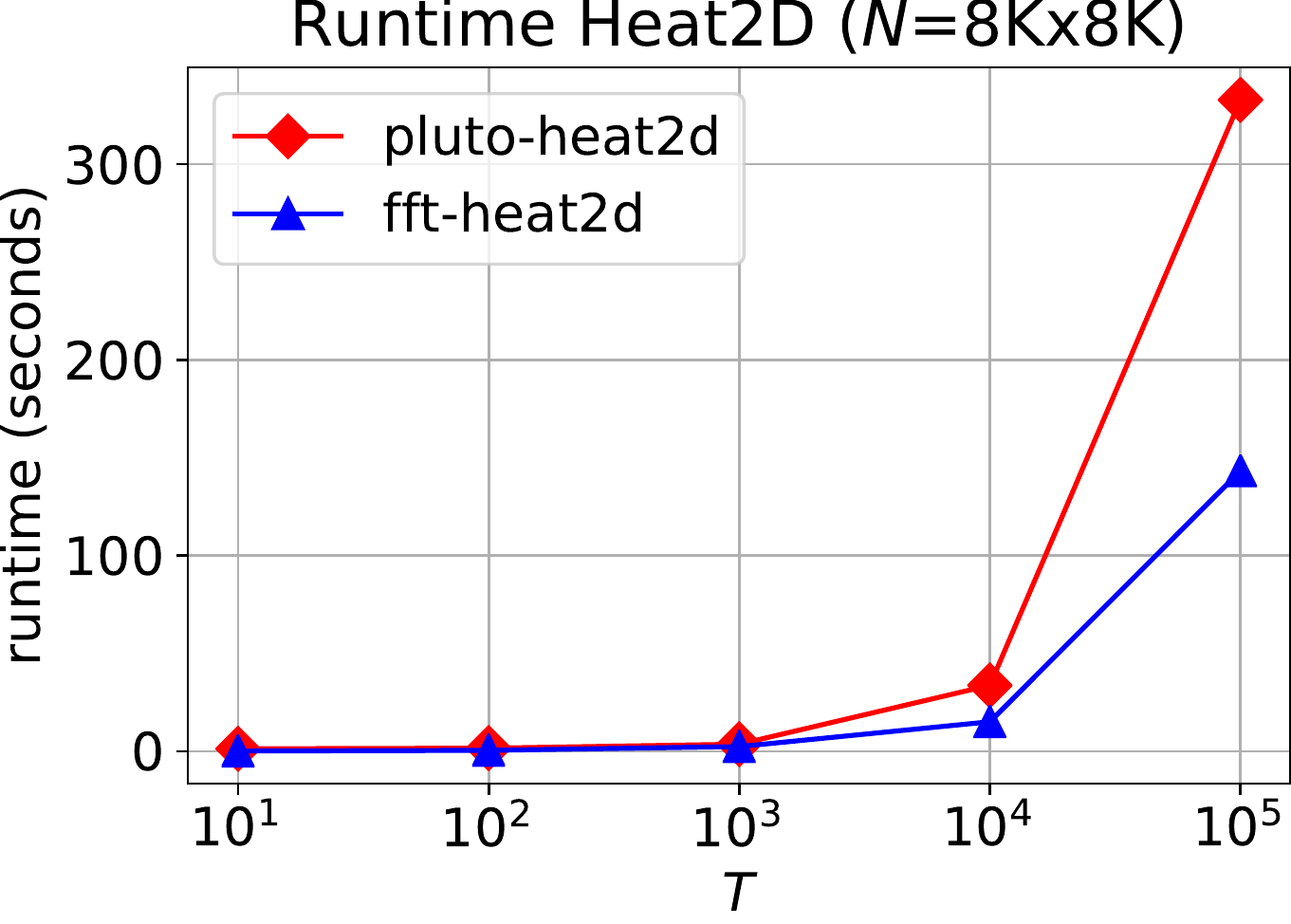}{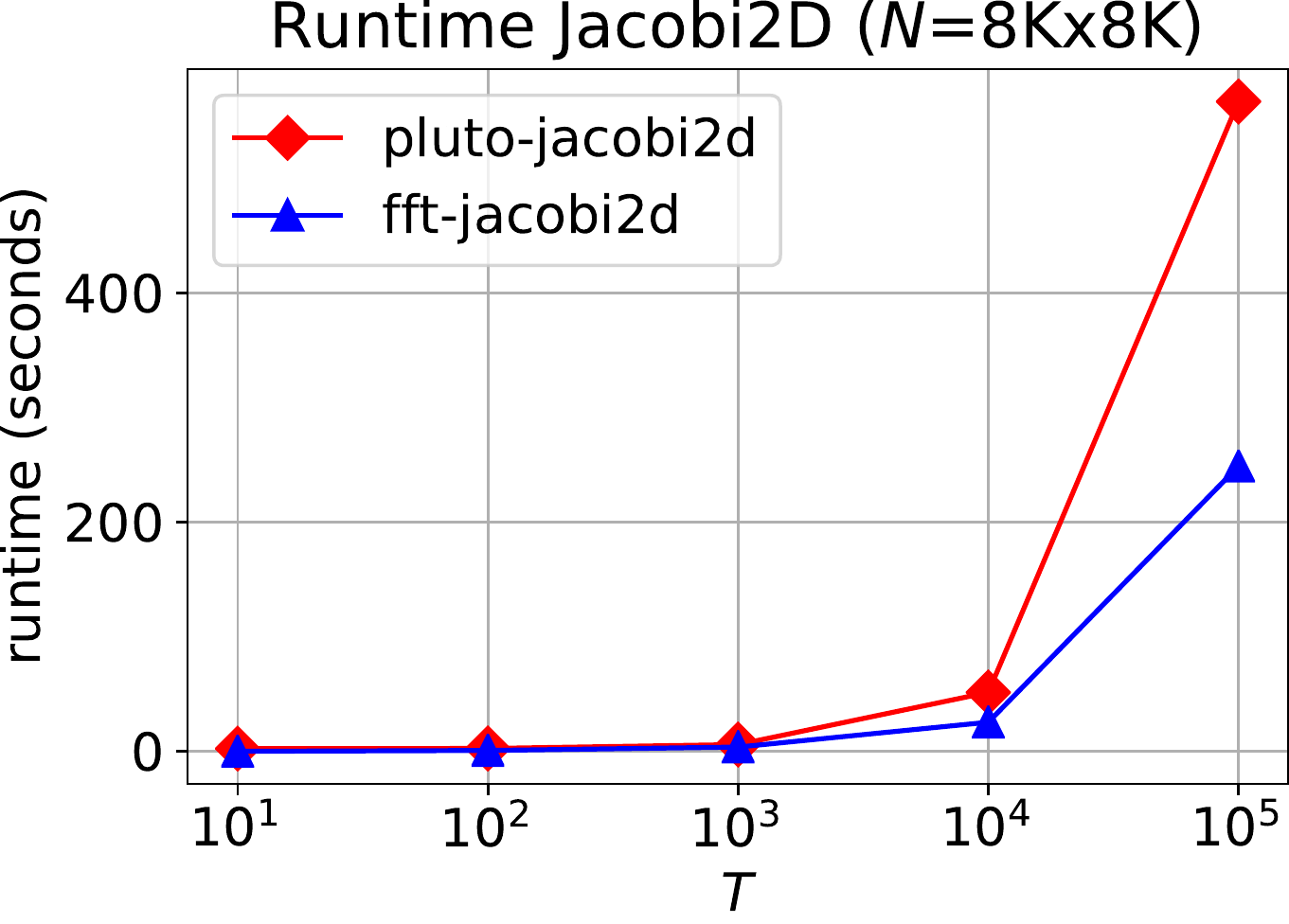}{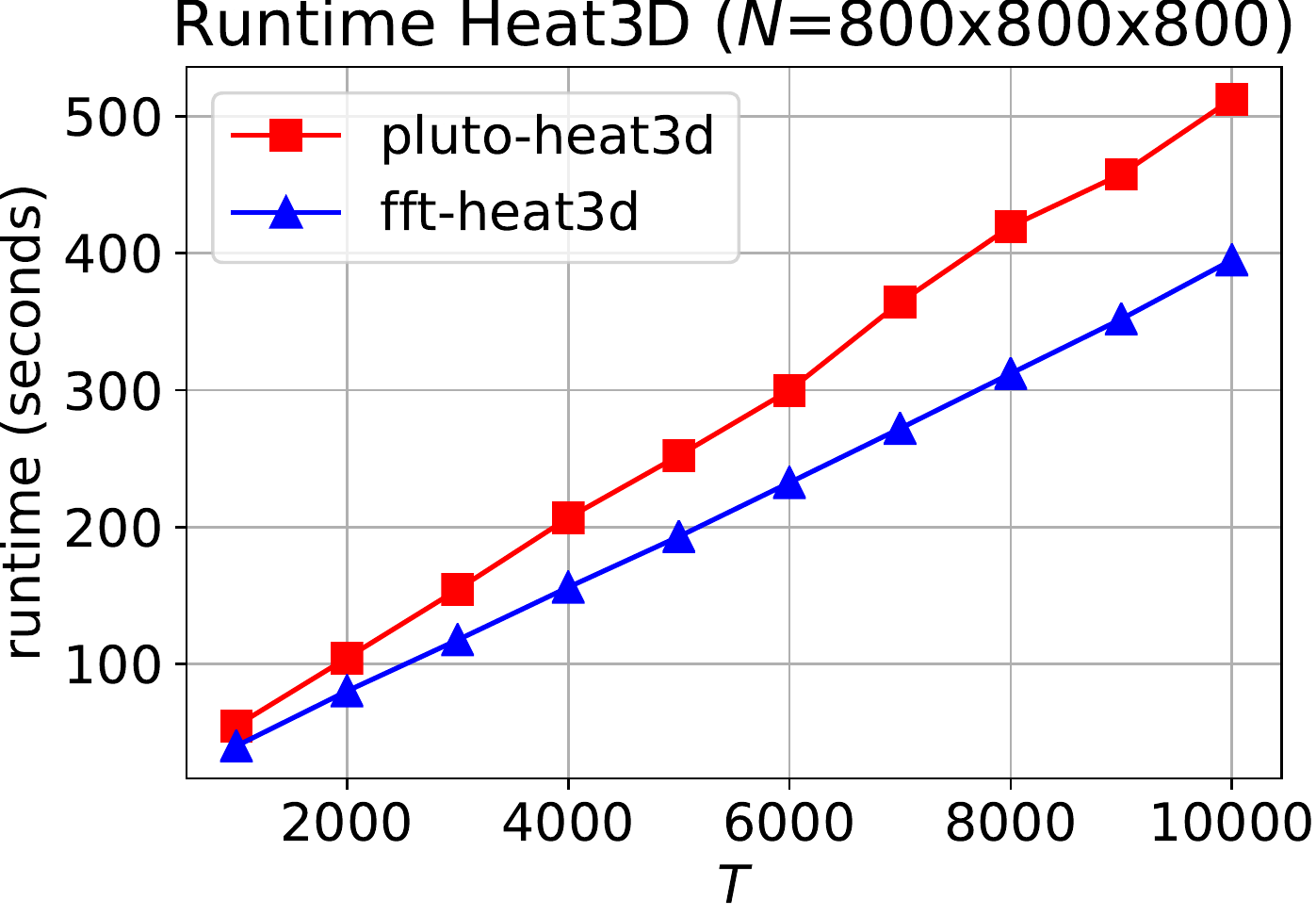}
$(i)$ & $(ii)$ & $(iii)$ \\\hline \rule{0pt}{1ex}\\[-0.5ex] 

\multicolumn{3}{c}{\textbf{KNL Node, Experiment 2}} \\[0.5ex]

\insertplotline{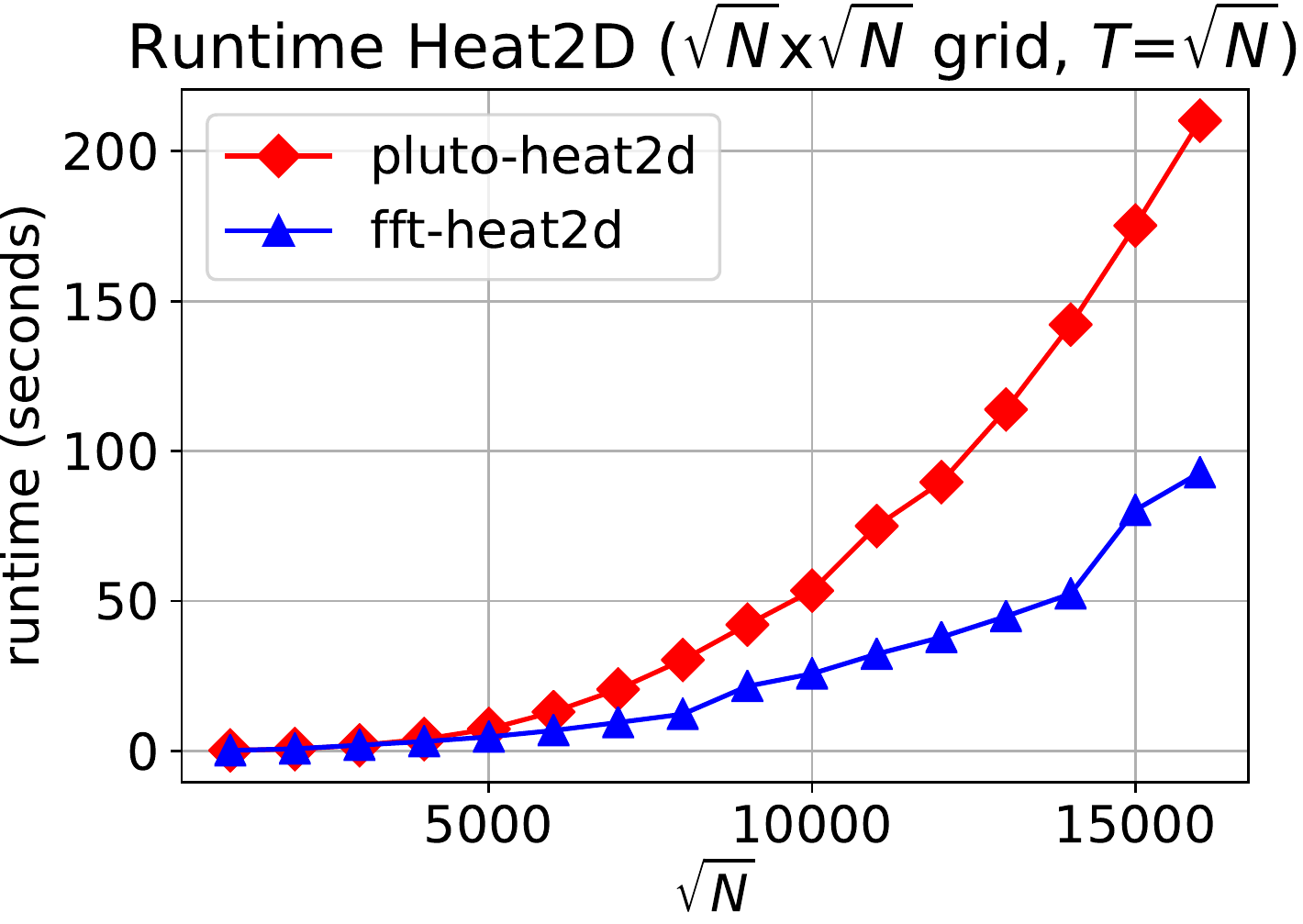}{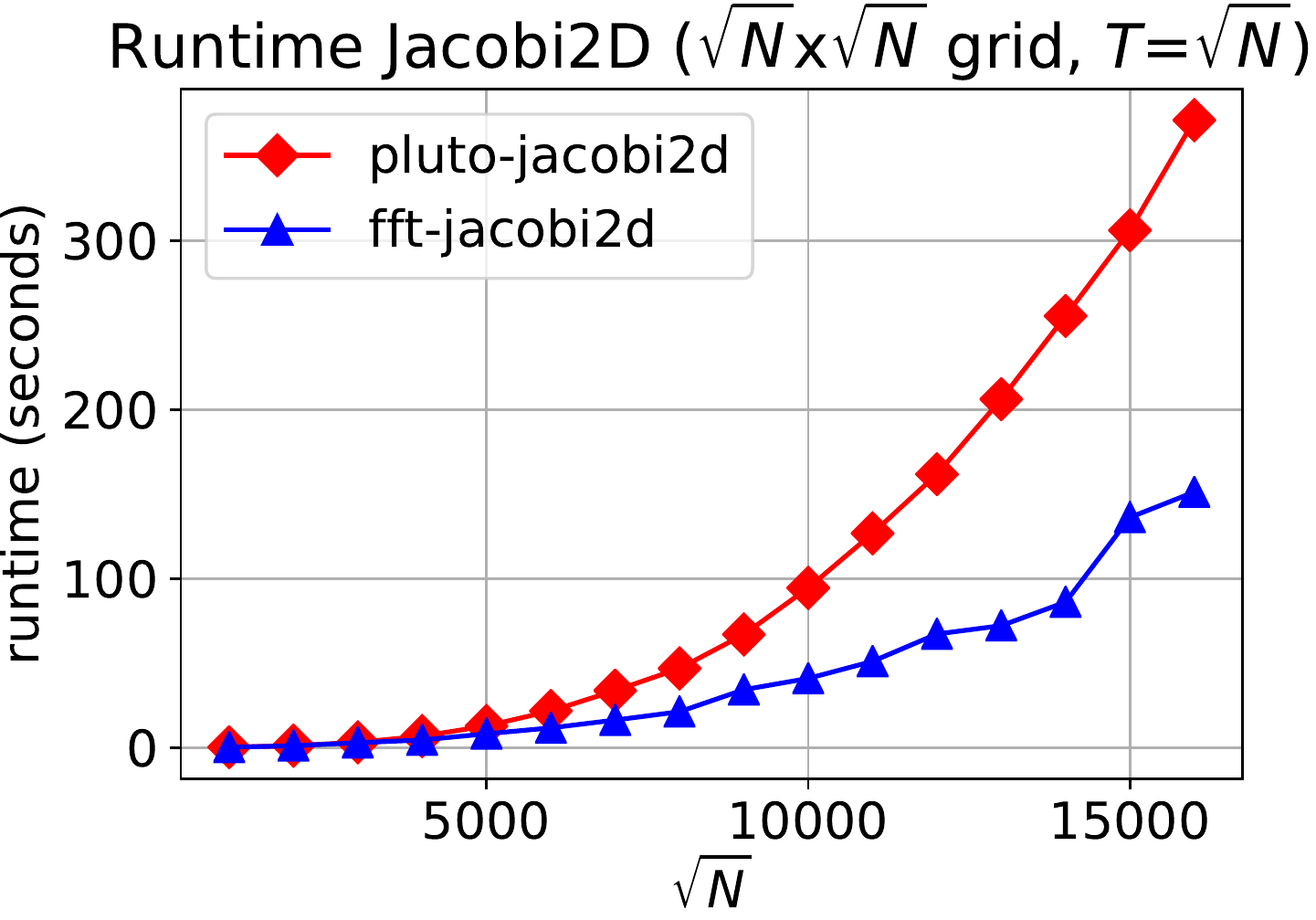}{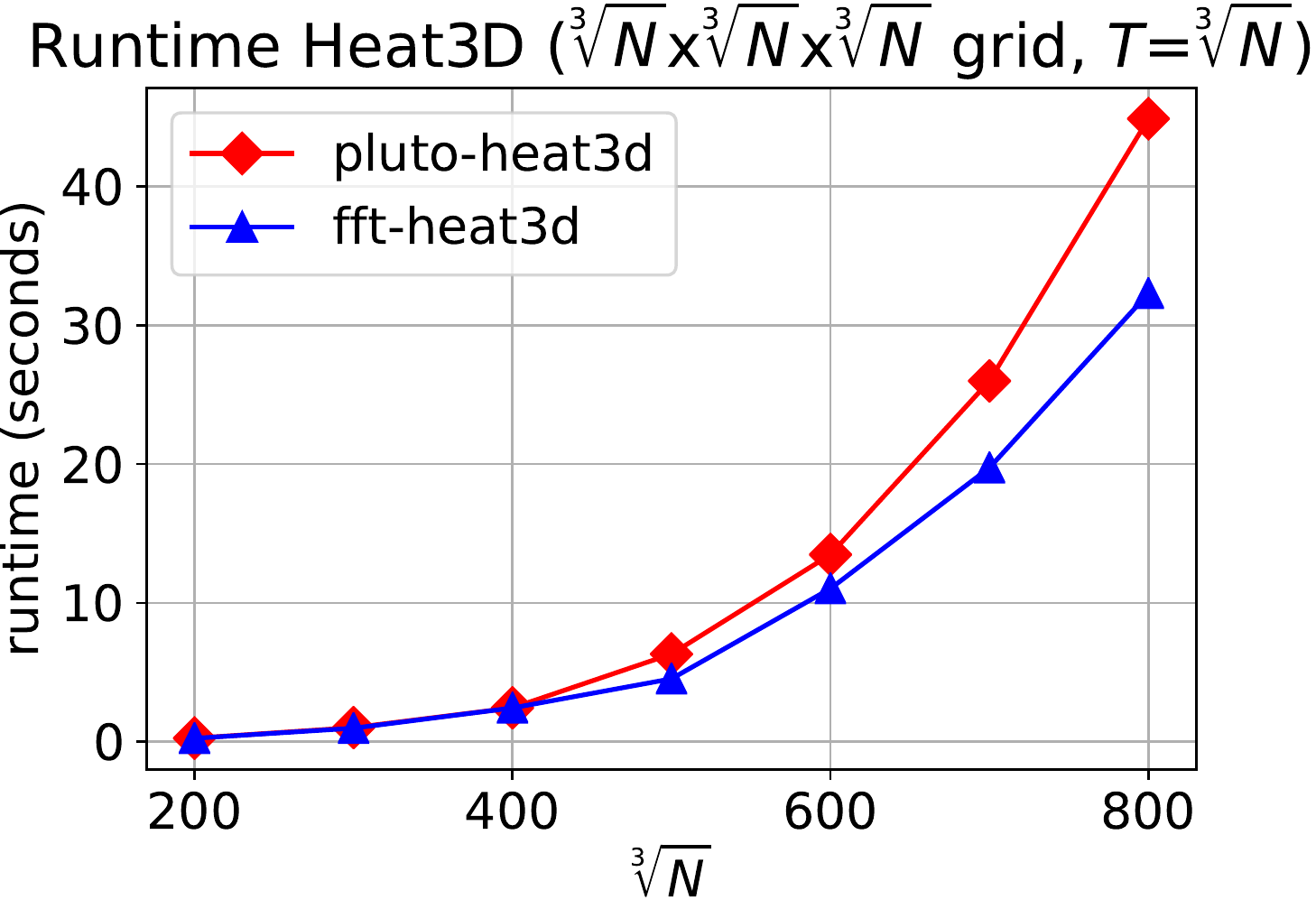}
$(iv)$ & $(v)$ & $(vi)$ \\\hline

\end{tabular}
\figcaption{Performance comparison of our FFT-based \highlight{aperiodic} algorithms with the existing best stencil programs.}
\label{fig:appendix-plots-aperiodic-knl}
\vgap{}\vgap{}\vgap{}
\end{table*}

\begin{table*}[!ht]
\begin{tabular}{ccc}
\hline\\[-1.5ex]

\multicolumn{3}{c}{\textbf{SKX Node, Experiment 1}} \\[0.5ex]

\insertplotline{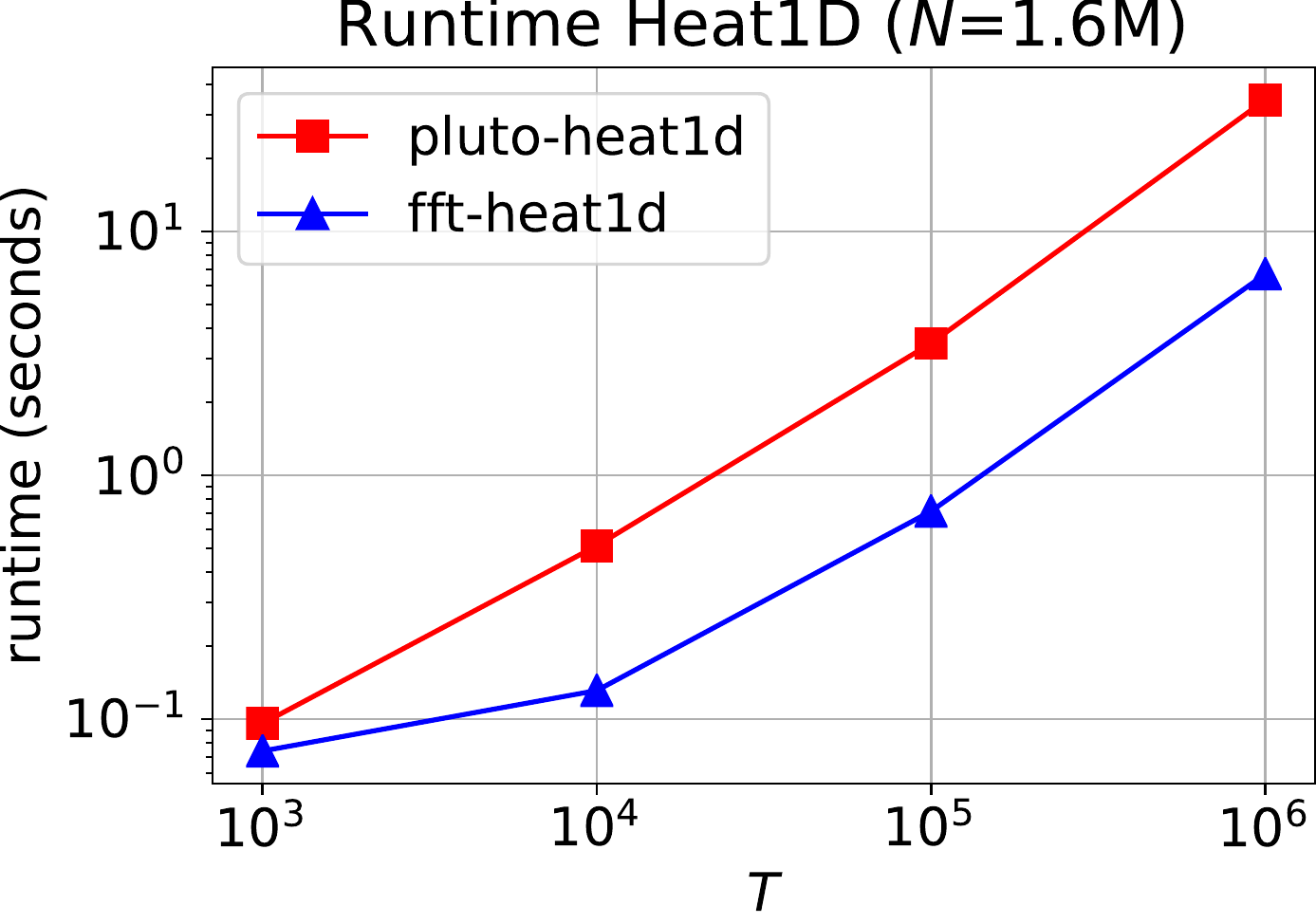}{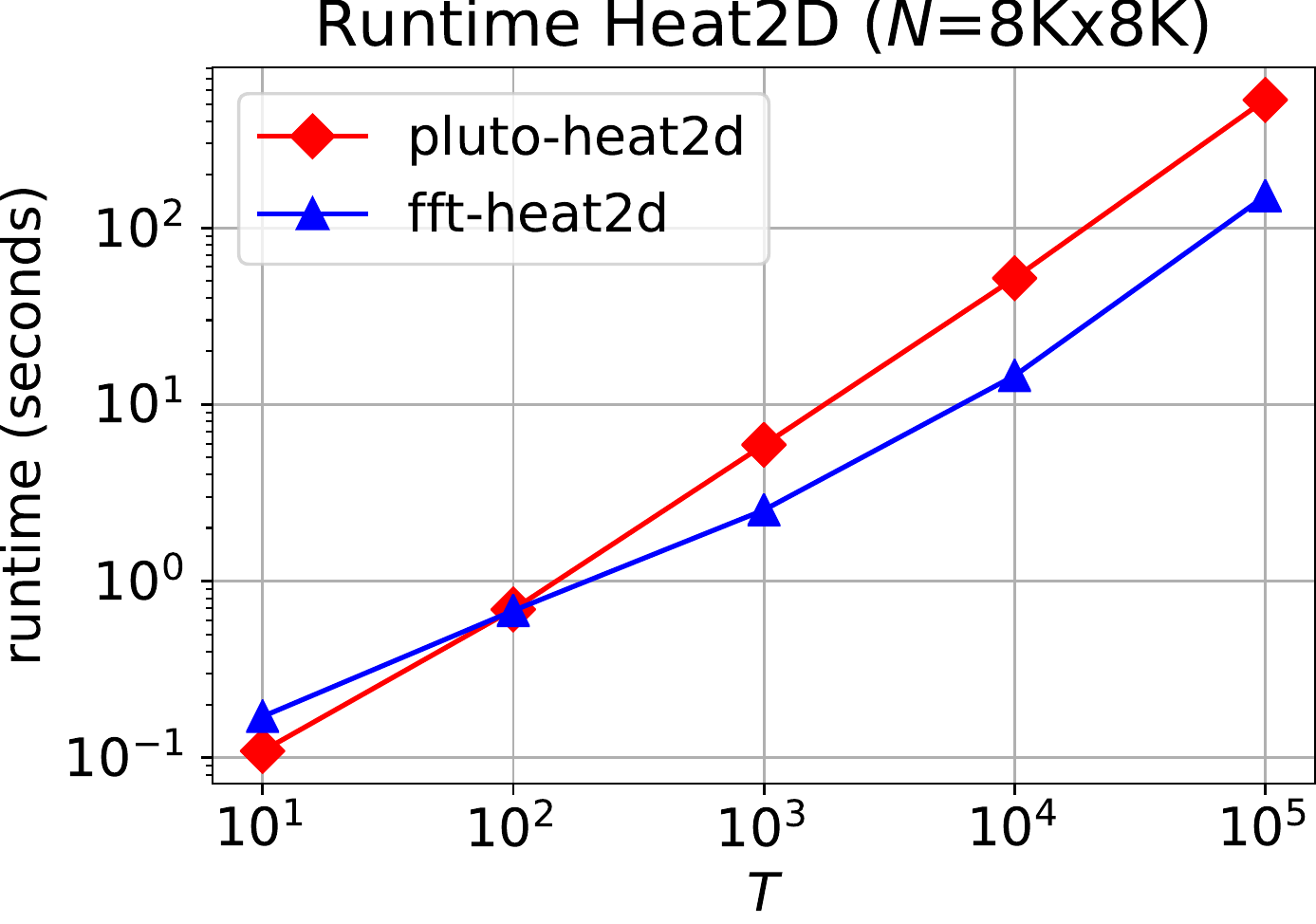}{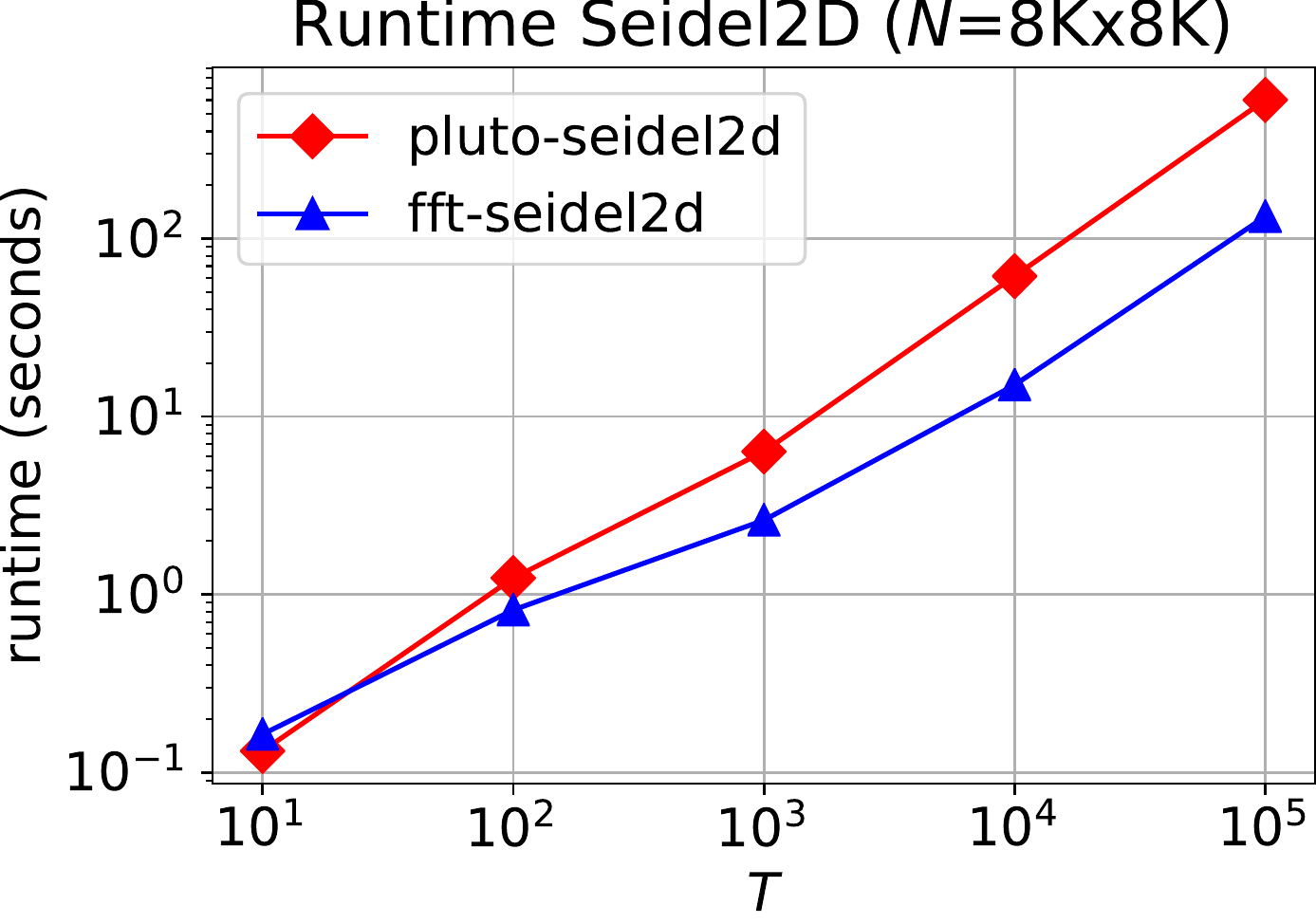}
$(i)$ & $(ii)$ & $(iii)$ \\[0.5ex]

\insertplotline{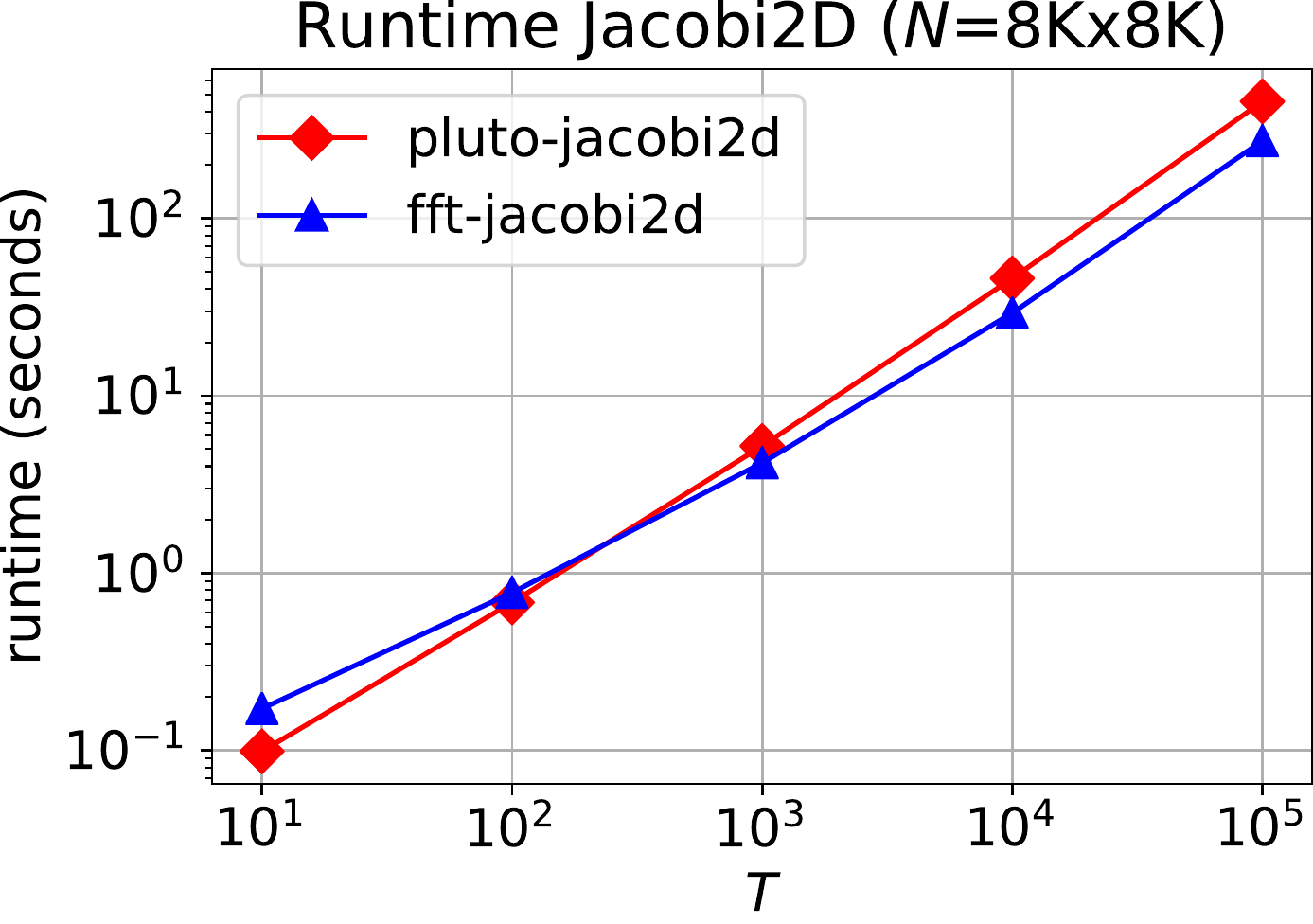}{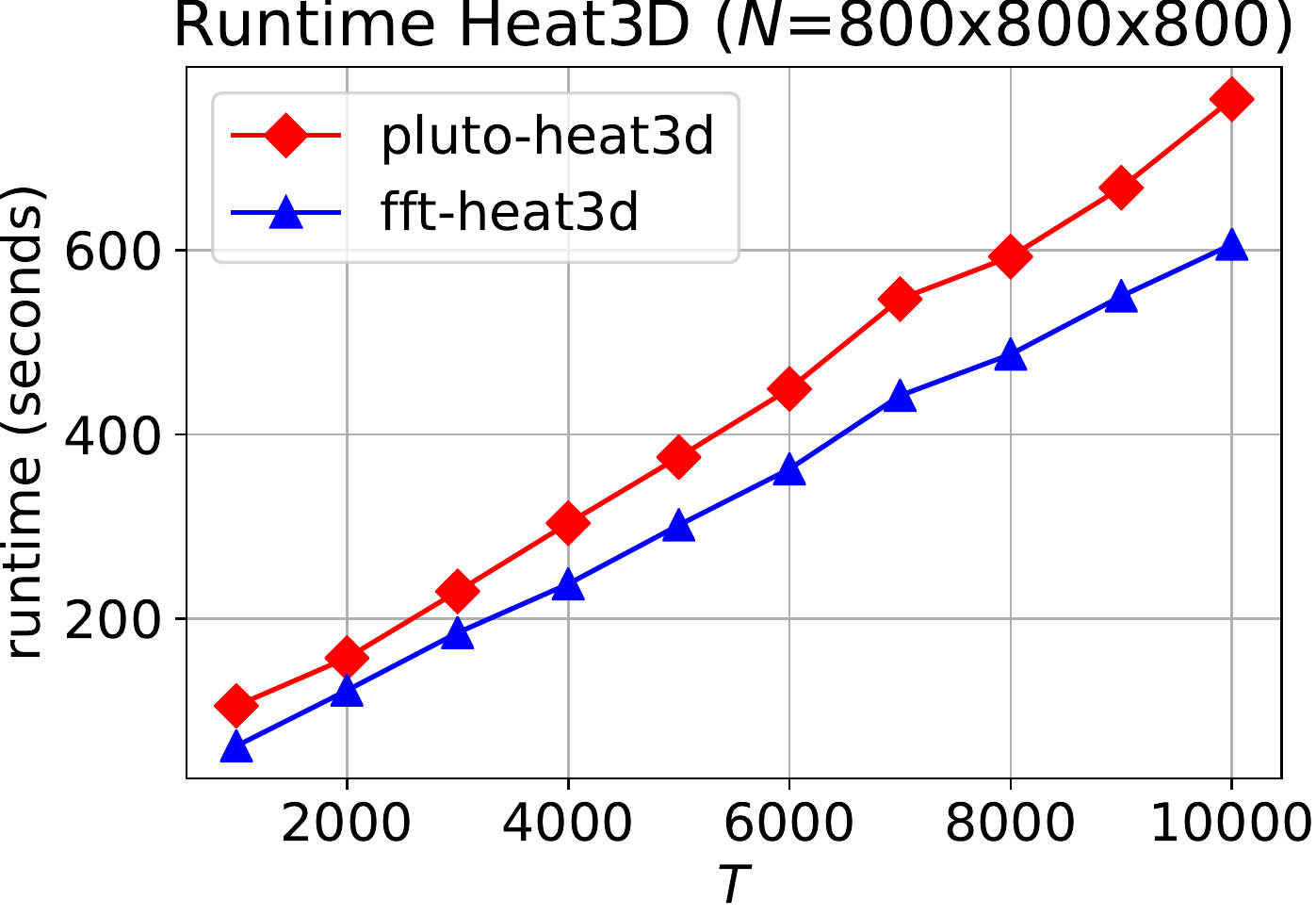}{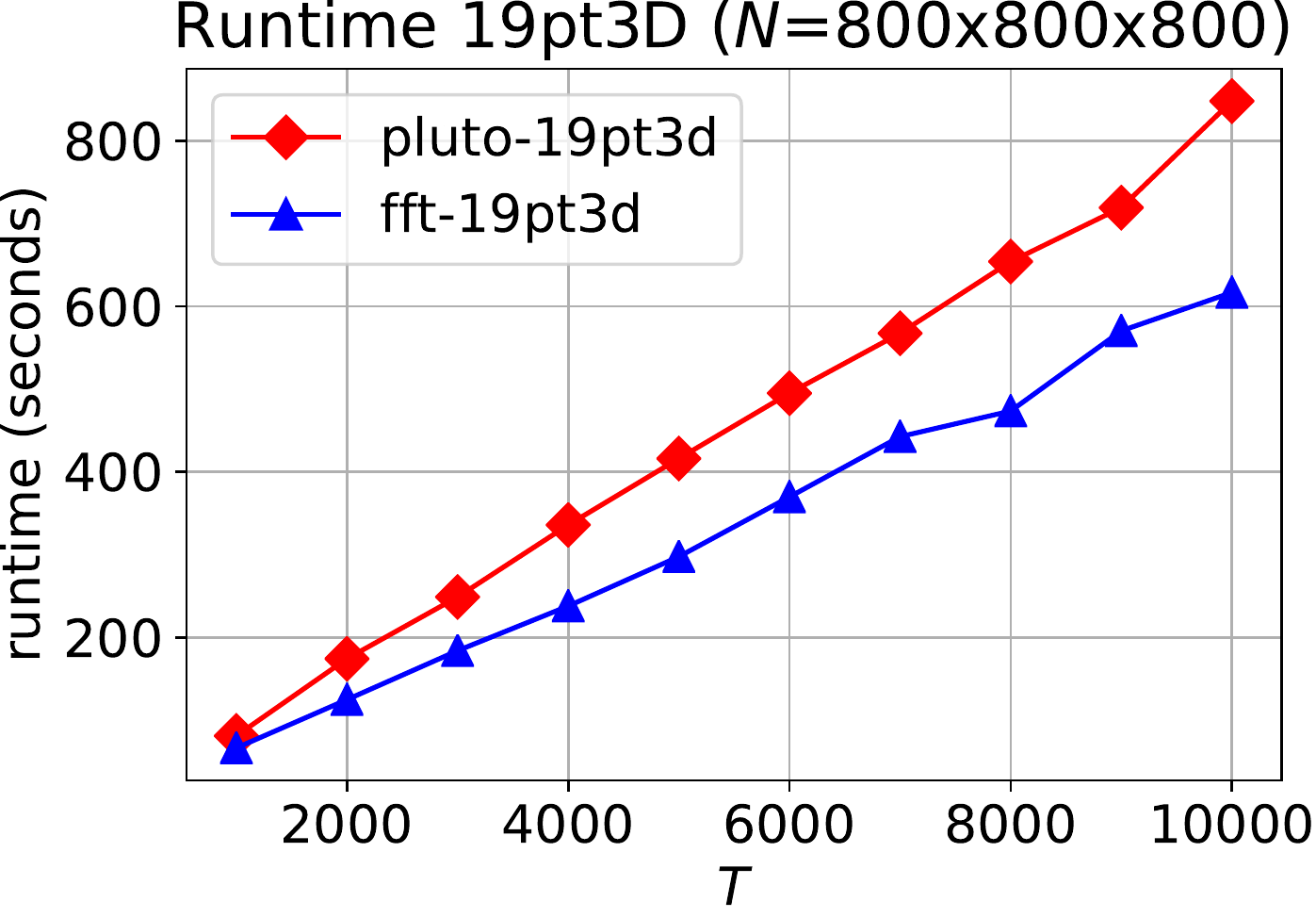}
$(iv)$ & $(v)$ & $(vi)$ \\[0.5ex]

\insertplotline{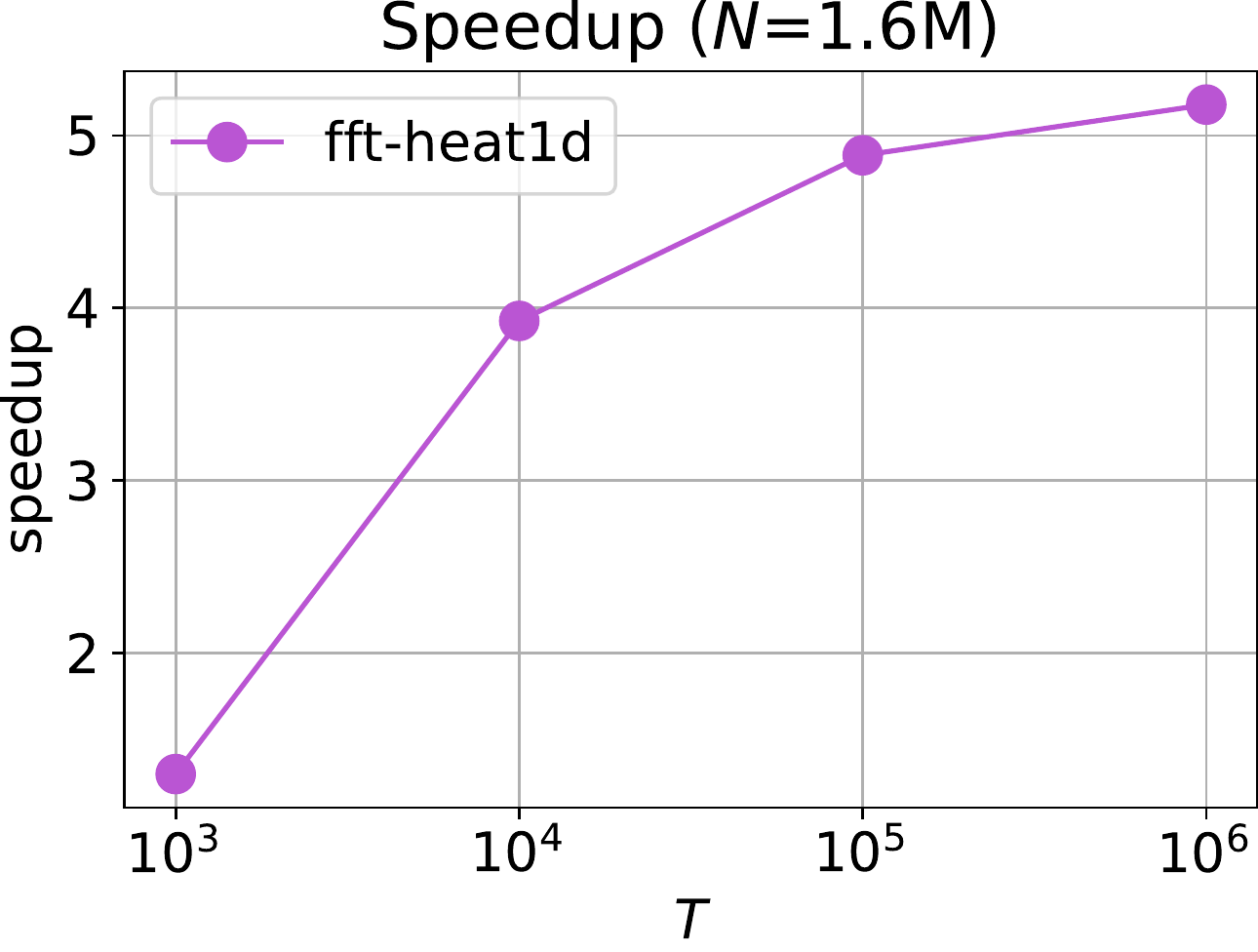}{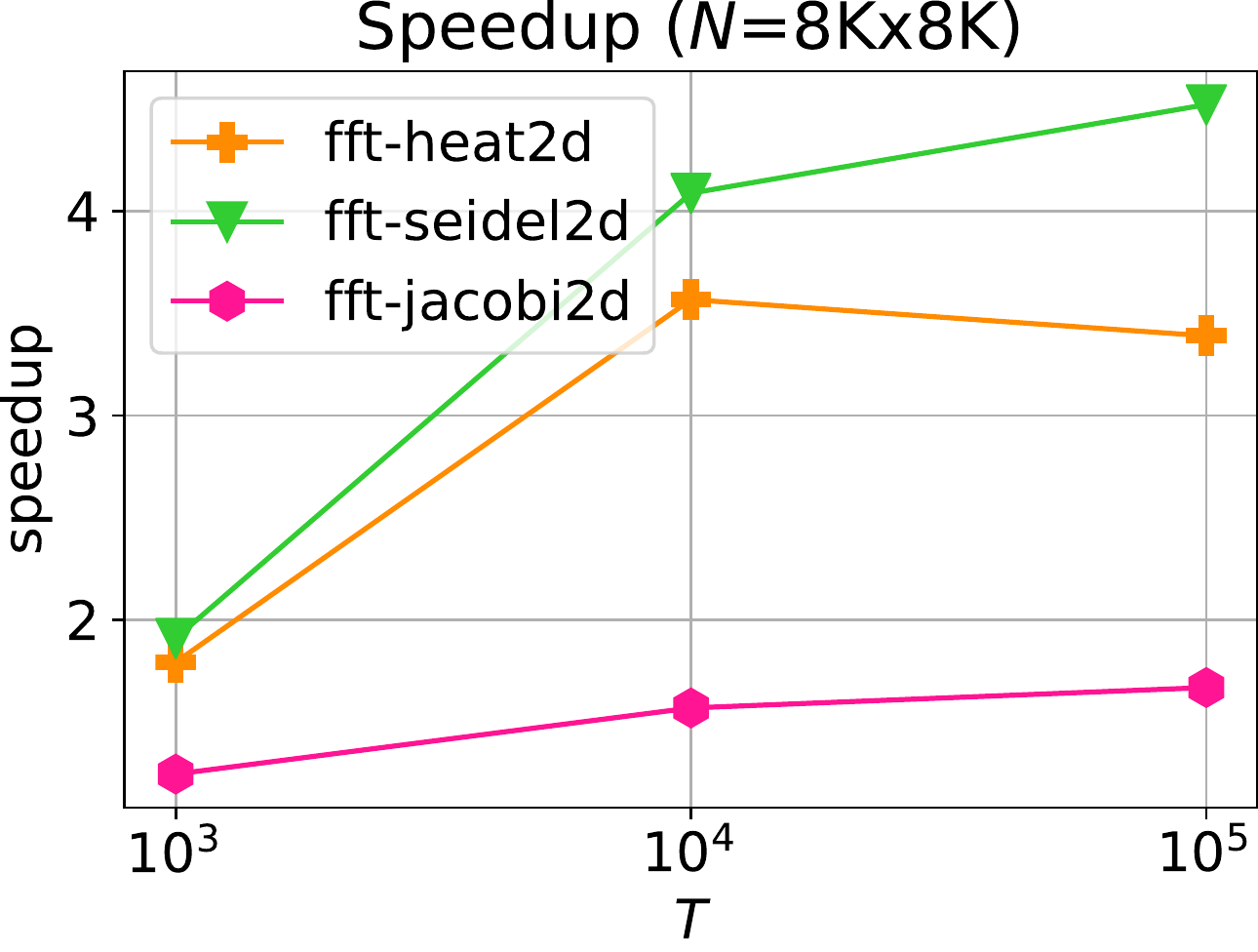}{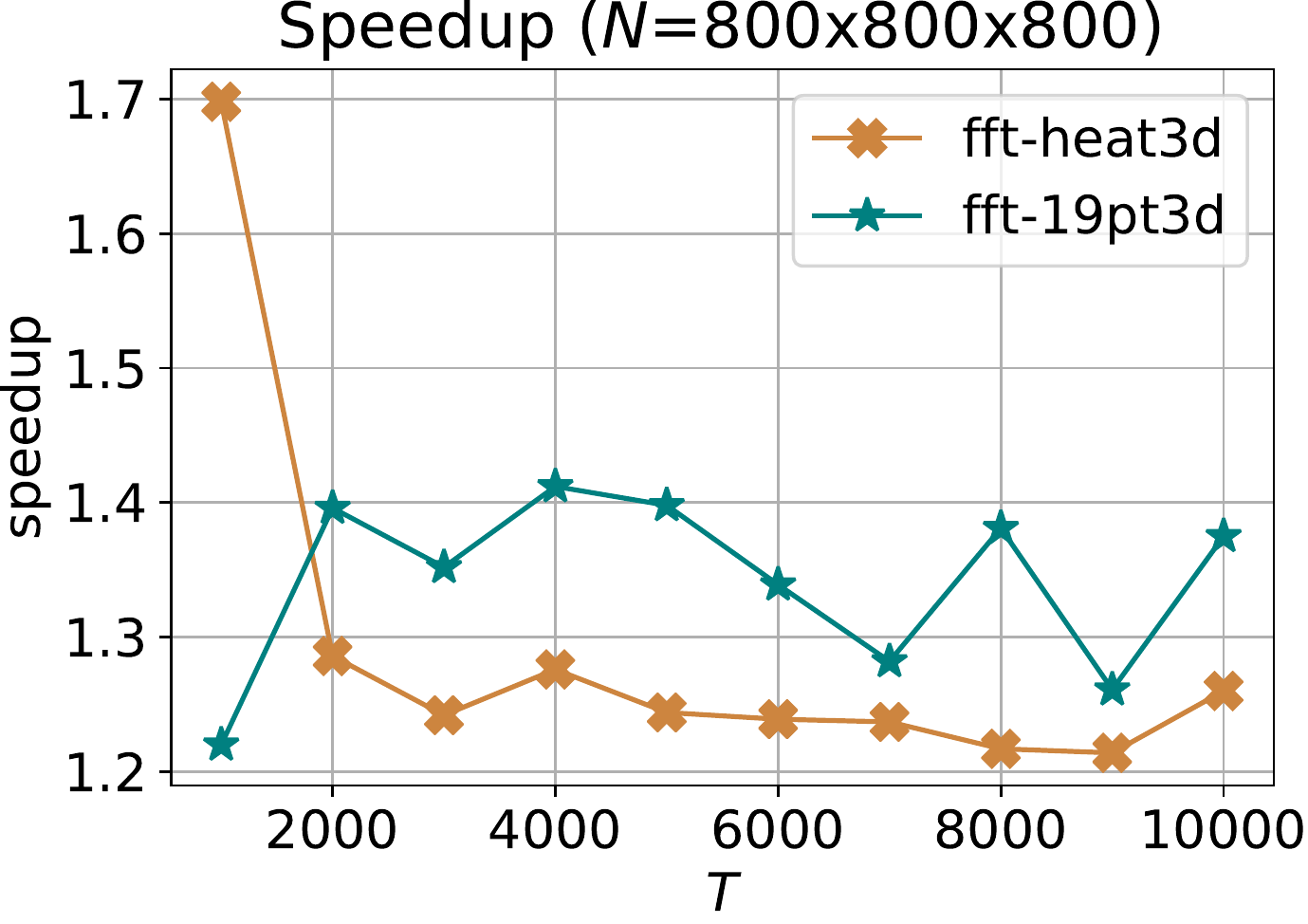}
$(vii)$ & $(viii)$ & $(ix)$ \\\hline

\end{tabular}
\figcaption{Performance comparison of our FFT-based \highlight{aperiodic}  algorithms with the existing best stencil programs.}
\label{fig:appendix-plots-aperiodic-skx-1}
\vgap{}\vgap{}\vgap{}
\end{table*}

\begin{table*}[!ht]
\begin{tabular}{ccc}
\hline\\[-1.5ex]

\multicolumn{3}{c}{\textbf{SKX Node, Experiment 2}} \\[0.5ex]

\insertplotline{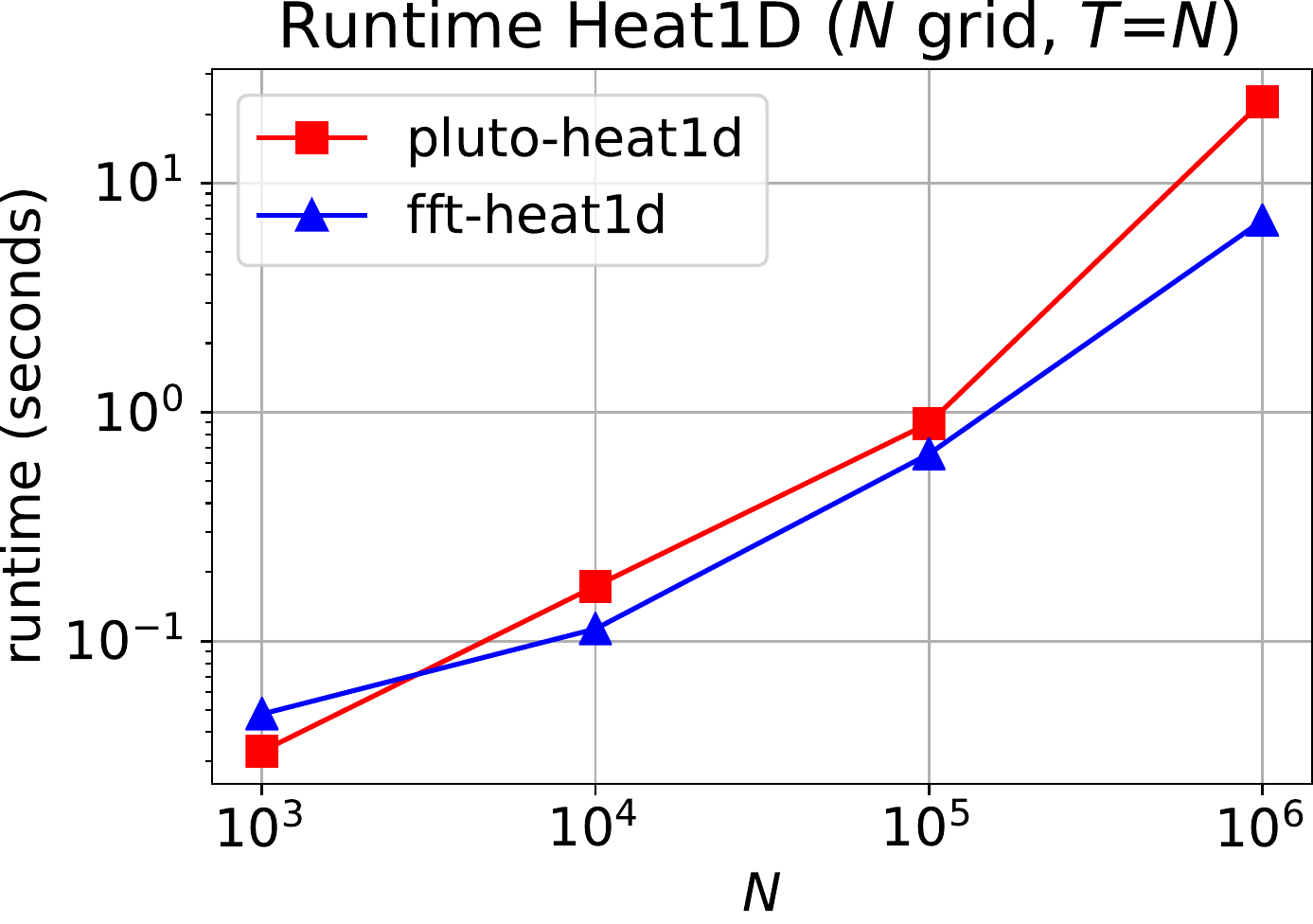}{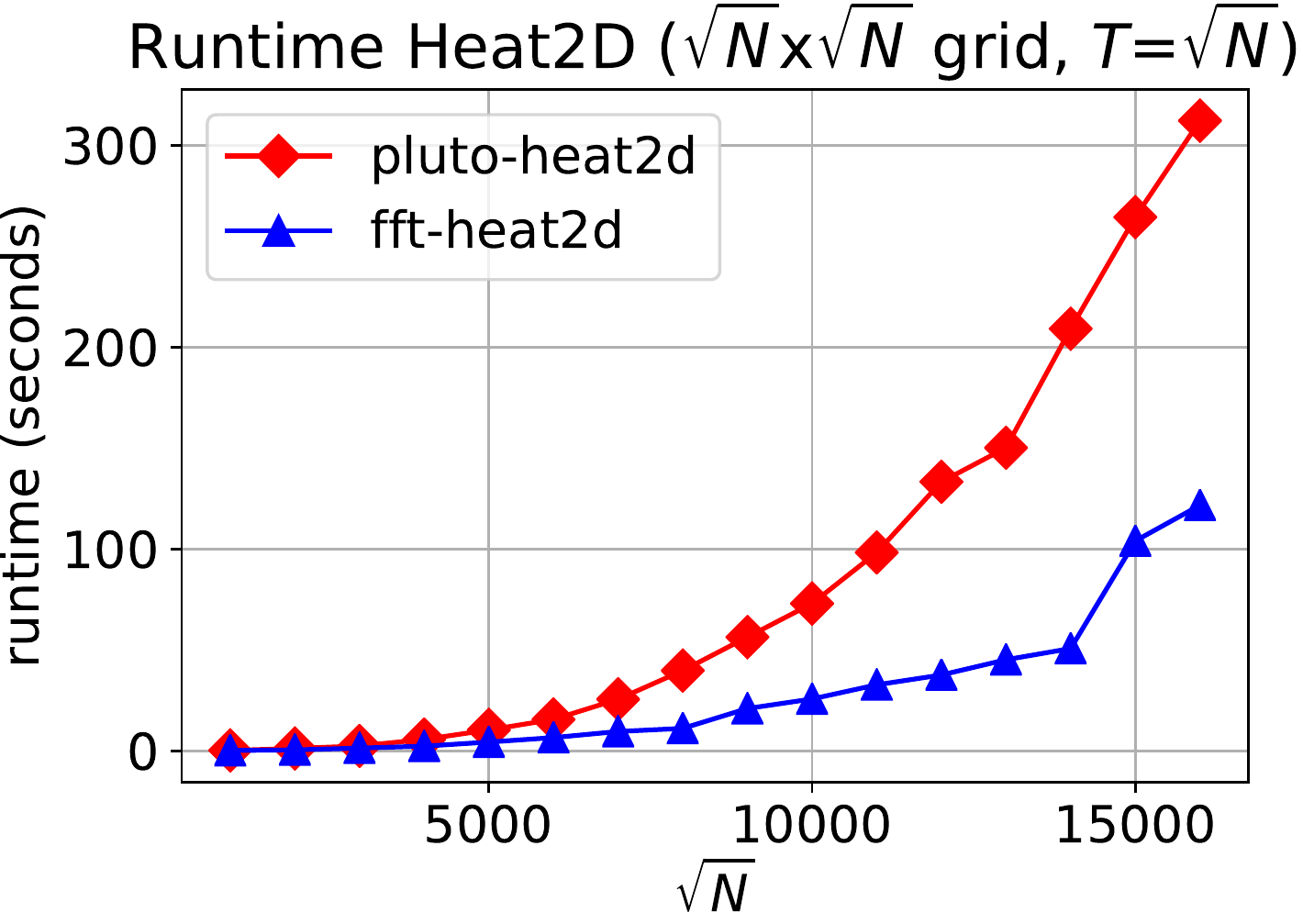}{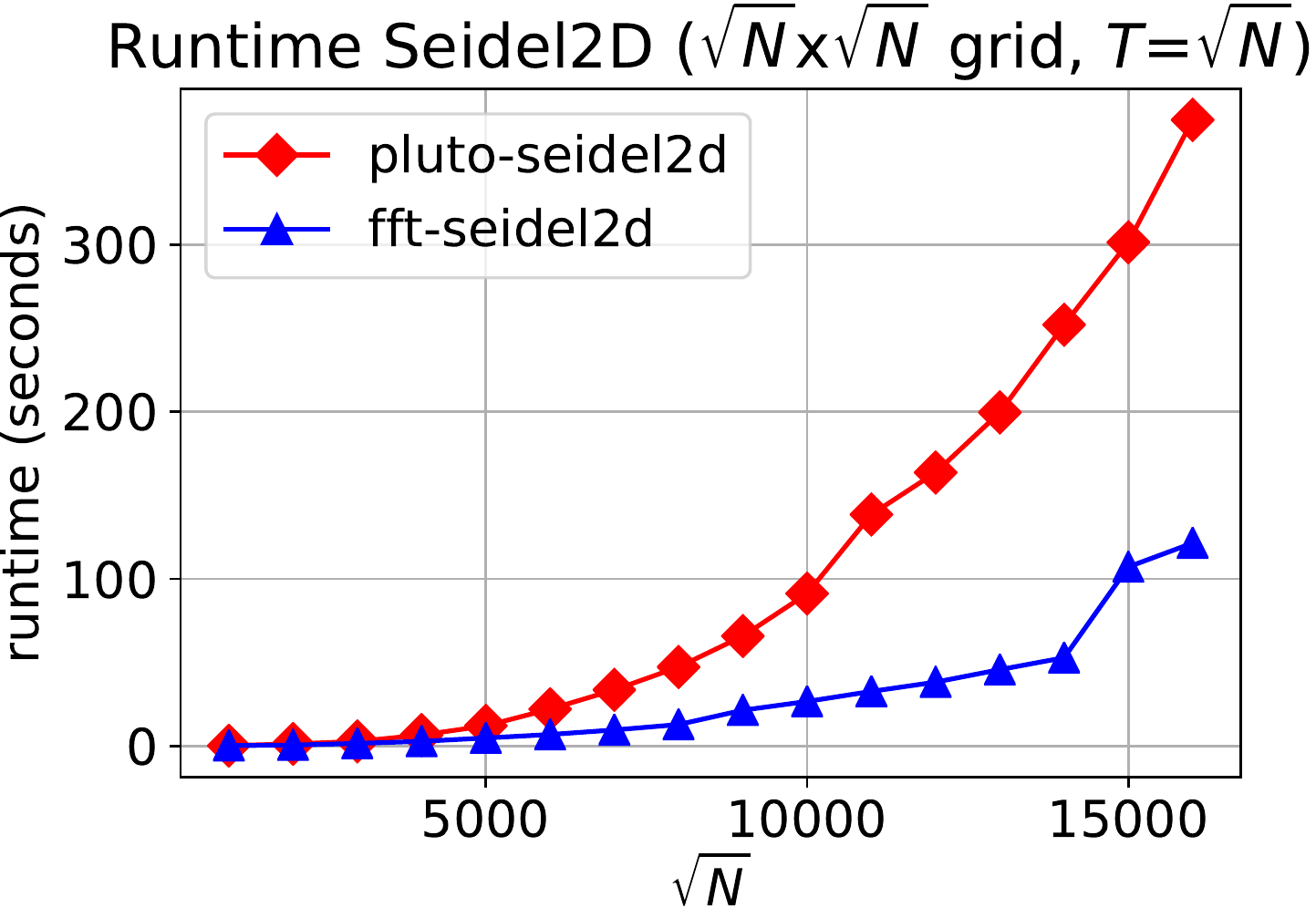}
$(i)$ & $(ii)$ & $(iii)$ \\[0.5ex]

\insertplotline{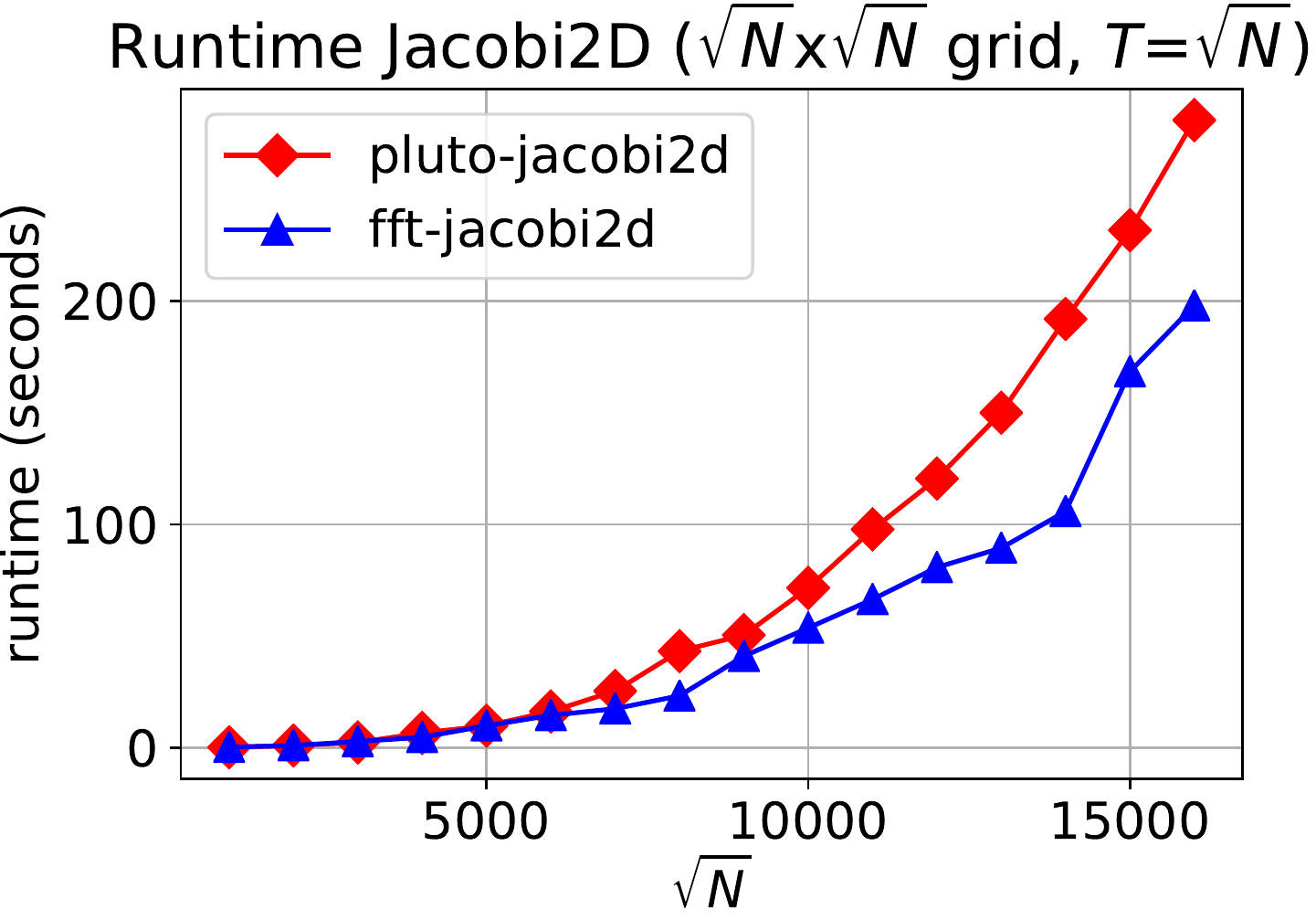}{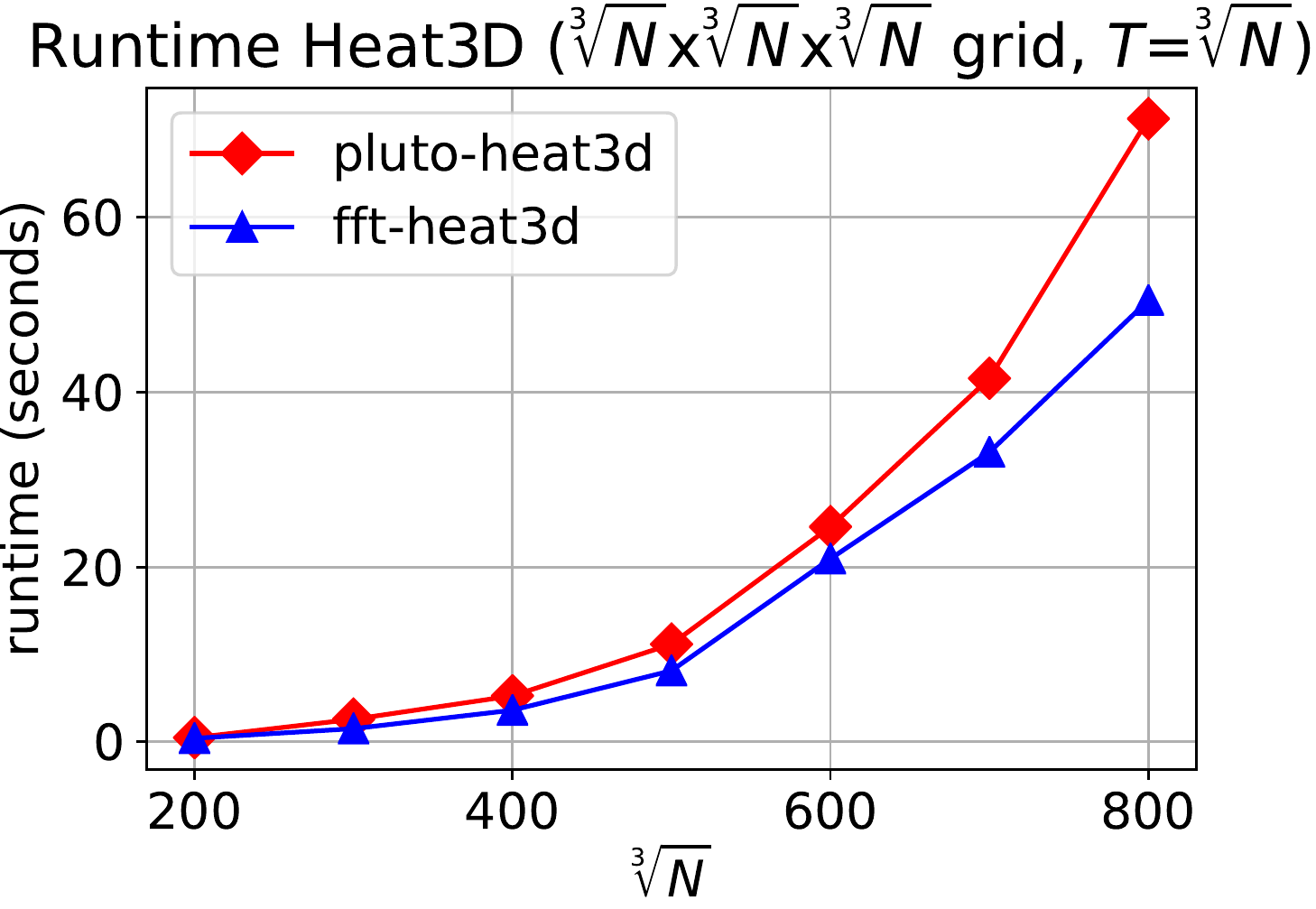}{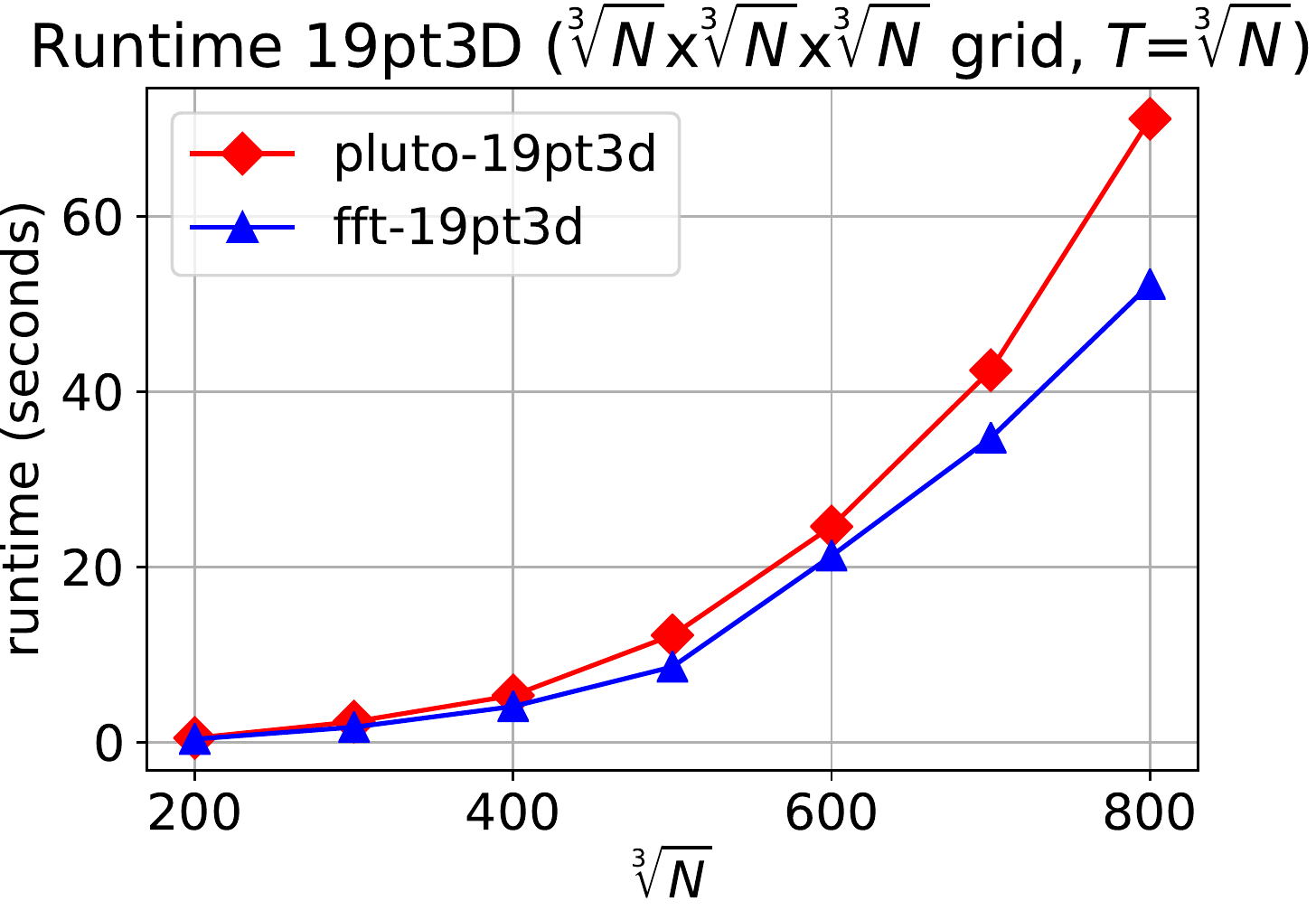}
$(iv)$ & $(v)$ & $(vi)$ \\[0.5ex]

\insertplotline{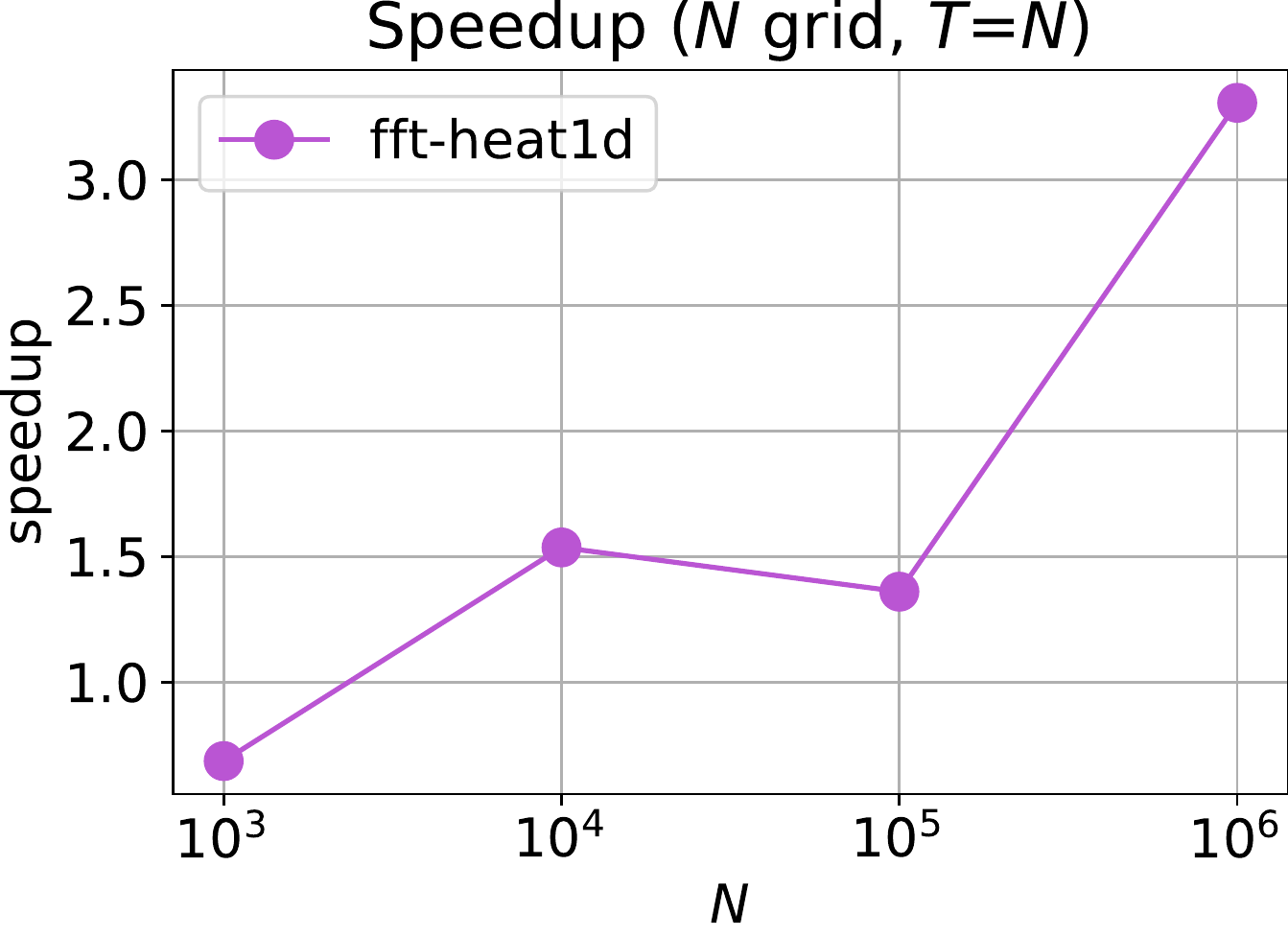}{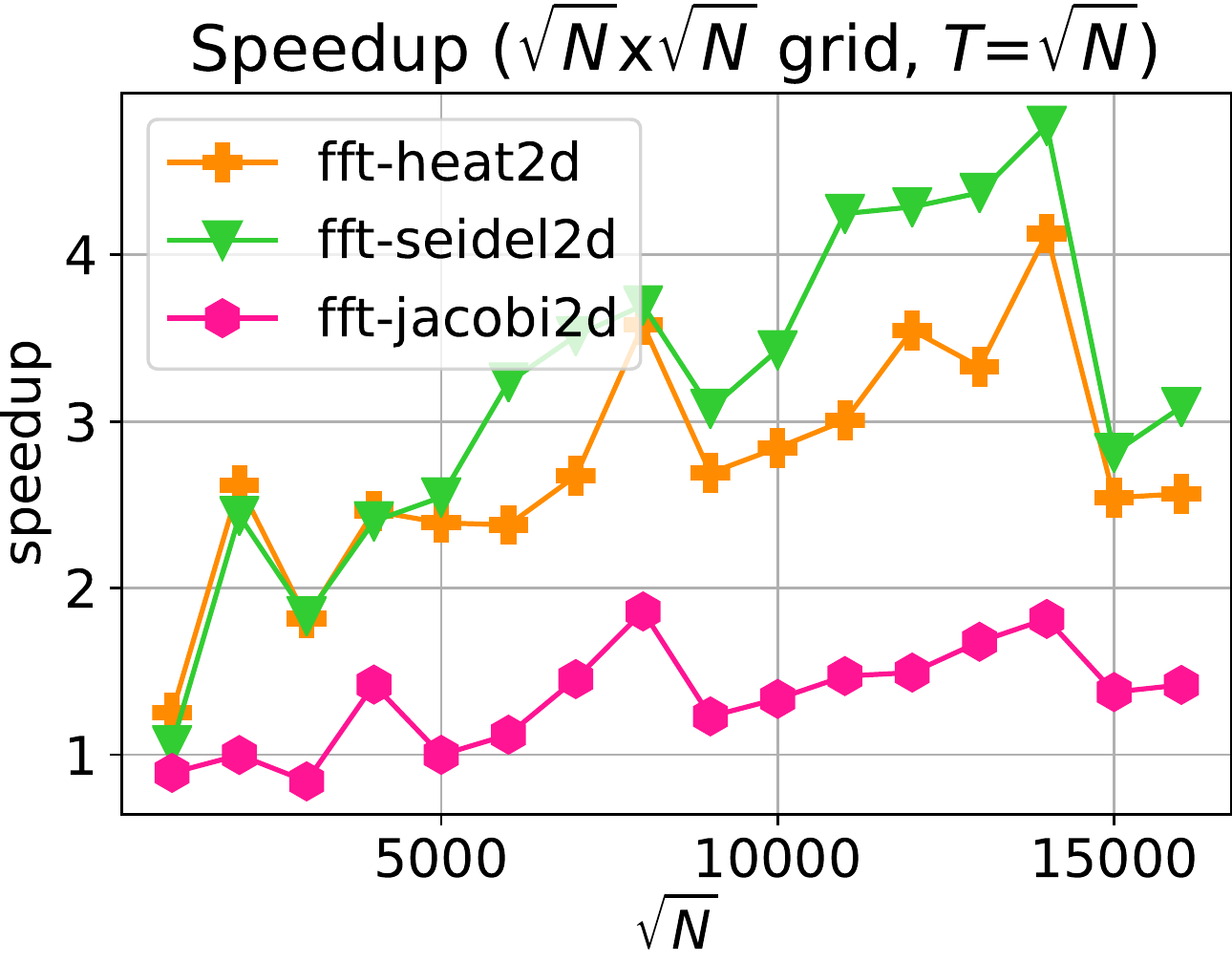}{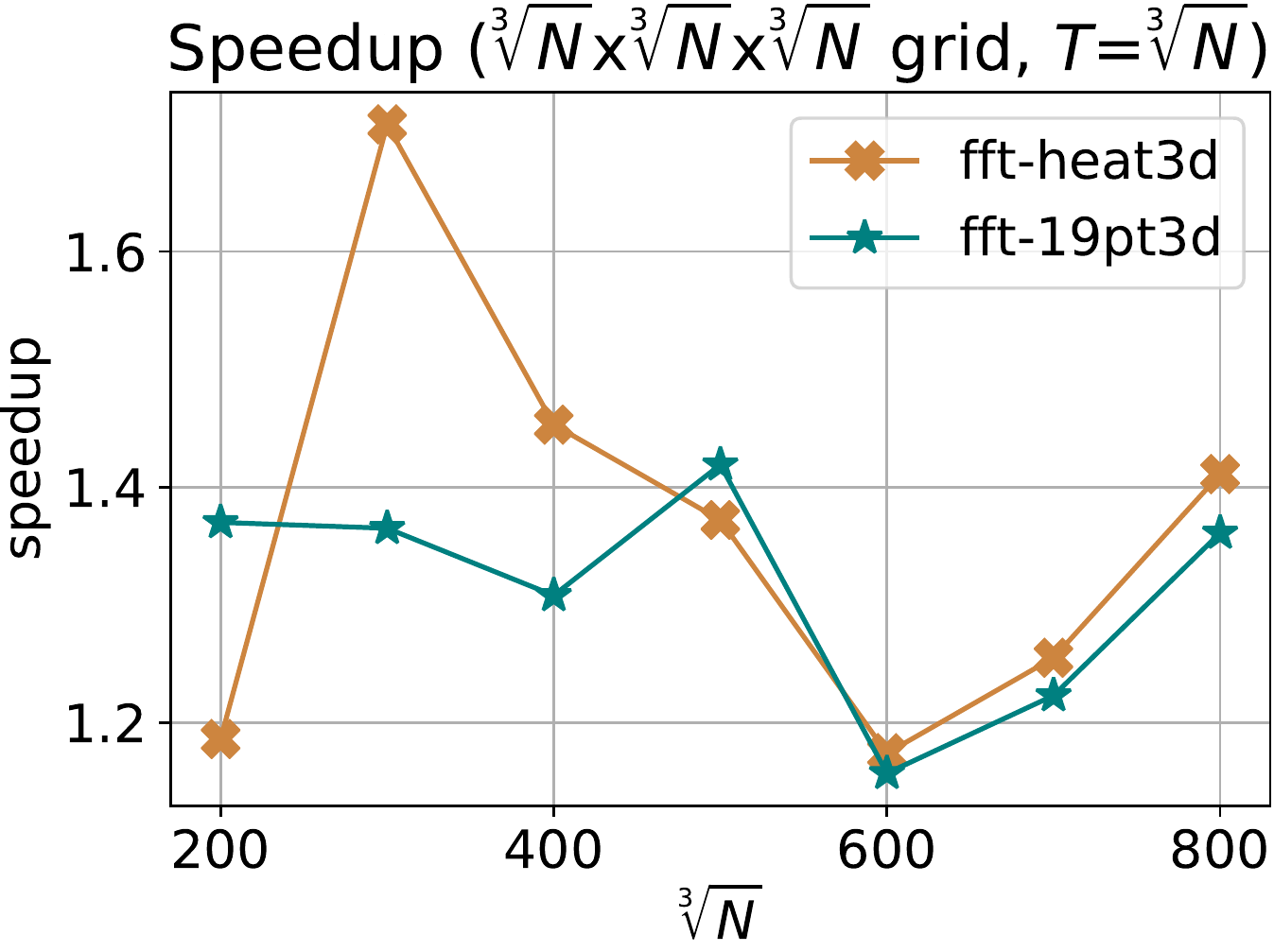}
$(vii)$ & $(viii)$ & $(ix)$ \\\hline

\end{tabular}
\figcaption{Performance comparison of our FFT-based \highlight{aperiodic} algorithms with the existing best stencil programs.}
\label{fig:appendix-plots-aperiodic-skx-2}
\vgap{}\vgap{}\vgap{}
\end{table*}

\section{Appendix}

\subsection{Supporting $d$-D Spatial Grids} \label{ssec:math_for_P_Dd}
Throughout this section we will use arguments identical to those given for the 1-D version of {\Periodic}, but with the modification that whereas the stencil used and grid data used to be indexed by an integer, they will now be indexed by a $d$-D vector $\bm{i} = [i_1, \dots, i_d]$. Let $a_0[\bm{i}]$ be the initial data for a grid of size $\ell_1 \times \cdots \times \ell_d = N$ to be acted on by stencil $S = \sum_{\bm{i}} S[\bm{i}, \bm{0}] X^{\bm{i}}$, where $X^{\bm{i}} = X_1^{i_1}\cdots X_d^{i_d}$ is a shift operator whose action is given by $(X^{\bm{k}} a)[\bm{i}] = a[\bm{i} - \bm{k}]$ and whose matrix elements are $X^{\bm{k}}[\bm{i},\bm{j}] = \delta_{\bm{i}, \bm{j} + \bm{k}}$. We have periodic boundary conditions, so $X_j$ can be viewed as cyclically permuting the $j$th dimension. Future timesteps are defined with $a_{t+1} = S a_{t} = S^{t+1} a_0$, where matrix multiplications are contractions over the internal vector indices, $(Sa)[\bm{i}] = \sum_{\bm{j}} S[\bm{i},\bm{j}] a[\bm{j}]$.

As with the 1-D case, $S$ is a circulant matrix, so we can apply DFTs over every spatial dimension to diagonalize it. These DFT matrices together make up an operator with elements given by $\mathcal{F}[\bm{i},\bm{j}] = \omega_{\ell_1}^{-i_1 j_1} \cdots \omega_{\ell_d}^{-i_d j_d}$ and $\mathcal{F}^{-1}[\bm{i},\bm{j}] = \omega_{\ell_1}^{i_1 j_1} \cdots \omega_{\ell_d}^{i_d j_d} / (\ell_1 \cdots \ell_d)$. Note of course that by ``diagonal'' here we mean that $\Lambda[\bm{i}, \bm{j}] = (\mathcal{F}S\mathcal{F}^{-1})[\bm{i}, \bm{j}] = 0$ for all $\bm{i} \neq \bm{j}$. The proof that $\Lambda$ can be found by applying a single {\MFFT} (see Figure \ref{fig:stencil-fft-algorithm-dD}) over the first column of $S$ is given in Appendix \ref{ssec:eig_fft_computing}. Here we only restate the result that $\Lambda[\bm{i}, \bm{i}] = (\mathcal{F}s)[\bm{i}]$, where $s[\bm{i}] = S[\bm{i}, \bm{0}]$.

Switching from 1-D to $d$-D results in the version of the {\Periodic} algorithm whose pseudocode given in Figure \ref{fig:stencil-fft-algorithm-dD} in the main body of this paper. Since {\MFFT} is of equivalent work, span, and serial cache complexity to a standard FFT \cite{frigo2005design}, this version of the {\Periodic} algorithm is of identical complexity to that already shown in Theorem \ref{th:periodic-stencil-fft}.

\subsection{Stencil Matrix Example} \label{ssec:stencil_matrix_walkthrough}
Suppose we have a 1-D periodic spatial grid of 4 cells, and we want to act on it with a stencil $S$ that computes the new value of a cell according to $a_{t+1}[i] = -2a_t[i-1] + a_{t}[i] + 3 a_t[i+1]$. We can write out the update equation in full as

$$\begin{sbmatrix}{a_{t+1}} a_{t+1}[0] \\ a_{t+1}[1] \\ a_{t+1}[2] \\ a_{t+1}[3]\end{sbmatrix} = \begin{sbmatrix}{S} 1 & 3 & 0 & -2 \\
-2 & 1 & 3 & 0 \\
0 & -2 & 1 & 3 \\
3 & 0 & -2 & 1 \end{sbmatrix} \begin{sbmatrix}{a_{t}} a_{t}[0] \\ a_{t}[1] \\ a_{t}[2] \\ a_{t}[3]\end{sbmatrix}.$$

\subsection{Proof of Theorem \ref{th:aperiodic-stencil}} \label{ssec:proof-aperiodic-stencil}

\begin{proof}
The complexities for {\Periodic} have already been given in Theorem \ref{th:periodic-stencil-fft}, so here we need only derive the complexities for the {\Boundary} computation algorithm.

Let the spatial grid be of size $\ell_1 \times \cdots \times \ell_d = N$, with $d = \Th{1}$. We bound the number of cells in the boundary's region of interest by $b \sigma T$, where $b = 2 N (\ell_1^{-1} + \cdots + \ell_d^{-1}) = \Om{N^{1 - 1/d}}$ is the size of the spatial grid's boundary and $\sigma = \Th{1}$ is the size of the stencil being applied.\\
\noindent
[Work.] At every stage of the divide and conquer algorithm we make two recursive calls and apply the periodic solver to $\Th{b T}$ grid cells. This gives us the recurrence
$$T_1(T) = \begin{cases}
\Th{b} & T < c, \\
2T_1(T/2) + \Th{bT \log (bT)} & T \geq c,
\end{cases}$$
for some positive constant $c$. This gives $T_1(T) = \LTh{bT \log (bT)}$ \\$\RTh{\log T}$. Taking into account the $\Th{N \log (TN)}$ work that will be done by the periodic solver for the interior region gives a final work bound of $\Th{bT \log (bT) \log T + N \log N}$.
\hide{
\begin{align*}
    W(T) &= \begin{cases}
    \Th{\sigma} & T < c \\
    2W(T/2) + \Th{\sigma T \log (\sigma T)} & T \geq c
    \end{cases} \\
    &= \Th{\sigma T + \sum_{i = 0}^{\log T} \sigma T \log (\sigma T 2^{-i})} \\
    &= \Th{\sigma T \sum_{i = 0}^{\log T} \log (\sigma) + i} \\
    &= \Th{\sigma T \log (\sigma T) \log T}
\end{align*}
}

\noindent
[Span.] At every level of recursion, both calls to {\Boundary} are done in sequence, but in parallel with the $\Th{1}$ associated periodic solver calls. The recurrence for span is thus
$$T_\infty(T) = \max \{2 T_\infty(T/2), c' \log T \log \log T\} + \Th{\log b},$$
where $c' = \Th{1}$. This gives $T_\infty(T) = \Th{T \log b}$, and incorporating the $\Th{\log T + \log N \log \log N}$ span from the periodic solver across the interior region gives a final bound of $\Th{T \log b + \log N \log \log N)}$.

\end{proof}

\subsection{Supporting Implicit Stencils} \label{ssec:implicit_stencils}
Here we show how to extend {\Periodic} to handle stencils which depend implicitly on field data from the timestep which is being computed. This will be done by computing a pseudoinverse after diagonalizing the implicit part of the stencil.

Let $S$ and $Q$ be explicit stencils, and suppose we wish to solve a stencil problem with update equation $Q a_{t+1} = S a_{t}$. First we diagonalize $\fft Q \ifft = \Lambda_Q$ and $\fft S \ifft = \Lambda_S$ by pushing them to the frequency domain, after which we will take the pseudoinverse of $\Lambda_Q$, defined as $\Lambda_Q^+[i,i] = \Lambda_Q[i,i]^{-1}$ if $\Lambda_Q[i,i] \neq 0$, and 0 otherwise. Now we find a new diagonalized stencil $\Lambda_R = \Lambda_Q^+ \Lambda_S$. This new stencil $R$ has the following property: if $a_{t+1} = R a_t$, then $Q a_{t+1} = S a_t$.

Thus we can evolve data according to $R$, and the evolved data will satisfy the original implicit stencil equation.

\subsection{Supporting Vector Valued Fields} \label{ssec:vector_valued_fields}
Here we show how to alter our periodic algorithm to handle vector-valued fields, i.e. fields where the grid data for each cell is an array of fixed length instead of being just a scalar. Consider a set of scalar-valued fields $\{a^{(i)}_t\}$ which evolve according to linear stencils as
\begin{align*}
    a^{(i)}_{t+1} = \sum_j S^{(i,j)} a^{(j)}_{t},
\end{align*}
where $S^{(i,j)}$ is a circulant stencil matrix describing how $a_{t+1}^{(i)}$ is dependent on $a_t^{(j)}$. As in the scalar-field case, we can diagonalize all the stencils $S^{(i,j)}$ by moving them to the frequency domain.

Since we now have more than one stencil, we have to revisit our treatment of sequential squaring. Let $\{R^{(i,j)}\}$ be a set of stencils which evolves $a_t$ forward $r$ timesteps to $a_{t+r}$,
$$a^{(i)}_{t+r} = \sum_j R^{(i,j)} a^{(j)}_{t}.$$
We can find a new set of stencils $\{Q^{(i,j)}\}$ which evolve data forward $2r$ timesteps by reading them off of
$$\sum_j Q^{(i,k)} a^{(k)}_{t} = a^{(i)}_{t+2r} = \sum_j R^{(i,j)} a^{(j)}_{t+r} = \sum_j R^{(i,j)} \sum_k R^{(j,k)} a^{(k)}_{t}.$$
This gives the pleasing result that
$$Q^{(i,k)} = \sum_j R^{(i,j)} R^{(j,k)},$$
so all we have to do to support vector-valued fields is to swap out our sequential squaring of $\Lambda$ with a sequential squaring of a matrix of stencils.

Furthermore, when there are only a constant number of scalar fields $\{a^{(i)}_t\}$ that make up our grid data, the computational complexity of finding $\{Q^{(i,j)}\}$ differs only by a constant factor from squaring $\Lambda$ in our scalar-field version of the algorithm.

As an example of what can be done with vector-valued fields, suppose we want to implement an affine stencil $Aa_t = Sa_t + c$ on some originally scalar field, where $S$ is a linear stencil and $c$ is a constant. We could then add a spatial field (making the underlying data consist of vectors with 2 elements each) and define $a^{(1)}_0 = c$, $S^{(1,1)} = S^{(0,1)} = I$, $S^{(1,0)} = 0$, and $S^{(0,0)} = A$. This realized the behavior of the affine stencil on $a^{(0)}_t$.

\subsection{Proof of Shift Matrix Decomposition of Circulant Matrices} \label{ssec:shift_matrix_decomposition}

Let $S$ be a circulant matrix, and let $X$ be the right shift matrix. Then for any vector $a$, we have

\begin{align*}
    (Sa)[i] &= \sum_j S[i,j]a[j] 
    = \sum_j S[i-j,0] a[j] \\
    &= \sum_j S[j, 0] a[i-j] 
    = \sum_j S[j,0] (X^j a)[i] \\
    &= \left(\left(\sum_j S[j,0] X^j\right)a\right)[i],
\end{align*}

which shows that $S = \sum_i S[i,0]X^i$.

\subsection{$\ell$-shell Spatial Grid Decomposition} \label{ssec:ell_shells}

In the main body of this paper we assume that stencils with some radius $\sigma$ use the values of cells at distance $\sigma$ in any particular direction. However, this is not necessarily the case; one could imagine a stencil that requires a large number of values along one dimension and only a few along another. Upwind stencils also break this pattern, requiring values from very asymmetric neighborhoods of cells.

The concept of region of influence used in our derivation of {\Boundary} is itself a sufficient basis for defining the regions $A_{1,2}, B_{1,2,3}$, and $C_{1,2}$ shown in Figure \ref{fig:alg_fine}. Let us define an $\ell$-shell of the boundary of a spatial grid to be the spatial region consisting of all cells that enter the boundary's region of influence after exactly $\ell$ timesteps. Obviously, there can be only one time $\ell$ when a cell begins to be influenced by the boundary; the set of all $\ell$-shells thus fill the spatial grid without overlapping one another.

To generalize the regions shown in Figure \ref{fig:alg_fine} to those for arbitrary stencils, all we have to do is switch the interpretation of the horizontal axis from ``distance'' to ``$\ell$'' and scale it by setting $\sigma = 1$. This yields an algorithm that is more efficient for upwind schemes and other biased stencils.

\subsection{Proof of Eigenvalue-FFT Relation} \label{ssec:eig_fft_computing}

Expanding the diagonal matrix of eigenvalues $\mathcal{F}A\mathcal{F}^{-1} = \Lambda$ directly gives us

\begin{align*}
    \Lambda[i,i] &= (\mathcal{F}A\mathcal{F}^{-1})[i,i] 
    = \left(\mathcal{F} \left( \sum_j A[j,0] X^j\right)\mathcal{F}^{-1}\right) [i,i] \\
    &= \sum_j A[j,0] (\mathcal{F}X^j\mathcal{F}^{-1})[i,i] \\
    &= \sum_j A[j,0] \sum_{k}\sum_{\ell} \omega_N^{-i\ell} \delta_{\ell, k+j} \omega_N^{ki}/N \\
    &= \sum_j A[j,0] \sum_{k} \omega_N^{-i(k+j)}\omega_N^{ki}/N 
    = \sum_j \omega_N^{-ij}A[j,0] \\
    &= (\mathcal{F}A)[i,0] 
    = (\mathcal{F}a)[i],
\end{align*}

where in the last line $a$ is the first column of $A$.

\clearpage
\bibliography{references}

\begin{thebibliography}{100}

\bibitem{MKL}
Intel math kernel library.
\newblock
  \url{https://software.intel.com/content/www/us/en/develop/tools/math-kernel-library.html},
  Intel MKL.

\bibitem{Polly}
Llvm framework for high-level loop and data-locality optimizations.
\newblock \url{https://polly.llvm.org/}, LLVM MKL.

\bibitem{Pluto}
An automatic parallelizer and locality optimizer for affine loop nests.
\newblock \url{http://pluto-compiler.sourceforge.net/}, Pluto.

\bibitem{PoCC}
The polyhedral compiler collection.
\newblock \url{http://web.cs.ucla.edu/~pouchet/software/pocc/}, PoCC.

\bibitem{Stampede2}
The stampede2 supercomputing cluster.
\newblock \url{https://www.tacc.utexas.edu/systems/stampede2}, Stampede2.

\bibitem{acary2010implicit}
V.~Acary and B.~Brogliato.
\newblock Implicit euler numerical scheme and chattering-free implementation of
  sliding mode systems.
\newblock {\em Systems \& Control Letters}, 59(5):284--293, 2010.

\bibitem{andreussi2012revised}
O.~Andreussi, I.~Dabo, and N.~Marzari.
\newblock Revised self-consistent continuum solvation in electronic-structure
  calculations.
\newblock {\em The Journal of chemical physics}, 136(6):064102, 2012.

\bibitem{asgharzadeh2017newton}
H.~Asgharzadeh and I.~Borazjani.
\newblock A newton--krylov method with an approximate analytical jacobian for
  implicit solution of navier--stokes equations on staggered
  overset-curvilinear grids with immersed boundaries.
\newblock {\em Journal of computational physics}, 331:227--256, 2017.

\bibitem{ashrafizadeh2015jacobian}
A.~Ashrafizadeh, C.~Devaud, and N.~Aydemir.
\newblock A jacobian-free newton--krylov method for thermalhydraulics
  simulations.
\newblock {\em International Journal for Numerical Methods in Fluids},
  77(10):590--615, 2015.

\bibitem{atangana2015numerical}
A.~Atangana and J.~J. Nieto.
\newblock Numerical solution for the model of rlc circuit via the fractional
  derivative without singular kernel.
\newblock {\em Advances in Mechanical Engineering}, 7(10):1687814015613758,
  2015.

\bibitem{avissar1989parameterization}
R.~Avissar and R.~A. Pielke.
\newblock A parameterization of heterogeneous land surfaces for atmospheric
  numerical models and its impact on regional meteorology.
\newblock {\em Monthly Weather Review}, 117(10):2113--2136, 1989.

\bibitem{Bandishti2012}
V.~Bandishti, I.~Pananilath, and U.~Bondhugula.
\newblock Tiling stencil computations to maximize parallelism.
\newblock In {\em International Conference on High Performance Computing,
  Networking, Storage and Analysis}, pages 1--11, 2012.

\bibitem{barnes1964technique}
S.~L. Barnes.
\newblock A technique for maximizing details in numerical weather map analysis.
\newblock {\em Journal of Applied Meteorology and Climatology}, 3(4):396--409,
  1964.

\bibitem{barth2013high}
T.~J. Barth and H.~Deconinck.
\newblock {\em High-order methods for computational physics}, volume~9.
\newblock Springer Science \& Business Media, 2013.

\bibitem{beggs1992finite}
J.~H. Beggs, R.~J. Luebbers, K.~S. Yee, and K.~S. Kunz.
\newblock Finite-difference time-domain implementation of surface impedance
  boundary conditions.
\newblock {\em IEEE Transactions on Antennas and propagation}, 40(1):49--56,
  1992.

\bibitem{Bender2016}
M.~A. Bender, E.~D. Demaine, R.~Ebrahimi, J.~T. Fineman, R.~Johnson,
  A.~Lincoln, J.~Lynch, and S.~McCauley.
\newblock Cache-adaptive analysis.
\newblock In {\em ACM Symposium on Parallelism in Algorithms and
  Architectures}, pages 135--144, 2016.

\bibitem{Bender2014}
M.~A. Bender, R.~Ebrahimi, J.~T. Fineman, G.~Ghasemiesfeh, R.~Johnson, and
  S.~McCauley.
\newblock Cache-adaptive algorithms.
\newblock In {\em ACM-SIAM Symposium on Discrete Algorithms}, 2014.

\bibitem{benzi2004preconditioner}
M.~Benzi and G.~H. Golub.
\newblock A preconditioner for generalized saddle point problems.
\newblock {\em SIAM Journal on Matrix Analysis and Applications}, 26(1):20--41,
  2004.

\bibitem{benzi2011relaxed}
M.~Benzi, M.~Ng, Q.~Niu, and Z.~Wang.
\newblock A relaxed dimensional factorization preconditioner for the
  incompressible navier--stokes equations.
\newblock {\em Journal of Computational Physics}, 230(16):6185--6202, 2011.

\bibitem{bilbao2013modeling}
S.~Bilbao.
\newblock Modeling of complex geometries and boundary conditions in finite
  difference/finite volume time domain room acoustics simulation.
\newblock {\em IEEE Transactions on Audio, Speech, and Language Processing},
  21(7):1524--1533, 2013.

\bibitem{blazek2015computational}
J.~Blazek.
\newblock {\em Computational fluid dynamics: principles and applications}.
\newblock Butterworth-Heinemann, 2015.

\bibitem{blelloch2019optimal}
G.~E. Blelloch, J.~T. Fineman, Y.~Gu, and Y.~Sun.
\newblock Optimal parallel algorithms in the binary-forking model.
\newblock {\em arXiv preprint arXiv:1903.04650}, 2019.

\bibitem{bondhugula2016plutoplus}
U.~Bondhugula, A.~Acharya, and A.~Cohen.
\newblock The pluto+ algorithm: A practical approach for parallelization and
  locality optimization of affine loop nests.
\newblock {\em ACM Transactions on Programming Languages and Systems},
  38(3):1--32, 2016.

\bibitem{bondhugula2017}
U.~Bondhugula, V.~Bandishti, and I.~Pananilath.
\newblock Diamond tiling: Tiling techniques to maximize parallelism for stencil
  computations.
\newblock {\em IEEE Transactions on Parallel and Distributed Systems},
  28(5):1285--1298, 2017.

\bibitem{bondhugula2008pluto}
U.~Bondhugula, A.~Hartono, J.~Ramanujam, and P.~Sadayappan.
\newblock Pluto: A practical and fully automatic polyhedral program
  optimization system.
\newblock In {\em Proceedings of the ACM SIGPLAN 2008 Conference on Programming
  Language Design and Implementation (PLDI 08), Tucson, AZ (June 2008)}.
  Citeseer, 2008.

\bibitem{borrell2011parallel}
R.~Borrell, O.~Lehmkuhl, F.~X. Trias, and A.~Oliva.
\newblock Parallel direct poisson solver for discretisations with one fourier
  diagonalisable direction.
\newblock {\em Journal of computational physics}, 230(12):4723--4741, 2011.

\bibitem{boyd2001chebyshev}
J.~P. Boyd.
\newblock {\em Chebyshev and Fourier spectral methods}.
\newblock Courier Corporation, 2001.

\bibitem{bracewell1986fourier}
R.~N. Bracewell and R.~N. Bracewell.
\newblock {\em The Fourier transform and its applications}, volume 31999.
\newblock McGraw-Hill New York, 1986.

\bibitem{brown1994convergence}
P.~N. Brown and Y.~Saad.
\newblock Convergence theory of nonlinear newton--krylov algorithms.
\newblock {\em SIAM Journal on Optimization}, 4(2):297--330, 1994.

\bibitem{Bruun1978}
G.~Bruun.
\newblock z-transform dft filters and fft's.
\newblock {\em IEEE Transactions on Acoustics, Speech, and Signal Processing},
  26(1):56--63, 1978.

\bibitem{chan1993fft}
R.~H. Chan, J.~G. Nagy, and R.~J. Plemmons.
\newblock Fft-based preconditioners for toeplitz-block least squares problems.
\newblock {\em SIAM journal on numerical analysis}, 30(6):1740--1768, 1993.

\bibitem{chan1984stability}
T.~F. Chan.
\newblock Stability analysis of finite difference schemes for the
  advection-diffusion equation.
\newblock {\em SIAM journal on numerical analysis}, 21(2):272--284, 1984.

\bibitem{chan1988optimal}
T.~F. Chan.
\newblock An optimal circulant preconditioner for toeplitz systems.
\newblock {\em SIAM journal on scientific and statistical computing},
  9(4):766--771, 1988.

\bibitem{chen2005time}
B.-F. Chen and R.~Nokes.
\newblock Time-independent finite difference analysis of fully non-linear and
  viscous fluid sloshing in a rectangular tank.
\newblock {\em Journal of Computational Physics}, 209(1):47--81, 2005.

\bibitem{chizhonkov2000domain}
E.~V. Chizhonkov and M.~A. Olshanskii.
\newblock On the domain geometry dependence of the lbb condition.
\newblock {\em ESAIM: Mathematical Modelling and Numerical Analysis},
  34(5):935--951, 2000.

\bibitem{choi1986finite}
D.~H. Choi and W.~J. Hoefer.
\newblock The finite-difference-time-domain method and its application to
  eigenvalue problems.
\newblock {\em IEEE Transactions on Microwave Theory and Techniques},
  34(12):1464--1470, 1986.

\bibitem{christen2011patus}
M.~Christen, O.~Schenk, and H.~Burkhart.
\newblock Patus: A code generation and autotuning framework for parallel
  iterative stencil computations on modern microarchitectures.
\newblock In {\em 2011 IEEE International Parallel \& Distributed Processing
  Symposium}, pages 676--687. IEEE, 2011.

\bibitem{CooleyTu65}
J.~W. Cooley and J.~W. Tukey.
\newblock An algorithm for the machine calculation of complex {F}ourier series.
\newblock {\em Mathematics of Computation}, 19(90):297--301, 1965.

\bibitem{CormenLeRiSt2009}
T.~H. Cormen, C.~E. Leiserson, R.~L. Rivest, and C.~Stein.
\newblock {\em Introduction to algorithms}.
\newblock MIT press, 2009.

\bibitem{crank1947practical}
J.~Crank and P.~Nicolson.
\newblock A practical method for numerical evaluation of solutions of partial
  differential equations of the heat-conduction type.
\newblock In {\em Mathematical Proceedings of the Cambridge Philosophical
  Society}, volume~43, pages 50--67. Cambridge University Press, 1947.

\bibitem{dahlquist2008numerical}
G.~Dahlquist and {\AA}.~Bj{\"o}rck.
\newblock {\em Numerical methods in scientific computing, volume I}.
\newblock SIAM, 2008.

\bibitem{datta2009optimization}
K.~Datta, S.~Kamil, S.~Williams, L.~Oliker, J.~Shalf, and K.~Yelick.
\newblock Optimization and performance modeling of stencil computations on
  modern microprocessors.
\newblock {\em SIAM review}, 51(1):129--159, 2009.

\bibitem{dietrich1997fast}
C.~R. Dietrich and G.~N. Newsam.
\newblock Fast and exact simulation of stationary gaussian processes through
  circulant embedding of the covariance matrix.
\newblock {\em SIAM Journal on Scientific Computing}, 18(4):1088--1107, 1997.

\bibitem{erlangga2006novel}
Y.~A. Erlangga, C.~W. Oosterlee, and C.~Vuik.
\newblock A novel multigrid based preconditioner for heterogeneous helmholtz
  problems.
\newblock {\em SIAM Journal on Scientific Computing}, 27(4):1471--1492, 2006.

\bibitem{ferziger2002computational}
J.~H. Ferziger, M.~Peri{\'c}, and R.~L. Street.
\newblock {\em Computational methods for fluid dynamics}, volume~3.
\newblock Springer, 2002.

\bibitem{fokas2000integrability}
A.~Fokas.
\newblock On the integrability of linear and nonlinear partial differential
  equations.
\newblock {\em Journal of Mathematical Physics}, 41(6):4188--4237, 2000.

\bibitem{frigo2005design}
M.~Frigo and S.~G. Johnson.
\newblock The design and implementation of fftw3.
\newblock {\em Proceedings of the IEEE}, 93(2):216--231, 2005.

\bibitem{FrigoLePrRa1999}
M.~Frigo, C.~E. Leiserson, H.~Prokop, and S.~Ramachandran.
\newblock Cache-oblivious algorithms.
\newblock In {\em Foundations of Computer Science}, pages 285--297, 1999.

\bibitem{frigo2005cache}
M.~Frigo and V.~Strumpen.
\newblock Cache oblivious stencil computations.
\newblock In {\em Proceedings of the 19th annual international conference on
  Supercomputing}, pages 361--366, 2005.

\bibitem{FrigoSt2005}
M.~Frigo and V.~Strumpen.
\newblock Cache oblivious stencil computations.
\newblock In {\em International conference on Supercomputing}, pages 361--366,
  2005.

\bibitem{FrigoSt2009}
M.~Frigo and V.~Strumpen.
\newblock The cache complexity of multithreaded cache oblivious algorithms.
\newblock {\em Theory of Computing Systems}, 45(2):203--233, 2009.

\bibitem{fritz2009application}
J.~Fritz, I.~Neuweiler, and W.~Nowak.
\newblock Application of fft-based algorithms for large-scale universal kriging
  problems.
\newblock {\em Mathematical Geosciences}, 41(5):509--533, 2009.

\bibitem{frommer2017block}
A.~Frommer, K.~Lund, D.~B. Szyld, et~al.
\newblock Block krylov subspace methods for functions of matrices.
\newblock 2017.

\bibitem{fuka2015poisfft}
V.~Fuka.
\newblock Poisfft--a free parallel fast poisson solver.
\newblock {\em Applied Mathematics and Computation}, 267:356--364, 2015.

\bibitem{gammie2003harm}
C.~F. Gammie, J.~C. McKinney, and G.~T{\'o}th.
\newblock Harm: a numerical scheme for general relativistic
  magnetohydrodynamics.
\newblock {\em The Astrophysical Journal}, 589(1):444, 2003.

\bibitem{garrappa2015solving}
R.~Garrappa, I.~Moret, and M.~Popolizio.
\newblock Solving the time-fractional schr{\"o}dinger equation by krylov
  projection methods.
\newblock {\em Journal of Computational Physics}, 293:115--134, 2015.

\bibitem{gmeiner2013optimization}
B.~Gmeiner, T.~Gradl, F.~Gaspar, and U.~R{\"u}de.
\newblock Optimization of the multigrid-convergence rate on semi-structured
  meshes by local fourier analysis.
\newblock {\em Computers \& Mathematics with Applications}, 65(4):694--711,
  2013.

\bibitem{gohberg1994fast}
I.~Gohberg and V.~Olshevsky.
\newblock Fast algorithms with preprocessing for matrix-vector multiplication
  problems.
\newblock {\em Journal of Complexity}, 10(4):411--427, 1994.

\bibitem{golub1992scientific}
G.~H. Golub, J.~M. Ortega, et~al.
\newblock {\em Scientific computing and differential equations: an introduction
  to numerical methods}.
\newblock Academic press, 1992.

\bibitem{Good1958}
I.~J. Good.
\newblock The interaction algorithm and practical fourier analysis.
\newblock {\em Journal of the Royal Statistical Society: Series B
  (Methodological)}, 20(2):361--372, 1958.

\bibitem{gray2006toeplitz}
R.~M. Gray.
\newblock {\em Toeplitz and circulant matrices: A review}.
\newblock now publishers inc, 2006.

\bibitem{guan2019two}
Y.~Guan and I.~Novosselov.
\newblock Two relaxation time lattice boltzmann method coupled to fast fourier
  transform poisson solver: Application to electroconvective flow.
\newblock {\em Journal of computational physics}, 397:108830, 2019.

\bibitem{gutknecht2007brief}
M.~H. Gutknecht.
\newblock A brief introduction to krylov space methods for solving linear
  systems.
\newblock In {\em Frontiers of Computational Science}, pages 53--62. Springer,
  2007.

\bibitem{guttel2014nleigs}
S.~Guttel, R.~Van~Beeumen, K.~Meerbergen, and W.~Michiels.
\newblock Nleigs: A class of fully rational krylov methods for nonlinear
  eigenvalue problems.
\newblock {\em SIAM Journal on Scientific Computing}, 36(6):A2842--A2864, 2014.

\bibitem{hagedorn2018high}
B.~Hagedorn, L.~Stoltzfus, M.~Steuwer, S.~Gorlatch, and C.~Dubach.
\newblock High performance stencil code generation with lift.
\newblock In {\em Proceedings of the 2018 International Symposium on Code
  Generation and Optimization}, pages 100--112, 2018.

\bibitem{he2011robust}
W.~He and S.~S. Ge.
\newblock Robust adaptive boundary control of a vibrating string under unknown
  time-varying disturbance.
\newblock {\em IEEE Transactions on Control Systems Technology}, 20(1):48--58,
  2011.

\bibitem{henretty2013stencil}
T.~Henretty, R.~Veras, F.~Franchetti, L.-N. Pouchet, J.~Ramanujam, and
  P.~Sadayappan.
\newblock A stencil compiler for short-vector simd architectures.
\newblock In {\em Proceedings of the 27th international ACM conference on
  International conference on supercomputing}, pages 13--24, 2013.

\bibitem{hirsch2007numerical}
C.~Hirsch.
\newblock {\em Numerical computation of internal and external flows: The
  fundamentals of computational fluid dynamics}.
\newblock Elsevier, 2007.

\bibitem{hockney1965fast}
R.~W. Hockney.
\newblock A fast direct solution of poisson's equation using fourier analysis.
\newblock {\em Journal of the ACM (JACM)}, 12(1):95--113, 1965.

\bibitem{ipsen1998idea}
I.~C. Ipsen and C.~D. Meyer.
\newblock The idea behind krylov methods.
\newblock {\em The American mathematical monthly}, 105(10):889--899, 1998.

\bibitem{isaacson2012analysis}
E.~Isaacson and H.~B. Keller.
\newblock {\em Analysis of numerical methods}.
\newblock Courier Corporation, 2012.

\bibitem{janpugdee2006accelerated}
P.~Janpugdee, P.~H. Pathak, P.~Mahachoklertwattana, and R.~J. Burkholder.
\newblock An accelerated dft-mom for the analysis of large finite periodic
  antenna arrays.
\newblock {\em IEEE transactions on antennas and propagation}, 54(1):279--283,
  2006.

\bibitem{johnston2002finite}
H.~Johnston and J.-G. Liu.
\newblock Finite difference schemes for incompressible flow based on local
  pressure boundary conditions.
\newblock {\em Journal of Computational Physics}, 180(1):120--154, 2002.

\bibitem{kabel2014efficient}
M.~Kabel, T.~B{\"o}hlke, and M.~Schneider.
\newblock Efficient fixed point and newton--krylov solvers for fft-based
  homogenization of elasticity at large deformations.
\newblock {\em Computational Mechanics}, 54(6):1497--1514, 2014.

\bibitem{kalnay1990global}
E.~Kalnay, M.~Kanamitsu, and W.~Baker.
\newblock Global numerical weather prediction at the national meteorological
  center.
\newblock {\em Bulletin of the American Meteorological Society},
  71(10):1410--1428, 1990.

\bibitem{kamil2010auto}
S.~Kamil, C.~Chan, L.~Oliker, J.~Shalf, and S.~Williams.
\newblock An auto-tuning framework for parallel multicore stencil computations.
\newblock In {\em 2010 IEEE international symposium on parallel \& distributed
  processing (IPDPS)}, pages 1--12. IEEE, 2010.

\bibitem{kamil2005impact}
S.~Kamil, P.~Husbands, L.~Oliker, J.~Shalf, and K.~Yelick.
\newblock Impact of modern memory subsystems on cache optimizations for stencil
  computations.
\newblock In {\em Proceedings of the 2005 workshop on Memory system
  performance}, pages 36--43, 2005.

\bibitem{Khokhriakov2018}
S.~Khokhriakov, R.~R. Manumachu, and A.~Lastovetsky.
\newblock Performance optimization of multithreaded 2d {FFT} on multicore
  processors: Challenges and solution approaches.
\newblock In {\em 2018 {IEEE} 25th International Conference on High Performance
  Computing Workshops ({HiPCW})}. {IEEE}, Dec. 2018.

\bibitem{kirby2018solver}
R.~C. Kirby and L.~Mitchell.
\newblock Solver composition across the pde/linear algebra barrier.
\newblock {\em SIAM Journal on Scientific Computing}, 40(1):C76--C98, 2018.

\bibitem{kloeden1992higher}
P.~E. Kloeden and E.~Platen.
\newblock Higher-order implicit strong numerical schemes for stochastic
  differential equations.
\newblock {\em Journal of statistical physics}, 66(1):283--314, 1992.

\bibitem{knyazev2001toward}
A.~V. Knyazev.
\newblock Toward the optimal preconditioned eigensolver: Locally optimal block
  preconditioned conjugate gradient method.
\newblock {\em SIAM journal on scientific computing}, 23(2):517--541, 2001.

\bibitem{komissarov2002time}
S.~Komissarov.
\newblock Time-dependent, force-free, degenerate electrodynamics.
\newblock {\em Monthly Notices of the Royal Astronomical Society},
  336(3):759--766, 2002.

\bibitem{kong2013polyhedral}
M.~Kong, R.~Veras, K.~Stock, F.~Franchetti, L.-N. Pouchet, and P.~Sadayappan.
\newblock When polyhedral transformations meet simd code generation.
\newblock In {\em Proceedings of the 34th ACM SIGPLAN conference on Programming
  language design and implementation}, pages 127--138, 2013.

\bibitem{konuk2020modeling}
T.~Konuk and J.~Shragge.
\newblock Modeling full-wavefield time-varying sea-surface effects on seismic
  data: A mimetic finite-difference approach.
\newblock {\em Geophysics}, 85(2):T45--T55, 2020.

\bibitem{korch2020depth}
M.~Korch and T.~Werner.
\newblock An in-depth introduction of multi-workgroup tiling for improving the
  locality of explicit one-step methods for ode systems with limited access
  distance on gpus.
\newblock {\em Concurrency and Computation: Practice and Experience}, page
  e6016, 2020.

\bibitem{koutsourelakis2009accurate}
P.-S. Koutsourelakis.
\newblock Accurate uncertainty quantification using inaccurate computational
  models.
\newblock {\em SIAM Journal on Scientific Computing}, 31(5):3274--3300, 2009.

\bibitem{kuijlaars2006convergence}
A.~B. Kuijlaars.
\newblock Convergence analysis of krylov subspace iterations with methods from
  potential theory.
\newblock {\em SIAM review}, 48(1):3--40, 2006.

\bibitem{kuzmin2010vertex}
D.~Kuzmin.
\newblock A vertex-based hierarchical slope limiter for p-adaptive
  discontinuous galerkin methods.
\newblock {\em Journal of computational and applied mathematics},
  233(12):3077--3085, 2010.

\bibitem{le2010spectral}
O.~Le~Ma{\^\i}tre and O.~M. Knio.
\newblock {\em Spectral methods for uncertainty quantification: with
  applications to computational fluid dynamics}.
\newblock Springer Science \& Business Media, 2010.

\bibitem{liao2013high}
W.~Liao.
\newblock A high-order adi finite difference scheme for a 3d reaction-diffusion
  equation with neumann boundary condition.
\newblock {\em Numerical Methods for Partial Differential Equations},
  29(3):778--798, 2013.

\bibitem{louboutin2019devito}
M.~Louboutin, M.~Lange, F.~Luporini, N.~Kukreja, P.~A. Witte, F.~J. Herrmann,
  P.~Velesko, and G.~J. Gorman.
\newblock Devito (v3. 1.0): an embedded domain-specific language for finite
  differences and geophysical exploration.
\newblock {\em Geoscientific Model Development}, 12(3):1165--1187, 2019.

\bibitem{luporini2020architecture}
F.~Luporini, M.~Louboutin, M.~Lange, N.~Kukreja, P.~Witte, J.~H{\"u}ckelheim,
  C.~Yount, P.~H. Kelly, F.~J. Herrmann, and G.~J. Gorman.
\newblock Architecture and performance of devito, a system for automated
  stencil computation.
\newblock {\em ACM Transactions on Mathematical Software (TOMS)}, 46(1):1--28,
  2020.

\bibitem{mang2015inexact}
A.~Mang and G.~Biros.
\newblock An inexact newton--krylov algorithm for constrained diffeomorphic
  image registration.
\newblock {\em SIAM journal on imaging sciences}, 8(2):1030--1069, 2015.

\bibitem{mendicino2015stability}
G.~Mendicino, J.~Pedace, and A.~Senatore.
\newblock Stability of an overland flow scheme in the framework of a fully
  coupled eco-hydrological model based on the macroscopic cellular automata
  approach.
\newblock {\em Communications in Nonlinear Science and Numerical Simulation},
  21(1-3):128--146, 2015.

\bibitem{moaddy2011non}
K.~Moaddy, S.~Momani, and I.~Hashim.
\newblock The non-standard finite difference scheme for linear fractional pdes
  in fluid mechanics.
\newblock {\em Computers \& Mathematics with Applications}, 61(4):1209--1216,
  2011.

\bibitem{moretti1979lambda}
G.~Moretti.
\newblock The $\lambda$-scheme.
\newblock {\em Computers \& Fluids}, 7(3):191--205, 1979.

\bibitem{mugler1988fast}
D.~Mugler and R.~Scott.
\newblock Fast fourier transform method for partial differential equations,
  case study: The 2-d diffusion equation.
\newblock {\em Computers \& Mathematics with Applications}, 16(3):221--228,
  1988.

\bibitem{mur1981absorbing}
G.~Mur.
\newblock Absorbing boundary conditions for the finite-difference approximation
  of the time-domain electromagnetic-field equations.
\newblock {\em IEEE transactions on Electromagnetic Compatibility},
  (4):377--382, 1981.

\bibitem{musco2015randomized}
C.~Musco and C.~Musco.
\newblock Randomized block krylov methods for stronger and faster approximate
  singular value decomposition.
\newblock {\em Advances in Neural Information Processing Systems},
  28:1396--1404, 2015.

\bibitem{najm1998semi}
H.~N. Najm, P.~S. Wyckoff, and O.~M. Knio.
\newblock A semi-implicit numerical scheme for reacting flow: I. stiff
  chemistry.
\newblock {\em Journal of Computational Physics}, 143(2):381--402, 1998.

\bibitem{nguyen20103}
A.~Nguyen, N.~Satish, J.~Chhugani, C.~Kim, and P.~Dubey.
\newblock 3.5-d blocking optimization for stencil computations on modern cpus
  and gpus.
\newblock In {\em SC'10: Proceedings of the 2010 ACM/IEEE International
  Conference for High Performance Computing, Networking, Storage and Analysis},
  pages 1--13. IEEE, 2010.

\bibitem{nkwunonwo2019urban}
U.~Nkwunonwo, M.~Whitworth, and B.~Baily.
\newblock Urban flood modelling combining cellular automata framework with
  semi-implicit finite difference numerical formulation.
\newblock {\em Journal of African Earth Sciences}, 150:272--281, 2019.

\bibitem{notay2008recursive}
Y.~Notay and P.~S. Vassilevski.
\newblock Recursive krylov-based multigrid cycles.
\newblock {\em Numerical Linear Algebra with Applications}, 15(5):473--487,
  2008.

\bibitem{ostashev2005equations}
V.~E. Ostashev, D.~K. Wilson, L.~Liu, D.~F. Aldridge, N.~P. Symons, and
  D.~Marlin.
\newblock Equations for finite-difference, time-domain simulation of sound
  propagation in moving inhomogeneous media and numerical implementation.
\newblock {\em The Journal of the Acoustical Society of America},
  117(2):503--517, 2005.

\bibitem{pang1999introduction}
T.~Pang.
\newblock An introduction to computational physics, 1999.

\bibitem{pereda2001analyzing}
A.~Pereda, L.~A. Vielva, A.~Vegas, and A.~Prieto.
\newblock Analyzing the stability of the fdtd technique by combining the von
  neumann method with the routh-hurwitz criterion.
\newblock {\em IEEE Transactions on Microwave Theory and Techniques},
  49(2):377--381, 2001.

\bibitem{peyre2011numerical}
G.~Peyr{\'e}.
\newblock The numerical tours of signal processing-advanced computational
  signal and image processing.
\newblock {\em IEEE Computing in Science and Engineering}, 13(4):94--97, 2011.

\bibitem{press2007numerical}
W.~H. Press, S.~A. Teukolsky, W.~T. Vetterling, and B.~P. Flannery.
\newblock {\em Numerical recipes 3rd edition: The art of scientific computing}.
\newblock Cambridge university press, 2007.

\bibitem{ragan2013halide}
J.~Ragan-Kelley, C.~Barnes, A.~Adams, S.~Paris, F.~Durand, and S.~Amarasinghe.
\newblock Halide: a language and compiler for optimizing parallelism, locality,
  and recomputation in image processing pipelines.
\newblock {\em Acm Sigplan Notices}, 48(6):519--530, 2013.

\bibitem{ramani1989painleve}
A.~Ramani, B.~Grammaticos, and T.~Bountis.
\newblock The painlev{\'e} property and singularity analysis of integrable and
  non-integrable systems.
\newblock {\em Physics Reports}, 180(3):159--245, 1989.

\bibitem{rappaz2010numerical}
M.~Rappaz, M.~Bellet, and M.~Deville.
\newblock {\em Numerical modeling in materials science and engineering},
  volume~32.
\newblock Springer Science \& Business Media, 2010.

\bibitem{Rawat2018}
P.~S. Rawat, M.~Vaidya, A.~Sukumaran-Rajam, M.~Ravishankar, V.~Grover,
  A.~Rountev, L.-N. Pouchet, and P.~Sadayappan.
\newblock Domain-specific optimization and generation of high-performance gpu
  code for stencil computations.
\newblock {\em Proceedings of the IEEE}, 106(11):1902--1920, 2018.

\bibitem{renson2016numerical}
L.~Renson, G.~Kerschen, and B.~Cochelin.
\newblock Numerical computation of nonlinear normal modes in mechanical
  engineering.
\newblock {\em Journal of Sound and Vibration}, 364:177--206, 2016.

\bibitem{robert1981stable}
A.~Robert.
\newblock A stable numerical integration scheme for the primitive
  meteorological equations.
\newblock {\em Atmosphere-Ocean}, 19(1):35--46, 1981.

\bibitem{robert1982semi}
A.~Robert.
\newblock A semi-lagrangian and semi-implicit numerical integration scheme for
  the primitive meteorological equations.
\newblock {\em Journal of the Meteorological Society of Japan. Ser. II},
  60(1):319--325, 1982.

\bibitem{rota1973foundations}
G.-C. Rota, D.~Kahaner, and A.~Odlyzko.
\newblock On the foundations of combinatorial theory. viii. finite operator
  calculus.
\newblock {\em Journal of Mathematical Analysis and Applications},
  42(3):684--760, 1973.

\bibitem{russ1994image}
J.~C. Russ, J.~R. Matey, A.~J. Mallinckrodt, and S.~McKay.
\newblock The image processing handbook.
\newblock {\em Computers in Physics}, 8(2):177--178, 1994.

\bibitem{saad1989krylov}
Y.~Saad.
\newblock Krylov subspace methods on supercomputers.
\newblock {\em SIAM Journal on Scientific and Statistical Computing},
  10(6):1200--1232, 1989.

\bibitem{saleheen1997new}
H.~I. Saleheen and K.~T. Ng.
\newblock New finite difference formulations for general inhomogeneous
  anisotropic bioelectric problems.
\newblock {\em IEEE transactions on biomedical engineering}, 44(9):800--809,
  1997.

\bibitem{sano2011scalable}
K.~Sano, Y.~Hatsuda, and S.~Yamamoto.
\newblock Scalable streaming-array of simple soft-processors for stencil
  computations with constant memory-bandwidth.
\newblock In {\em 2011 IEEE 19th Annual International Symposium on
  Field-Programmable Custom Computing Machines}, pages 234--241. IEEE, 2011.

\bibitem{schumann1988fast}
U.~Schumann and R.~A. Sweet.
\newblock Fast fourier transforms for direct solution of poisson's equation
  with staggered boundary conditions.
\newblock {\em Journal of Computational Physics}, 75(1):123--137, 1988.

\bibitem{sitko2016time}
M.~Sitko, M.~Pietrzyk, and L.~Madej.
\newblock Time and length scale issues in numerical modelling of dynamic
  recrystallization based on the multi space cellular automata method.
\newblock {\em Journal of computational science}, 16:98--113, 2016.

\bibitem{somers1993direct}
J.~Somers.
\newblock Direct simulation of fluid flow with cellular automata and the
  lattice-boltzmann equation.
\newblock {\em Applied Scientific Research}, 51(1-2):127--133, 1993.

\bibitem{storti2013fft}
M.~A. Storti, R.~R. Paz, L.~D. Dalcin, S.~D. Costarelli, and S.~R. Idelsohn.
\newblock A fft preconditioning technique for the solution of incompressible
  flow on gpus.
\newblock {\em Computers \& Fluids}, 74:44--57, 2013.

\bibitem{strikwerda2004finite}
J.~C. Strikwerda.
\newblock {\em Finite difference schemes and partial differential equations}.
\newblock SIAM, 2004.

\bibitem{sweby1984high}
P.~K. Sweby.
\newblock High resolution schemes using flux limiters for hyperbolic
  conservation laws.
\newblock {\em SIAM journal on numerical analysis}, 21(5):995--1011, 1984.

\bibitem{szilard2004theories}
R.~Szilard.
\newblock Theories and applications of plate analysis: classical, numerical and
  engineering methods.
\newblock {\em Appl. Mech. Rev.}, 57(6):B32--B33, 2004.

\bibitem{taflove2005computational}
A.~Taflove and S.~C. Hagness.
\newblock {\em Computational electrodynamics: the finite-difference time-domain
  method}.
\newblock Artech house, 2005.

\bibitem{tam1994wall}
C.~K. Tam and Z.~Dong.
\newblock Wall boundary conditions for high-order finite-difference schemes in
  computational aeroacoustics.
\newblock {\em Theoretical and Computational Fluid Dynamics}, 6(5):303--322,
  1994.

\bibitem{tang2011pochoir}
Y.~Tang, R.~A. Chowdhury, B.~C. Kuszmaul, C.-K. Luk, and C.~E. Leiserson.
\newblock The pochoir stencil compiler.
\newblock In {\em Proceedings of the twenty-third annual ACM symposium on
  Parallelism in algorithms and architectures}, pages 117--128, 2011.

\bibitem{Tang2011}
Y.~Tang, R.~A. Chowdhury, B.~C. Kuszmaul, C.-K. Luk, and C.~E. Leiserson.
\newblock The pochoir stencil compiler.
\newblock In {\em ACM Symposium on Parallelism in Algorithms and
  Architectures}, pages 117--128, 2011.

\bibitem{teixeira2008time}
F.~L. Teixeira.
\newblock Time-domain finite-difference and finite-element methods for maxwell
  equations in complex media.
\newblock {\em IEEE Transactions on Antennas and Propagation},
  56(8):2150--2166, 2008.

\bibitem{teixeira1998finite}
F.~L. Teixeira, W.~C. Chew, M.~Straka, M.~Oristaglio, and T.~Wang.
\newblock Finite-difference time-domain simulation of ground penetrating radar
  on dispersive, inhomogeneous, and conductive soils.
\newblock {\em IEEE Transactions on Geoscience and remote sensing},
  36(6):1928--1937, 1998.

\bibitem{thijssen2007computational}
J.~Thijssen.
\newblock {\em Computational physics}.
\newblock Cambridge university press, 2007.

\bibitem{turkel1987preconditioned}
E.~Turkel.
\newblock Preconditioned methods for solving the incompressible and low speed
  compressible equations.
\newblock {\em Journal of computational physics}, 72(2):277--298, 1987.

\bibitem{tyrtyshnikov1996unifying}
E.~E. Tyrtyshnikov.
\newblock A unifying approach to some old and new theorems on distribution and
  clustering.
\newblock {\em Linear algebra and its applications}, 232:1--43, 1996.

\bibitem{van2003iterative}
H.~A. Van~der Vorst.
\newblock {\em Iterative Krylov methods for large linear systems}.
\newblock Number~13. Cambridge University Press, 2003.

\bibitem{van1984enhancements}
J.~P. Van~Doormaal and G.~D. Raithby.
\newblock Enhancements of the simple method for predicting incompressible fluid
  flows.
\newblock {\em Numerical heat transfer}, 7(2):147--163, 1984.

\bibitem{van1979towards}
B.~Van~Leer.
\newblock Towards the ultimate conservative difference scheme. v. a
  second-order sequel to godunov's method.
\newblock {\em Journal of computational Physics}, 32(1):101--136, 1979.

\bibitem{van2012numerical}
U.~Van~Rienen.
\newblock {\em Numerical methods in computational electrodynamics: linear
  systems in practical applications}, volume~12.
\newblock Springer Science \& Business Media, 2012.

\bibitem{verdoolaege2013}
S.~Verdoolaege, J.~Carlos~Juega, A.~Cohen, J.~Ignacio~Gomez, C.~Tenllado, and
  F.~Catthoor.
\newblock Polyhedral parallel code generation for cuda.
\newblock {\em ACM Transactions on Architecture and Code Optimization},
  9(4):1--23, 2013.

\bibitem{vese2002numerical}
L.~A. Vese and S.~J. Osher.
\newblock Numerical methods for p-harmonic flows and applications to image
  processing.
\newblock {\em SIAM Journal on Numerical Analysis}, 40(6):2085--2104, 2002.

\bibitem{vesely1994computational}
F.~J. Vesely.
\newblock {\em Computational Physics}.
\newblock Springer, 1994.

\bibitem{virieux1986p}
J.~Virieux.
\newblock P-sv wave propagation in heterogeneous media: Velocity-stress
  finite-difference method.
\newblock {\em Geophysics}, 51(4):889--901, 1986.

\bibitem{wang1993finite}
T.~Wang and G.~W. Hohmann.
\newblock A finite-difference, time-domain solution for three-dimensional
  electromagnetic modeling.
\newblock {\em Geophysics}, 58(6):797--809, 1993.

\bibitem{weickert1996theoretical}
J.~Weickert.
\newblock Theoretical foundations of anisotropic diffusion in image processing.
\newblock In {\em Theoretical foundations of computer vision}, pages 221--236.
  Springer, 1996.

\bibitem{weickert2000applications}
J.~Weickert.
\newblock Applications of nonlinear diffusion in image processing and computer
  vision.
\newblock 2000.

\bibitem{wesseling1996neumann}
P.~Wesseling.
\newblock von neumann stability conditions for the convection-diffusion
  eqation.
\newblock {\em IMA journal of Numerical Analysis}, 16(4):583--598, 1996.

\bibitem{Winograd1978}
S.~Winograd.
\newblock On computing the discrete fourier transform.
\newblock {\em Mathematics of computation}, 32(141):175--199, 1978.

\bibitem{Wolf1991}
M.~E. Wolf and M.~S. Lam.
\newblock A data locality optimizing algorithm.
\newblock In {\em ACM SIGPLAN conference on Programming language design and
  implementation}, pages 30--44, 1991.

\bibitem{Wolf1996}
M.~E. Wolf, D.~E. Maydan, and D.-K. Chen.
\newblock Combining loop transformations considering caches and scheduling.
\newblock In {\em IEEE/ACM International Symposium on Microarchitecture}, pages
  274--286, 1996.

\bibitem{Wolfe1987}
M.~J. Wolfe.
\newblock Iteration space tiling for memory hierarchies.
\newblock {\em Parallel Processing for Scientific Computing}, 357:361, 1987.

\bibitem{Wonnacott2002}
D.~Wonnacott.
\newblock Achieving scalable locality with time skewing.
\newblock {\em International Journal of Parallel Programming}, 30(3):181--221,
  2002.

\bibitem{xu2006simulations}
A.~Xu, G.~Gonnella, and A.~Lamura.
\newblock Simulations of complex fluids by mixed lattice boltzmann—finite
  difference methods.
\newblock {\em Physica A: Statistical Mechanics and its Applications},
  362(1):42--47, 2006.

\bibitem{yuki2015polybench}
T.~Yuki and L.-N. Pouchet.
\newblock Polybench 4.0, 2015.

\bibitem{yuste2005explicit}
S.~B. Yuste and L.~Acedo.
\newblock An explicit finite difference method and a new von neumann-type
  stability analysis for fractional diffusion equations.
\newblock {\em SIAM Journal on Numerical Analysis}, 42(5):1862--1874, 2005.

\end{thebibliography}

\end{document}